\documentclass[runningheads,envcountsame]{llncs}

\pdfoutput=1 
\usepackage[T1]{fontenc}
\usepackage{amssymb}
\usepackage{mathtools}
\usepackage{xfrac}
\usepackage{tikz}
\usetikzlibrary{arrows, automata}
\usepackage{hyperref}
\usepackage{cleveref}
\usepackage{thm-restate}
\usepackage{subcaption}
\usepackage{bbding}

\newcommand{\prob}{P}
\newcommand{\probMatrix}{\mathbf{P}}
\newcommand{\symmetryMatrix}{\mathbf{S}}
\newcommand{\jordanMatrix}{\mathbf{J}}
\newcommand{\diagonalMatrix}{\mathbf{D}}
\newcommand{\generatorMatrix}{\mathbf{Q}}
\newcommand{\unifProb}{\overline{P}}
\newcommand{\unifProbMatrix}{\overline{\mathbf{P}}}
\newcommand{\erlang}{\mathcal{E}}
\newcommand{\probMeasure}{\mathrm{Pr}}

\newcommand{\rate}{E}
\newcommand{\Succ}{\mathrm{Succ}}
\newcommand{\Distr}{\mathrm{Distr}}
\newcommand{\supp}{\mathrm{supp}}
\newcommand{\Paths}{\mathrm{Paths}}
\newcommand{\labelFunction}{L}
\newcommand{\stateSpace}{S}
\newcommand{\initialState}{s_{init}}
\newcommand{\uniformizationRate}{q}
\newcommand{\uniformization}[2]{\mathrm{unif}(#1, #2)}
\newcommand{\probability}{\mathrm{Prob}}
\newcommand{\Diff}{\mathrm{Diff}}
\newcommand{\reward}{\rho}
\newcommand{\dirac}{1}
\newcommand{\xfrac}[2]{%
	\mbox{\raisebox{0.3ex}{\ensuremath{\displaystyle #1}\hspace{-0.2ex}}%
		/%
		\raisebox{-0.3ex}{\footnotesize{\ensuremath{#2}}}%
	}%
}

\newcommand{\reachability}{\lozenge}
\author{
	Timm Spork\inst{1}\Envelope\orcidID{0009-0008-4461-0667} \and
	Christel Baier\inst{1}\orcidID{0000-0002-5321-9343} \and
	Joost-Pieter Katoen \inst{2}\orcidID{0000-0002-6143-1926} \and
	Sascha Klüppelholz\inst{1}\orcidID{0000-0003-1724-2586} \and
	Jakob Piribauer\inst{1,3}\orcidID{0000-0003-4829-0476}
}

\authorrunning{T. Spork et al.}

\institute{
	\textsuperscript{1}Technische Universität Dresden, Dresden, Germany\\
	\email{firstname.lastname@tu-dresden.de} \\
	\textsuperscript{2}RWTH Aachen University, Aachen, Germany \\
	\email{katoen@cs.rwth-aachen.de} \\
	\textsuperscript{3}Universität Leipzig, Leipzig, Germany 
}

\title{Approximate Probabilistic Bisimulation for Continuous-Time Markov Chains}
\allowdisplaybreaks
\begin{document}
	\maketitle
	
	\begin{abstract}
		We introduce $(\varepsilon, \delta)$-bisimulation, a novel type of approximate probabilistic bisimulation for continuous-time Markov chains. In contrast to related notions, $(\varepsilon, \delta)$-bisimulation allows the use of different tolerances for the transition probabilities ($\varepsilon$, additive) and total exit rates ($\delta$, multiplicative) of states. Fundamental properties of the notion, as well as bounds on the absolute difference of time- and reward-bounded reachability probabilities for $(\varepsilon,\delta)$-bisimilar states, are established.
		
	\end{abstract}

	
	\section{Introduction}
	Continuous-time Markov chains (CTMCs) are a prominent probabilistic
	model in various application fields, e.g., reliability engineering,
	systems biology, modeling of chemical reactions, and performance evaluation. CTMCs are state-based
	models whose transitions yield a discrete probability distribution over
	states---as in discrete-time Markov chains (DTMCs)---while the state residence
	times are governed by exponential distributions. Various model-checking
	approaches for CTMCs exist
	\cite{DBLP:journals/pe/AmparoreD18,DBLP:journals/tse/BaierCHKS07,DBLP:conf/lics/ChenHKM09,DBLP:journals/corr/abs-1104-4983}
	and are supported by tools such as PRISM \cite{PRISM} and Storm \cite{PMCS}. CTMC model checking is used
	to analyze, e.g., stochastic Petri nets
	\cite{DBLP:conf/apn/AmparoreBD14}, fault trees
	\cite{DBLP:conf/safecomp/ArnoldBBGS13,DBLP:journals/tii/VolkJK18},
	biological systems \cite{DBLP:series/natosec/KwiatkowskaT14,DBLP:journals/jetc/MadsenZRWM14}, and chemical reactions \cite{AAMCQACRN,RASBS}.
	
	The central issue in CTMC model checking is computing timed reachability
	probabilities: what is the probability to reach a set of goal states
	within a given deadline from a given start state? The reliability of a
	fault tree, or the probability that all molecules have been catalyzed
	within two days, are instances of this question. Computing timed
	reachability probabilities reduces to computing transient probabilities
	in a uniformized CTMC, i.e., a CTMC in which the state residence times
	are ``normalized'' \cite{MCACTMC}. This method is
	quite efficient, numerically stable, and scales to CTMCs with millions
	of states.
	
	In practical applications, however, transition probabilities and state residence
	time distributions---defined by exit rates---are usually not known exactly. Component failure rates
	in fault trees are vulnerable to environmental conditions, and
	reaction rates of molecules are obtained experimentally. This raises
	the question whether CTMC model-checking results are robust w.r.t.
	perturbations of their stochastic aspects. The aim of this paper is to
	investigate to what extent transition probabilities and exit rates in a
	given CTMC can be altered while ensuring that timed reachability
	probabilities are preserved up to a small tolerance $\theta$.
	
	To this end we define the novel notion of
	$(\varepsilon,\delta)$-bisimulation on CTMCs, investigate its
	fundamental properties and derive bounds for timed reachability
	probabilities. The results yield under which (absolute) $\varepsilon$-tolerance on
	transition probabilities and (relative) $\delta$-tolerance on exit rates, timed
	reachability probabilities are close up to $\theta$. This enables, e.g.,
	to determine the maximal tolerances in components' failure rates while
	ensuring the fault tree's (i.e., overall systems') reliability. Our
	notion generalizes strong probabilistic bisimulation \cite{EOLFMC} (also known as lumpability) that  preserves timed reachability probabilities exactly. 
	
	Let us illustrate the conceptual difference of perturbing exit rates and transition probabilities separately in
	$(\varepsilon,\delta)$-bisimulations compared to existing notions such as 
	$\tau$-quasi-lumpability (also known as near-lumpability) \cite{EOLFMC,BQLMC,CBPIQLSWFN} that consider \emph{transition rates}, i.e., products of exit rates and transition probabilities.
		
		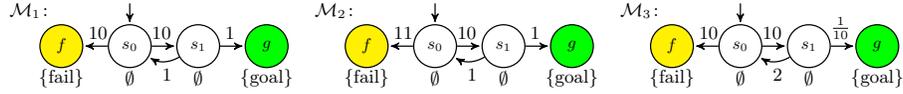
\begin{figure}[t]
			\centering
			\resizebox{!}{0.0715\textheight}{
				\begin{tikzpicture}[->,>=stealth',shorten >=1pt,auto, semithick]
						\tikzstyle{every state} = [text = black, scale = 0.9]
						
						\node[state] (s1) [] {$s_0$};  
						\node[state] (s2) [ right of = s1, node distance = 1.3cm] {$s_1$}; 
						\node[state] (g1) [fill = green, right of = s2, node distance = 1.3cm] {$g$}; 
						\node[state] (f1) [left of = s1, node distance = 1.3cm, fill = yellow] {$f$}; 
						\node[left of = f1,node distance=.5cm, yshift = 0.6cm] {$\mathcal{M}_1\colon$};
						
						\node[] (sinit) [above of = s1, node distance = 0.85cm] {}; 
						
						\path 
						(sinit) edge (s1)
						(s1) edge node [above, pos = 0.4] {$10$} (s2)
						(s2) edge node [above, pos = 0.4] {$1$} (g1)
						(s2) edge [bend left] node [below, pos = 0.4] {$1$} (s1)
						(s1) edge node [above, pos = 0.4] {$10$} (f1)
						;
						
						\node [below of = s1, node distance = 0.55cm] {$\emptyset$};
						\node [below of = s2, node distance = 0.55cm] {$\emptyset$};
						
						\node [below of = f1, node distance = 0.55cm] {$\{\text{fail}\}$};
						\node [below of = g1, node distance = 0.55cm] {$\{\text{goal}\}$}; 
						
						\node[state] (s11) [right of = s1, node distance = 5.8cm] {$s_0$}; 
						\node[state] (s21) [ right of = s11, node distance = 1.3cm] {$s_1$}; 
						\node[state] (g11) [fill = green, right of = s21, node distance = 1.3cm] {$g$}; 
						\node[state] (f11) [left of = s11, node distance = 1.3cm, fill = yellow] {$f$}; 
						\node[left of = f11,node distance=.5cm, yshift = 0.6cm] {$\mathcal{M}_2\colon$};
						
						\node[] (sinit1) [above of = s11, node distance = 0.85cm] {}; 
						
						\path 
						(sinit1) edge (s11)
						(s11) edge node [above, pos = 0.4] {$10$} (s21)
						(s21) edge node [above, pos = 0.4] {$1$} (g11)
						(s21) edge [bend left] node [below, pos = 0.4] {$1$} (s11)
						(s11) edge node [above, pos = 0.4] {$11$} (f11)
						;
						
						\node [below of = s11, node distance = 0.55cm] {$\emptyset$};
						\node [below of = s21, node distance = 0.55cm] {$\emptyset$};
						
						\node [below of = f11, node distance = 0.55cm] {$\{\text{fail}\}$};
						\node [below of = g11, node distance = 0.55cm] {$\{\text{goal}\}$};

						\node[state] (s12) [right of = s11, node distance = 5.8cm] {$s_0$}; 
						\node[state] (s22) [ right of = s12, node distance = 1.3cm] {$s_1$}; 
						\node[state] (g12) [fill = green, right of = s22, node distance = 1.3cm] {$g$}; 
						\node[state] (f12) [left of = s12, node distance = 1.3cm, fill = yellow] {$f$}; 
						\node[left of = f12,node distance=.5cm, yshift = 0.6cm] {$\mathcal{M}_3\colon$};
						
						\node[] (sinit2) [above of = s12, node distance = 0.85cm] {}; 
						
						\path 
						(sinit2) edge (s12)
						(s12) edge node [above, pos = 0.4] {$10$} (s22)
						(s22) edge node [above, pos = 0.4] {$\frac{1}{10}$} (g12)
						(s22) edge [bend left] node [below, pos = 0.4] {$2$} (s12)
						(s12) edge node [above, pos = 0.4] {$10$} (f12)
						;
						
						\node [below of = s12, node distance = 0.55cm] {$\emptyset$};
						
						\node [below of = s22, node distance = 0.55cm] {$\emptyset$};
						\node [below of = f12, node distance = 0.55cm] {$\{\text{fail}\}$};
						\node [below of = g12, node distance = 0.55cm] {$\{\text{goal}\}$}; 

					\end{tikzpicture}}
			\caption{Three CTMCs where the copies of the states  can be related by a $\tau$-quasi-lumpability iff $\tau \geq 1$. The numbers indicate the transition rates between states.}
			\label{fig:example-untinuitive-quasi-lumpabilitiy}
		\end{figure}
	
	\begin{example}\label{ex:intro}
		Consider the CTMCs in \Cref{fig:example-untinuitive-quasi-lumpabilitiy}. The value attached to an arrow from state $s$ to state $s'$ describes the \emph{transition rate} $R(s,s')$ between $s$ and $s'$, given as the product of the total exit rate $\rate(s)$ of $s$ and the probability $\prob(s,s')$ to move from $s$ to $s'$. The colors of the states indicate their labels (also written next to the states). Note that the CTMCs $\mathcal{M}_1$ and $\mathcal{M}_2$ behave very similarly---in $\mathcal{M}_2$, the transition to $f$ is only taken with a slightly higher rate, and the probability to reach the goal state $g$ until time $t$ is almost the same in the two models. The CTMCs $\mathcal{M}_1$ and $\mathcal{M}_3$, however, behave very differently---reaching $g$ until $t$ is much less likely in $\mathcal{M}_3$ because of the significantly lower rate $R(s_1, g)$.
		
		A $\tau$-quasi-lumpability \cite{EOLFMC,BQLMC,CBPIQLSWFN} is an equivalence relation such that related states have the same label and the total transition rates to move to any other equivalence class from related states differ by at most $\tau$.
		If we want to relate the states of $\mathcal{M}_1$ with their copies in $\mathcal{M}_2$ with a $\tau$-quasi-lumpability, $\tau$ has to be at least $1$ due to the change of the transition rate from $s_0$ to $f$ from $10$ to $11$.
		For $\tau=1$, however, we can also relate the states of $\mathcal{M}_1$ with their copies in $\mathcal{M}_3$. Hence, $\tau$-quasi-lumpability does not allow us to capture the intuition that $\mathcal{M}_1$ and $\mathcal{M}_2$ behave similarly while $\mathcal{M}_1$ and $\mathcal{M}_3$ behave very differently.
		
	\end{example}

	The central idea behind $(\epsilon,\delta)$-bisimulations is to decouple the changes of exit rates and transition probabilities of related states.
	In $\mathcal{M}_1$ and $\mathcal{M}_2$ from \Cref{fig:example-untinuitive-quasi-lumpabilitiy}, the  total exit rate $\rate(s_0) = \sum_{s'} R(s_0, s')$ of $s_0$ changes by a factor of $1.05$ from $20$ to $21$, and the probability $\prob(s_0, f) = \frac{R(s_0,f)}{\rate(s_0)}$ to transition from $s_0$ to $f$ is $1/2$ and $11/21$, respectively.
	The absolute change of transition probabilities and the relative change of exit rates is small, so the chains behave similarly.
	For $\mathcal{M}_1$ and $\mathcal{M}_3$, however, $\prob(s_1, g)$ changes from $1/2$ to $1/21$, a huge difference.
	In this way, $(\epsilon,\delta)$-bisimulation can express (dis)similarities of CTMCs in a more fine-grained manner than notions that consider the transition rates between states. 
	
	\paragraph{Main Contributions.} Our main contributions are the following: 
	\begin{itemize}
		\item We introduce \emph{$(\varepsilon,\delta)$-bisimulation} for CTMCs, a novel type of approximate probabilistic bisimulation that allows an absolute tolerance of $\varepsilon$ for the transition probabilities of related states, and a relative tolerance of $\delta$ for their total exit rates. We prove that the union of all $(\varepsilon, \delta)$-bisimulations, called $(\varepsilon,\delta)$-bisimilarity, is an $(\varepsilon, \delta)$-bisimulation itself, that it is additive in $\varepsilon$ and $\delta$, and that it coincides with strong bisimilarity \cite{CBTSMC,EOLFMC} iff $\varepsilon = \delta = 0$. Moreover, we discuss how quasi-lumpability \cite{EOLFMC,BQLMC,CBPIQLSWFN} and $(\varepsilon, \delta)$-bisimilarity relate, and show how to ``split'' $(\varepsilon, \delta)$-bisimilarity into $(\varepsilon,0)$- and $(0,\delta)$-bisimilarity, allowing an individual treatment of the parameters (\Cref{sec:definition-fundamental-properties}).
		\item We derive bounds on the absolute difference of timed reachability probabilities of $(\varepsilon, \delta)$-bisimilar states by uniformizing the CTMC and applying a bound from \cite{RBBTEAPC} for the difference in finite horizon reachability probabilities of $\varepsilon$-bisimilar states of DTMCs \cite{AAPP}. The bounds are tight if $\delta = 0$ (\Cref{sec:time-bounded-reachability-bounds}). 
		\item We analyze the absolute difference of timed reachability probabilities between $(0, \delta)$-bisimilar states in more detail. Using so-called \emph{Erlang CTMCs} \cite{SSPTPSATE} we show how to compute them \emph{exactly} if $\varepsilon = 0$. Subsequently, we derive different bounds on the difference: an easy-to-compute one based on an Erlang CTMC of a specific length, and bounds that utilize a spectral decomposition of the transition probability matrix of the given CTMC (\Cref{sec:errors-0-delta-states}). 
		\item We extend (some of) our results to reward-bounded reachability probabilities in CTMCs with nonnegative rewards, by utilizing a method from \cite{LCPP} that allows us to express them as timed reachability probabilities (\Cref{sec:reward-bounds}). 
	\end{itemize}
	
	\paragraph{Related Work.}
	In the discrete-time setting of DTMCs a lot of work has been done on approximate probabilistic bisimulations, mainly focusing on $\varepsilon$-bisimulations as introduced by Desharnais \textit{et al.} \cite{RBBTEAPC,AAPP}. Other notions include approximate probabilistic bisimulations with precision $\varepsilon$ \cite{AMBPBGSSMP,PMCLMPFAB,RPCTLMC}, up-to-$(n, \varepsilon)$-bisimulations \cite{SAAPBPCTL,AAPP}, or approximate versions of weak- and branching probabilistic bisimulation \cite{NAWPB,SAPB}. See, e.g., \cite{SAPB} for a comparison of several different notions.
	
	For CTMCs, an approximate version of lumpability, called $\tau$-\emph{quasi}- or \emph{near-lumpability}, is known \cite{EOLFMC,BQLMC,CBPIQLSWFN}. It is best suited for the analysis of long-run or stationary properties of a chain, but does not provide good guarantees on its \emph{transient} behavior. A more recent notion is \emph{proportional lumpability} \cite{PL,PLPB,RPL,EAPLASMPB}. Intuitively, states $s, s'$ are proportionally lumpable w.r.t. a function $\kappa$ that maps states to values in $\mathbb{R}_{>0}$ iff there is an equivalence $R$ such that the transition rates of $s$ and $s'$ to any equivalence class of $R$ differ by the constant factors $\kappa(s)$ resp. $\kappa(s')$. Proportional lumpability preserves the \emph{exact} stationary distributions.
	
	\emph{Uncertain continuous-time Markov chains} (\emph{UCTMC}) \cite{LUCTMC,AMUCTMC} are CTMCs with time-varying transition rates. At each time $t \geq 0$ and for all states $s,s'$, the transition rate to move from $s$ to $s'$ must be contained in a fixed interval $[m(s,s'), M(s,s')] \subseteq \mathbb{R}_{\geq 0}$. Lumpabilities for UCTMCs are equivalences that require related states to have the same ``extremal realizations'', defined as the (time-invariant) CTMCs in which all transition rates $R(s,s')$ coincide with $m(s,s')$ resp. $M(s,s')$. These lumpabilities can be computed efficiently and are characterized both via value functions and satisfaction equivalence of formulas in a variant of the continuous stochastic logic CSL \cite{MCCTMC,MCACTMC}. However, there does not seem to be a direct connection to our notion of $(\varepsilon, \delta)$-bisimulation.
	
	Another related line of research is \emph{perturbation theory} for Markov chains. Given a CTMC $\mathcal{M}$ and a (slightly) perturbed $\mathcal{M}'$, the goal is to bound the difference of (some) performance measures of the models. See, e.g., \cite{PACTMCWS,PACTMC,EBATAMCPM,SCUEMC,PTEMCANA} or the overview article \cite{APBFCTMC}. 
	The bounds presented in the literature are usually on the \emph{total} error between the \emph{stationary} distributions of $\mathcal{M}$ and $\mathcal{M}'$. 
	Furthermore, they are oftentimes obtained under the assumption of \emph{drift conditions} or rely on the \emph{ergodicity} of the chain. 
	This is in contrast to our work, which focuses on the \emph{componentwise} difference in \emph{transient} reachability probabilities up to time $t$.
			
	\section{Preliminaries}\label{sec:preliminaries}
	\textbf{Distributions and vectors.} Let $S \neq \emptyset$ be finite. The set of \emph{distributions} on $S$ is $\Distr(S) = \{\mu \colon S \to [0,1] \mid \sum_{s \in S} \mu(s) = 1\}$. The \emph{support} of $\mu \in \Distr(S)$ is $\supp(\mu) = \{s \in S \mid \mu(s) > 0\}$. Given $A \subseteq S$, we set $\mu(A) = \sum_{a \in A} \mu(a)$. The \emph{Dirac-distribution} $\dirac_s$ satisfies $\dirac_s(s') = 1$ if $s = s'$ and $\dirac_s(s') = 0$ otherwise. We may associate $\mu \in \Distr(S)$ with a row vector $\boldsymbol{\mu} \in [0,1]^{1 \times \vert S \vert}$ such that $\boldsymbol{\mu}[i] = \mu(i)$ for every $i \in S$. Vectors are written in bold face, and $\boldsymbol{\mu}^\top$ is the transpose of $\boldsymbol{\mu}$. 
	
	\smallskip
	\noindent
	\textbf{Relations.} Given a relation $R \subseteq S \times S$, the \emph{image} of $A \subseteq S$ under $R$ is $R(A) = \{t \in S \mid \exists \, s \in A \colon (s,t) \in R\}$. $A$ is called \emph{$R$-closed} if $R(A) \subseteq A$. If $R$ is an equivalence, we write $\xfrac{S}{R}$ for the set of $R$ equivalence classes. The $R$-closed sets of an equivalence $R$ are precisely the (unions of) elements of $\xfrac{S}{R}$.
	
	\smallskip
	\noindent
	\textbf{Markov chains.} Fix a countable set $AP$ of \emph{atomic propositions}. A \emph{discrete-time Markov chain} (\emph{DTMC}) is a tuple $\mathcal{D} = (\stateSpace, \prob, \initialState, \labelFunction)$, with $\stateSpace$ a finite set of \emph{states}, $\prob \colon \stateSpace \to \Distr(\stateSpace)$ a \emph{transition distribution function}, $\initialState \in \stateSpace$ a unique \emph{initial state} and $\labelFunction\colon \stateSpace \to 2^{AP}$ a \emph{labeling function}. $\prob(s,s') = \prob(s)(s')$ denotes the probability to move from $s$ to $s'$ in one step, and we write $\probMatrix \in [0,1]^{\vert S \vert \times \vert S \vert}$ for the \emph{transition probability matrix} of $\mathcal{D}$ with entries $\probMatrix_{s, s'} = \prob(s, s')$. $\Succ(s) = \{s' \in S \mid \prob(s,s') > 0\}$ is the set of \emph{successors} of $s$. Given $s \in \stateSpace$, let $\mathcal{D}_s$ be the DTMC that is exactly like $\mathcal{D}$, but with initial state $s$. The \emph{direct sum} $\mathcal{D} \oplus \mathcal{D}'$ of DTMCs $\mathcal{D}, \mathcal{D}'$ is the DTMC obtained from the disjoint union of $\mathcal{D}$ and $\mathcal{D}'$. The initial state of $\mathcal{D} \oplus \mathcal{D}'$ is not relevant for our purposes. 
	
	A \emph{continuous-time Markov chain} (\emph{CTMC}) is a tuple $\mathcal{M} = (\stateSpace, \prob, \rate, \initialState, \labelFunction)$ such that $\mathcal{D}_\mathcal{M} = (\stateSpace, \prob, \initialState, \labelFunction)$ is a DTMC, called the \emph{embedded} or \emph{underlying DTMC} of $\mathcal{M}$, and $\rate\colon \stateSpace \to \mathbb{R}_{> 0}$ is an \emph{exit rate function}. We use the same notations for CTMCs as for DTMCs, and we let $\mathcal{M}, \mathcal{N}$ range over CTMCs. Note that we do not exclude the possibility of $\prob(s,s) > 0$, i.e., we allow self-loops in CTMCs. The residence time in a state $s$ of $\mathcal{M}$ is negative exponentially distributed with rate $\rate(s)$, so the probability to take \emph{any} outgoing transition of $s$ until time $t \geq 0$ (including self-loops) is $1 - e^{-\rate(s) \cdot t}$, and the probability to take a transition from $s$ to some specific state $s'$ until $t$ is $\prob(s,s') \cdot (1-e^{-\rate(s) \cdot t})$. $\generatorMatrix \in \mathbb{R}^{\vert \stateSpace \vert \times \vert \stateSpace \vert}$ denotes the (\emph{infinitesimal}) \emph{generator} of $\mathcal{M}$. The entries of $\generatorMatrix$ are given as $\generatorMatrix_{i,j} = \prob(i,j) \cdot \rate(i)$ if $i \neq j$ and $\generatorMatrix_{i,i} = - \sum_{j \neq i} \generatorMatrix_{i,j}$. For more information on CTMCs and DTMCs see, e.g., \cite{PoMC,FMC,IMASS}. 
	
	\smallskip
	\noindent
	\textbf{Paths and probability measures.} Let $\mathcal{D}$ be a DTMC. A sequence $\sigma = s_0s_1\ldots \in S^{\omega}$ is an (\emph{infinite}) \emph{path} of $\mathcal{D}$ if $s_{i+1} \in Succ(s_i)$ for all $i \in \mathbb{N}$.
	$\sigma[i]=s_i$ is the state at position $i$ of $\sigma$, and  $trace(\sigma) = \labelFunction(s_0)\labelFunction(s_1)\ldots \in (2^{AP})^\omega$ is the \emph{trace} of $\sigma$.
	The set of infinite paths is $\Paths(\mathcal{D})$. \emph{Finite} paths $\sigma = s_0s_1\ldots s_k \in \stateSpace^{k+1}$ and their traces are defined similarly. $\Paths^*(\mathcal{D})$ is the set of finite paths of $\mathcal{D}$.
	
	Let $s \in S$. We consider the standard probability measure $\probMeasure_s^{\mathcal{D}}$ on subsets of $\Paths(\mathcal{D})$, defined via \emph{cylinder sets} $\mathit{Cyl}(\rho) = \{ \sigma \in \Paths(\mathcal{D}) \mid \rho \text{ is a prefix of } \sigma \}$ for $\rho \in \Paths^*(\mathcal{D})$.
	See, e.g., \cite{PoMC} for details.
	We write $\probMeasure_{}^{\mathcal{D}}$ for $\probMeasure_{\initialState}^{\mathcal{D}}$ and drop the superscript if $\mathcal{D}$ is clear from the context.
		
	Now, let $\mathcal{M}$ be a CTMC. $\sigma = s_0 t_0 s_1 t_1 \ldots \in (\stateSpace \cdot \mathbb{R}_{>0})^{\omega}$, where $\cdot$ denotes concatenation, is an (\emph{infinite}) \emph{timed path} of $\mathcal{M}$ if $s_{i+1} \in \Succ(s_i)$ for all $i \in \mathbb{N}$. $\sigma[i]$ and $trace(\sigma)$ are defined as for DTMCs. A \emph{finite} timed path is a finite prefix of a timed path that ends in a state. We write $\Paths(\mathcal{M})$ (resp. $\Paths^*(\mathcal{M})$) for the set of infinite (resp. finite) timed paths of $\mathcal{M}$. The value $\sigma_i = t_i$ describes the residence time in state $s_i$ along $\sigma$. Given some $t \in \mathbb{R}_{\geq 0}$, $\sigma @ t$ is the state of $\sigma$ at time $t$, i.e., $\sigma @ t = \sigma[k]$ for $k$ the smallest index such that $t \leq \sum_{i=0}^{k} t_i$ \cite{MCACTMC}. 
	
	Let $k \in \mathbb{N}$, $s, s_0, \ldots, s_{k+1} \in S$ and $I_0, \ldots, I_{k}$ nonempty intervals in $\mathbb{R}_{\geq 0}$. We consider the standard probability measure $\probMeasure_s^{\mathcal{M}}$ on subsets of $\Paths(\mathcal{M})$, defined via \emph{timed cylinder sets} $\mathit{Cyl}(s_0 I_0 \ldots s_{k} I_{k} s_{k+1}) = \{\sigma \in \Paths(\mathcal{M}) \mid \sigma[i] = s_i \text{ for } i \leq k+1 \text{ and} \sigma_j \in I_j \text{ for} j \leq k \}$. See, e.g., \cite{MCACTMC} for details. 
	We use the same abbreviations as for DTMCs regarding $\probMeasure^\mathcal{M}_s$. For a finite untimed path $\pi = s_0s_1\ldots s_n$, let $\probMeasure^{\mathcal{M}}(\pi)$ denote the probability of $\mathcal{M}$ to follow any timed path with a state sequence prefixed by $\pi$. 
	
	Given $G \subseteq \stateSpace$ and $t \geq 0$, $\probMeasure^\mathcal{M}_s(\reachability^{\leq t} G)$ is the probability to reach in $\mathcal{M}$ from $s \in \stateSpace$ a state in $G$ after time at most $t$. For a DTMC $\mathcal{D}$, $\probMeasure^\mathcal{D}_{s}(\reachability^{\leq k} G)$ is the probability to reach in $\mathcal{D}$ from $s$ a state in $G$ after at most $k \in \mathbb{N}$ steps.
	
	\smallskip
	\noindent
	\textbf{Transient probabilities.} 
	Let $n \in \mathbb{N}$. The probability of a DTMC $\mathcal{D}$ to be in state $s'$ after $n$ steps when starting in state $s$ is $\pi_n^\mathcal{D}(s,s') = \probMatrix^n_{s,s'}$. 
	
	For CTMC $\mathcal{M}$, the values $\pi^\mathcal{M}_t(s)$ of the \emph{transient probability distribution} $\pi^\mathcal{M}_t \in \Distr(S)$ describe the probabilities of $\mathcal{M}$ to be in $s$ at time $t$. Let $\boldsymbol{\pi}_t^\mathcal{M} \in [0,1]^{1 \times \vert S \vert}$ denote the vector representation of $\pi^\mathcal{M}_t$, and let $\pi^\mathcal{M}_t(s,s') = \pi^{\mathcal{M}_s}_t(s')$ be the probability to be in $s'$ after time $t$ when starting in $s$. $\boldsymbol{\pi}_t^\mathcal{M}$ can be computed by solving the following system of forward Kolmogorov differential equations \cite{ÜAMW}:
	\begin{center}
		$\frac{d}{dt} \boldsymbol{\pi}^\mathcal{M}_t = \boldsymbol{\pi}^\mathcal{M}_t \cdot \generatorMatrix \qquad \text{ given } \boldsymbol{\pi}^\mathcal{M}_0 = \mathbf{\dirac}_{\initialState}.$
	\end{center}
	Let $\uniformizationRate \geq \max_{s \in \stateSpace} \rate(s)$. The DTMC $\uniformization{\mathcal{M}}{\uniformizationRate} = (\stateSpace, \unifProb, \initialState, \labelFunction)$ with $\unifProb(s,s') = \frac{\prob(s,s') \cdot \rate(s)}{\uniformizationRate}$ if $s \neq s'$, and $\unifProb(s,s) = 1 + \frac{\prob(s,s) \cdot \rate(s)}{\uniformizationRate} - \frac{\rate(s)}{\uniformizationRate}$ is the \emph{uniformization} of $\mathcal{M}$ w.r.t. $\uniformizationRate$. $\boldsymbol{\pi}_t^\mathcal{M}$ can also be computed by the \emph{uniformization method} via \cite{TSMQS,MCASMP}
	\begin{center}
		$\boldsymbol{\pi}_t^\mathcal{M} = \sum_{k=0}^{\infty} e^{-\uniformizationRate \cdot t} \cdot \frac{(\uniformizationRate \cdot t)^k}{k!} \cdot \mathbf{\dirac}_{\initialState} \cdot \unifProbMatrix^k.$
	\end{center}
	We omit the superscript from $\pi_t^\mathcal{M}$ (resp. $\pi_n^\mathcal{D}$) if no confusion can arise. 
	
	\smallskip 
	\noindent
	\textbf{Bisimulation.} A \emph{(probabilistic) bisimulation} $R \subseteq S \times S$ for DTMC $\mathcal{D}$ is an equivalence such that for all $(s,s') \in R$ it holds that $\labelFunction(s) = \labelFunction(s')$ and $\prob(s,C) = \prob(s', C)$ for every $C \in \xfrac{\stateSpace}{R}$ \cite{BTPT}. For CTMC $\mathcal{M}$, an equivalence $R$ is a \emph{strong (probabilistic) bisimulation} if it is a bisimulation on $\mathcal{D}_\mathcal{M}$, and additionally satisfies $\rate(s)=\rate(s')$ for all $(s,s') \in R$ \cite{CBTSMC,EOLFMC}. States $s$ and $s'$ of $\mathcal{M}$ are \emph{strongly (probabilistic) bisimilar}, written $s \sim s'$, if there is a strong bisimulation $R$ on $\mathcal{M}$ with $(s,s') \in R$. CTMCs $\mathcal{M}$ and $\mathcal{N}$ are strongly (probabilistic) bisimilar, written $\mathcal{M}\sim\mathcal{N}$, if $\initialState^\mathcal{M} \sim \initialState^\mathcal{N}$ in $\mathcal{M} \oplus \mathcal{N}$. 
	
	\section{$(\varepsilon, \delta)$-Bisimulation on CTMCs} \label{sec:definition-fundamental-properties}
	If not specified otherwise, we always assume that $\varepsilon, \delta \geq 0$. In this section, after a short recap on $\varepsilon$-bisimulation for DTMCs (\Cref{sec:epsilon-bisim-dtmc}), we formally define $(\varepsilon, \delta)$-bisimulation for CTMCs (\Cref{sec:def-epsilon-delta}) and establish fundamental properties of this notion (\Cref{sec:fundamental-properties}).
	
	\subsection{$\varepsilon$-Bisimulation in the Discrete-Time Setting}\label{sec:epsilon-bisim-dtmc}
	In the discrete-time setting of DTMCs, the most well-known and well-studied \cite{RBBTEAPC,AAPP,ABM} notion of approximate probabilistic bisimulation are \emph{$\varepsilon$-bisimulations}. They were introduced by Desharnais \textit{et al.} \cite {AAPP} for \emph{labeled Markov processes} \cite{LCBLMP,BfLMP} and have since been adjusted to other models like DTMC \cite{RBBTEAPC,ABM}. 
	\begin{definition}[\cite{RBBTEAPC,AAPP}]\label{def:epsilon-bisimulation}
		Let $\mathcal{D}$ be a DTMC. A reflexive\footnote{In contrast to \cite{RBBTEAPC,AAPP} we require reflexivity of $\varepsilon$-bisimulations. This assumption is rather natural (a state should always simulate itself) and does not affect $\sim_{\varepsilon}$.} and symmetric relation $R \subseteq \stateSpace \times \stateSpace$ is an \emph{$\varepsilon$-bisimulation} if for all $(s,t) \in R$ and any $A \subseteq \stateSpace$ 
		\begin{center}
			$(\text{\emph{i}}) \, \, \labelFunction(s) = \labelFunction(t) \quad \text{ and } \quad (\text{\emph{ii}}) \, \, \prob(s,A) \leq \prob(t,R(A)) + \varepsilon. $
		\end{center}
		States $s$ and $t$ are \emph{$\varepsilon$-bisimilar}, denoted $s \sim_{\varepsilon} t$, if there is an $\varepsilon$-bisimulation $R$ with $(s,t) \in R$. DTMCs $\mathcal{D}_1$ and $\mathcal{D}_2$ are  $\varepsilon$-bisimilar, denoted $\mathcal{D}_1 \sim_\varepsilon \mathcal{D}_2$, if $\initialState^{\mathcal{D}_1} \sim_\varepsilon \initialState^{\mathcal{D}_2}$ in $\mathcal{D}_1 \oplus \mathcal{D}_2$.
	\end{definition}
	
	The tolerance $\varepsilon$ describes by how much the transition probabilities of related states may differ. If $\varepsilon \approx 1$, many states can be related, but their behavior might be significantly different, whereas an $\varepsilon$ close to $0$ allows us to only relate states with almost the same behavior, decreasing the number of relatable states. 
	
	The satisfaction of condition (ii) of \Cref{def:epsilon-bisimulation} can be characterized by the existence of suitable \emph{weight functions} that describe how to split the successor probabilities of related states such that the condition holds. This approach is used in, e.g., \cite{FMOSCSMDP,ABM,CDBPA}. We will later use the following characterization (where we slightly abuse notation and write $\Delta_{s,t}(s', t')$ instead of $\Delta_{s,t}(s')(t')$). 
	\begin{lemma}[\cite{SAPB}]\label{lem:characterization-weight-functions-concur}
		A reflexive and symmetric relation $R \subseteq S \times S$ that only relates states with the same label is an $\varepsilon$-bisimulation iff for all $(s,t) \in R$ there is a map $\Delta_{s,t}\colon \Succ(s) \to \mathit{\Distr}(\Succ(t))$ such that 
		\begin{enumerate}
			\item
			for all $t^\prime \in\Succ(t) $ we have $\prob(t,t^\prime)= \sum_{s^\prime\in \Succ(s) } \prob(s,s^\prime)\cdot\Delta_{s,t}(s^\prime,t^\prime) $, and
			\item 
			$
			\sum_{s^\prime\in \Succ(s)} \prob(s,s^\prime)\cdot \Delta_{s,t}(s^\prime,R(s^\prime) \cap \Succ(t)) \geq 1 -\varepsilon
			$.
		\end{enumerate}
	\end{lemma}
	
	Furthermore, step-bounded reachability probabilities $\probMeasure_{*}(\lozenge^{\leq n} g)$ of a goal state $g$ for $\varepsilon$-bisimilar states $s,s'$ are related as follows.
	\begin{theorem}[\cite{RBBTEAPC}]\label{thm:main-result-rbbteapc}
		Let $\mathcal{D}$ be a DTMC, $k \in \mathbb{N}$ and $s \sim_\varepsilon s'$. Then 
		\begin{center}$\vert \probMeasure_s(\lozenge^{\leq k} g) - \probMeasure_{s'}(\lozenge^{\leq k} g) \vert \leq 1 - (1-\varepsilon)^k.$\end{center}
	\end{theorem}
	
	\subsection{Definition of $(\varepsilon,\delta)$-Bisimulation}\label{sec:def-epsilon-delta}
	In contrast to the discrete-time case, the behavior of a CTMC $\mathcal{M}$ is not only influenced by the transition probabilities, but also by the exit rates of the states. Thus, it is natural to consider two different tolerance values $\varepsilon$ and $\delta$ for the probabilities resp. the exit rates. As the values of $\prob(s, \cdot)$ are in $[0,1]$ for any $s \in \stateSpace$, similar to the discrete-time case an additive tolerance $\varepsilon$ is suitable. 
	
	The rates $\rate(s)$, however, can take on any (finite) positive value. This makes the use of \emph{additive} tolerances $\delta$ for the rates difficult. Consider, for example, a CTMC with states $s,s'$ such that $\rate(s) =1$ and $\rate(s') = 100$, and assume that the exit rates of related states are allowed to differ by at most $10\%$. For $s$, an additive $\delta$ would need a value of $0.1$ (i.e., all states with a rate in $[0.9, 1.1]$ are considered to behave almost the same as $s$), while $\delta = 10$ would be necessary for $s'$. The former allows virtually no tolerance for $\rate(s')$ (only $0.001\%$ instead of the desired $10\%$), while the latter allows a tolerance of up to $1000\%$ for $\rate(s)$. 
	
	To circumvent this issue we propose the use of a \emph{multiplicative} tolerance $\delta$ for the exit rates of related states, so that the tolerable difference is relative to the actual rates. For the rest of the paper, let $\ln(\cdot)$ denote the natural logarithm. \unskip
	\begin{definition}\label{def:epsilon-delta-bisimulation}
		Let $\mathcal{M}$ be a CTMC and let $R \subseteq S \times S$ be a reflexive and symmetric relation. $R$ is an $(\varepsilon, \delta)$-bisimulation if for all $(s,s') \in R$ it holds that
		\begin{enumerate}
			\item $\labelFunction(s) = \labelFunction(s')$ \hfill \emph{(labeling condition)}
			\item $\vert \ln(\rate(s)) - \ln(\rate(s')) \vert \leq \delta$ \hfill \emph{($\delta$-condition)}
			\item for all $A \subseteq \stateSpace$: $\prob(s, A) \leq \prob(s', R(A)) + \varepsilon$. \hfill \emph{($\varepsilon$-condition)}
		\end{enumerate}
		States $s, s' \in S$ are \emph{$(\varepsilon, \delta)$-bisimilar}, written $s \sim_{\varepsilon, \delta} s'$, if $(s, s') \in R$ for an $(\varepsilon, \delta)$-bisimulation $R$. CTMCs $\mathcal{M}_1$ and $\mathcal{M}_2$ are \emph{$(\varepsilon, \delta)$-bisimilar}, written $\mathcal{M}_1 \, {\sim_{\varepsilon, \delta}} \, \mathcal{M}_2$, if $\initialState^{\mathcal{M}_1} \, {\sim_{\varepsilon, \delta}} \, \initialState^{\mathcal{M}_2}$ in $\mathcal{M}_1 \oplus \mathcal{M}_2$. 
	\end{definition} 

	Any $(\varepsilon, \delta)$-bisimulation on $\mathcal{M}$ induces an $\varepsilon$-bisimulation on the embedded DTMC $\mathcal{D}_\mathcal{M}$, i.e., $s \sim_{\varepsilon, \delta}^\mathcal{M} t$ implies $s \sim_\varepsilon^{\mathcal{D}_\mathcal{M}} t$. Hence, there are weight functions $\Delta_{s,t}$ as in \Cref{lem:characterization-weight-functions-concur} (w.r.t. the probabilities of $\mathcal{M}$) for $(\varepsilon, \delta)$-bisimilar states $s,t$.   
	
	\begin{figure}[t!]
		\centering
		\resizebox{!}{0.15\textheight}{\begin{tikzpicture}[->,>=stealth',shorten >=1pt,auto, semithick]
			\tikzstyle{every state} = [text = black]
			
			\node[state] (s0) [fill = yellow] {$s_0$};
			\node[state] (s1) [right of = s0, fill = yellow, node distance = 2cm] {$s_1$};
			\node[state] (s2) [below of = s0, node distance = 1.5cm, fill = yellow] {$s_2$};
			\node[state] (s3) [left of = s0, node distance = 2cm, fill = yellow] {$s_3$};
			\node[state] (s4) [below of = s3, node distance = 1.5cm, fill = yellow] {$s_4$}; 
			\node[state] (g) [right of = s1, node distance = 2cm, fill = green] {$g$}; 
			\node (sinit) [above left of = s0, node distance = 1cm] {}; 
			
			\path 
			(sinit) edge (s0)
			(s0) edge node [above] {$\frac{1}{6}$} (s1)
			(s0) edge node [above] {$\frac{1}{3}$} (s3)
			(s0) edge node [right, pos = 0.4] {$\frac{1}{2}$} (s2)
			
			(s1) edge [loop below] node {$1{-}\varepsilon$} (s1)
			(s1) edge node [above] {$\varepsilon$} (g)
			
			(s2) edge [loop left] node [below, pos = 0.15] {$\frac{5}{6}{-}\varepsilon$} (s2)
			(s2) edge node [below, xshift = 0.2cm] {$\frac{1}{6}{+}\varepsilon$} (s1)
			
			(s3) edge [loop left] node {$\frac{5}{6}$} (s3)
			(s3) edge [bend left = 20] node [right, pos = 0.4] {$\frac{1}{6}$} (s4)
			
			(s4) edge [loop left] node  {$1{-}\frac{\varepsilon}{2}$} (s3)
			(s4) edge [bend left] node [pos = 0.4, left] {$\frac{\varepsilon}{2}$} (s3)
			
			(g) edge [loop right] node {$1$} (g)
			
			;
			
			\node[above of = s3, node distance = 0.7cm] {$\{a\}, e^{\frac{\delta}{2}}$};
			\node[below of = s4, node distance = 0.65cm] {$\{a\}, e^{-\frac{3\delta}{2}}$};
			\node[right of = s2, node distance = 0.9cm] {$\{a\}, e^{\delta}$};
			\node[above of = s1, node distance = 0.7cm] {$\{a\}, e^{-\delta}$};
			\node[above of = g, node distance = 0.7cm] {$\{b\}, 1$};
			\node[above of = s0, node distance = 0.7cm] {$\{a\}, 1$};
			
		\end{tikzpicture}}
		\caption{The CTMC used in \Cref{ex:epsilon-delta-bisim}.}
		\label{fig:example-epsilon-delta-bisim}
	\end{figure}
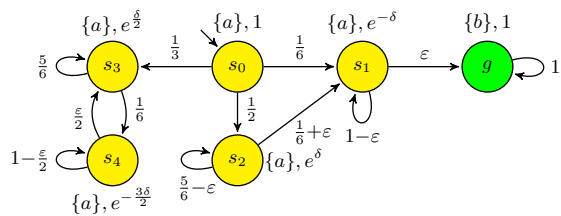
		\begin{example}\label{ex:epsilon-delta-bisim}
		Let $\varepsilon < \frac{1}{2}$ and $\delta > 0$. In \Cref{fig:example-epsilon-delta-bisim},   $\sim_{\varepsilon, \delta}$ is the reflexive and symmetric closure of $\{(s_0, s_2), (s_0, s_3), (s_1, s_4), (s_2, s_3)\}$. States $s_1$ and $s_2$ are not $(\varepsilon, \delta)$-bisimilar as they violate the $\delta$-condition: $\vert \ln (\rate(s_1)) - \ln(\rate(s_2)) \vert = 2\delta$. Moreover, \mbox{$s_0 \nsim_{\varepsilon, \delta} s_1$} since, although the pair of states satisfies the $\delta$-condition, the $\varepsilon$-condition is violated: $\prob(s_0, \{s_2\}) = \frac{1}{2} > \varepsilon = \prob(s_1, {\sim_{\varepsilon, \delta}}(\{s_2\})) + \varepsilon.$ 
		
		Note that the steady state distributions of $(\varepsilon, \delta)$-bisimilar states can differ significantly: starting from $s_0$ the probability to reach $g$ (and stay there forever) is $2/3$, from $s_2$ it is $1$ and
		from $s_3$ it is $0$, even though $s_0, s_2$, and $s_3$ are pairwise $(\varepsilon, \delta)$-bisimilar. 
	\end{example}
	
	\subsection{Fundamental Properties of $\sim_{\varepsilon, \delta}$}\label{sec:fundamental-properties}
	We now establish some fundamental properties of $\sim_{\varepsilon, \delta}$ that are generally desirable for approximate probabilistic bisimulations \cite{AAPP,SAPB}. 
	\begin{restatable}{theorem}{ThmFundamentalProperties}\label{thm:fundamental-properties}
		Let $\mathcal{M}$ be a CTMC and $s, s', s'' \in \stateSpace$.
		\begin{enumerate}
			\item ${\sim_{\varepsilon, \delta}}$ is the largest $(\varepsilon, \delta)$-bisimulation on $\mathcal{M}$.
			\item If $s \sim_{\varepsilon_1, \delta_1} s'$ and $s' \sim_{\varepsilon_2, \delta_2} s''$ then $s \sim_{\varepsilon_1 + \varepsilon_2, \delta_1 + \delta_2} s''$. 
			\item $s \sim s'$ iff $s \sim_{0,0} s'$.
		\end{enumerate}
	\end{restatable}
	
	\Cref{thm:fundamental-properties} shows that $\sim_{\varepsilon, \delta}$ is itself an $(\varepsilon, \delta)$-bisimulation, and that it is the largest such relation (item 1). Furthermore $\sim_{\varepsilon, \delta}$ is additive in both $\varepsilon$ and $\delta$ \mbox{(item 2)} and coincides with strong bisimilarity $\sim$ iff $\varepsilon = \delta = 0$ (item 3).
	
	By slightly adjusting a procedure proposed in \cite{AAPP} that constructs $\sim_\varepsilon$ for DTMCs by solving flow networks á la \cite{PTATPBS,PSfPP}, $\sim_{\varepsilon, \delta}$ can be computed in time polynomial in the number of states of $\mathcal{M}$. 
	
	\begin{restatable}{corollary}{CorAlgorithmicComplexity}\label{cor:algorithmic-complexity-sim-epsilon-delta}
		For a given CTMC $\mathcal{M}$, $\sim_{\varepsilon, \delta}$ can be computed in time $\mathcal{O}(\vert \stateSpace \vert^7)$.
	\end{restatable}
	
	Furthermore, $(\varepsilon, \delta)$-bisimilarity in $\mathcal{M}$ implies $\tau$-bisimilarity (for a $\tau \geq 0$ that depends on $\varepsilon$ and $\delta$) in the uniformization $\uniformization{\mathcal{M}}{\uniformizationRate}$ of $\mathcal{M}$ for $q \geq \max_{s \in \stateSpace} \rate(s)$. \unskip
	\begin{restatable}{lemma}{LemApproximatBisimulationInUniformization}\label{lem:approximate-bisimulation-in-uniformization}
		$s \sim_{\varepsilon, \delta}^\mathcal{M} s'$ implies $s \sim_\tau^{\uniformization{\mathcal{M}}{\uniformizationRate}} s'$, where $\tau = e^\delta \cdot(1 + \varepsilon) - 1$.
	\end{restatable}

	A notion from the literature related to $(\varepsilon, \delta)$-bisimulations is that of $\tau$-quasi-lumpability, which is also known as $\tau$- or near-lumpability \cite{EOLFMC,BQLMC,CBPIQLSWFN}. It is defined for partitions $\Omega = \{\Omega_1, \ldots, \Omega_m\}$ of the state space $\stateSpace$ of $\mathcal{M}$. More precisely, $\Omega$ is a $\tau$-quasi-lumpability for some $\tau \geq 0$ if for all $1 \leq i,j \leq m$ and all $s, s' \in \Omega_i$ it holds that $\vert \prob(s, \Omega_j) \cdot \rate(s) - \prob(s', \Omega_j) \cdot \rate(s')\vert \leq \tau$ \cite{PL}. The partitions of the state space induced by \emph{transitive} $(\varepsilon, \delta)$-bisimulations are quasi-lumpabilities. 	\begin{restatable}{proposition}{PropConnectionTauLump}
		Let $R$ be a transitive $(\varepsilon, \delta)$-bisimulation and $\uniformizationRate \geq \max_{s \in \stateSpace} \rate(s)$. The partition induced by $R$, i.e., the set $\xfrac{\stateSpace}{R}$, is a $q \cdot (e^{\delta} \cdot(1 + \varepsilon) - 1)$-lumpability.
	\end{restatable}
	
	On the other hand, states related by a $\tau$-lumpability might require large values of $\varepsilon, \delta$ to be able to relate them by an $(\varepsilon, \delta)$-bisimulation.
	\begin{restatable}{proposition}{PropTauLumpDoesNotImplyEpsilonDelta}\label{prop:quasi-lumpability-does-not-imply-epsilon-delta}
		For any given $\varepsilon \in (0,1), \delta > 0$ and $\tau > 0$ there is a CTMC $\mathcal{M} = \mathcal{M}(\varepsilon, \delta, \tau)$ with states $s,s'$ such that $\{s,s'\}$ is a block of a $\tau$-quasi lumpability on $\mathcal{M}$ but the pair $(s,s')$ does neither satisfy the $\varepsilon$- nor the $\delta$-condition.
	\end{restatable}
	\Cref{prop:quasi-lumpability-does-not-imply-epsilon-delta} illustrates the advantage of considering differences in the exit rates and transition probabilities of states separately when arguing about the similarity of their behavior, instead of considering the transition rates, i.e., the product of these two values, as is usually done in the literature \cite{EOLFMC,BQLMC,PL}.
	
	We finish the section by showing how to ``split'' $(\varepsilon,\delta)$-bisimilarity into $(\varepsilon, 0)$- and $(0, \delta)$-bisimilarity. More precisely, given $\mathcal{M} \sim_{\varepsilon, \delta} \mathcal{N}$ we construct CTMCs $\mathcal{M}', \mathcal{N}'$ with the same graph structure and such that \mbox{$\mathcal{M} \sim_{\varepsilon, 0} \mathcal{M}' \sim_{0, \delta} \mathcal{N}' \sim \mathcal{N}$}. This decomposition makes it possible to treat the two parameters individually and, together with the additivity of $\sim_{\varepsilon, \delta}$, allows us to extend results shown for $(\varepsilon,0)$- and $(0, \delta)$-bisimilar states or chains to $(\varepsilon, \delta)$-bisimilar ones. 
	\begin{restatable}{theorem}{ThmSameStructureConstruction}\label{thm:same-structure-no-additional-tolerance}
		Let $\mathcal{M} \sim_{\varepsilon, \delta} \mathcal{N}$. Then there are CTMCs $\mathcal{M}'$ and $\mathcal{N}'$ with the same graph structure and such that
$			\mathcal{M} \sim_{\varepsilon, 0} \mathcal{M}' \sim_{0, \delta} \mathcal{N}' \sim \mathcal{N}.$
	\end{restatable}
	\paragraph{Proof sketch.}
		We describe how to construct $\mathcal{M}'$ and $\mathcal{N}'$. The models share the state space $\stateSpace' = \stateSpace^\mathcal{M} \times \stateSpace^\mathcal{N}$, with initial state $(\initialState^\mathcal{M}, \initialState^\mathcal{N})$. The label function $\labelFunction'$ of both models is $\labelFunction'((s,t)) = \labelFunction^\mathcal{N}(t)$, and the exit rate functions are defined via 
		\begin{center}
			${\rate^\mathcal{M}}'((s,t)) = \begin{cases}
				\rate^\mathcal{M}(s), &\text{if } s \sim_{\varepsilon, \delta} \! \! t \\
				\rate^\mathcal{N}(t), &\text{if} s \nsim_{\varepsilon, \delta} t
			\end{cases} \qquad \text{and} \qquad
			{\rate^\mathcal{N}}'((s,t)) = \rate^\mathcal{N}(t).$
		\end{center}
		Furthermore, the common transition probability function of $\mathcal{M}'$ and $\mathcal{N}'$ is 
		\begin{align*}
			\prob'((s,t), (s',t')) &= \begin{cases}
				\Delta_{s,t}(s',t') \prob^\mathcal{M}(s,s'), &\text{if } s \sim_{\varepsilon, \delta} t, s' \in \Succ(s), t' \in \Succ(t) \\
				\prob^\mathcal{M}(s,s')  \prob^\mathcal{N}(t,t'), &\text{if } s \nsim_{\varepsilon, \delta} t \\
				0, &\text{otherwise}
			\end{cases}
		\end{align*}
		where, for a given pair of $(\varepsilon, \delta)$-bisimilar states $(s,t) \in \stateSpace^\mathcal{M} \times \stateSpace^\mathcal{N}$, the function $\Delta_{s,t} \colon \Succ(s) \to \Distr(\Succ(t))$ is a weight function as in \Cref{lem:characterization-weight-functions-concur}, which exists for all such $(s,t)$ since $s \sim_{\varepsilon} t$ in the underlying DTMCs of $\mathcal{M}$ and $\mathcal{N}$. The proof proceeds by showing the desired relations between the models.  \qed
		\begin{figure}[t!]
			\centering
			\resizebox{!}{0.28\textheight}{
			\begin{tikzpicture}[->,>=stealth',shorten >=1pt,auto, semithick]
				\tikzstyle{every state} = [text = black]
				
				\node[state] (s0) [fill = yellow] {$s_0$};
				\node (tempm) [right of = s0, node distance = 2.5cm] {};
				\node[state] (s1) [fill = yellow, above of = tempm, node distance = 0.75cm] {$s_1$};
				\node[state] (s2) [fill = green, below of = tempm, node distance = 0.75cm] {$s_2$}; 
				\node (minit) [above left of = s0, node distance = 1cm] {}; 
				
				\path
					(minit) edge (s0)
					(s0) edge node [above] {$\frac{1}{2}$} (s1)
					(s0) edge node [below] {$\frac{1}{2}$} (s2)
					(s1) edge [loop right] node {$\frac{1}{2}$} (s1)
					(s1) edge node {$\frac{1}{2}$} (s2)
					(s2) edge [loop right] node {$1$} (s2)
				;
				
				\node[below of = s0, node distance = 0.65cm] {$\{a\}, 1$}; 
				\node[above of = s1, node distance = 0.65cm] {$\{a\}, e^{\delta}$}; 
				\node[below of = s2, node distance = 0.65cm] {$\{b\}, e^{\frac{\delta}{2}}$};

				\node[state] (t0) [fill = yellow, below of = s0, node distance = 3.5cm] {$t_0$};
				\node (tempn) [right of = t0, node distance = 2.5cm] {};
				\node[state] (t1) [above of = tempn, node distance = 0.75cm] {$t_1$};
				\node[state] (t2) [fill = green, below of = tempn, node distance = 0.75cm] {$t_2$}; 
				\node (ninit) [above left of = t0, node distance = 1cm] {}; 
				
				\path
				(ninit) edge (t0)
				(t0) edge node [above] {$\varepsilon$} (t1)
				(t0) edge[loop above] node {$\frac{1}{2}{-}\varepsilon$} (t1)
				(t0) edge node [below] {$\frac{1}{2}$} (t2)
				(t1) edge [loop right] node {$1$} (t1)
				(t2) edge node {$\varepsilon$} (t1)
				(t2) edge [loop right] node {$1{-}\varepsilon$} (t2)
				;
				
				\node[below of = t0, node distance = 0.65cm] {$\{a\}, e^{\frac{\delta}{2}}$}; 
				\node[above of = t1, node distance = 0.65cm] {$\{c\}, e^{-\delta}$}; 
				\node[below of = t2, node distance = 0.65cm] {$\{b\}, e^{-\frac{\delta}{2}}$};

				\node[] (tempmid) [below of = s2, node distance = 1cm] {};
				\node[state, right of = tempmid, node distance = 3.5cm, fill = yellow, inner sep = 0pt] (s0t0) {$(s_0, t_0)$}; 
				\node[state] (s1t0) [right of = s0t0, node distance = 2cm, fill = yellow, inner sep = 0pt] {$(s_1, t_0)$}; 
				\node[state] (s1t1) [above of = s1t0, node distance = 2cm, inner sep = 0pt] {$(s_1, t_1)$}; 
				\node[state] (s2t2) [below of = s1t0, node distance = 2cm, inner sep = 0pt, fill = green] {$(s_2, t_2)$}; 
				\node[state] (s2t1) [right of = s1t1, node distance = 2cm, inner sep = 0pt] {$(s_2, t_1)$};
				\node (mninit) [above left of = s0t0, node distance = 1.2cm] {};
				
				\path 
					(mninit) edge (s0t0)
					(s0t0) edge [bend left = 20] node [above] {$\varepsilon$} (s1t1)
					(s0t0) edge node [above] {$\frac{1}{2}{-} \varepsilon$} (s1t0)
					(s0t0) edge [bend right = 20] node [below] {$\frac{1}{2}$} (s2t2)
					
					(s1t1) edge [loop above] node[left, pos = 0.15] {$\frac{1}{2}$} (s1t1)
					(s1t1) edge node [above] {$\frac{1}{2}$} (s2t1)
					
					(s1t0) edge node [left, pos = 0.4] {$\varepsilon$} (s1t1)
					(s1t0) edge [loop right] node [below, pos = 0.85] {$\frac{1}{2}{-}\varepsilon$} (s1t0)
					(s1t0) edge node [left, pos = 0.4] {$\frac{1}{2}$} (s2t2)
					
					(s2t2) edge [loop below] node [pos = 0.15, right] {$1{-}\varepsilon$} (s2t2)
					(s2t2) edge [bend right = 30] node [below, yshift = -0.2cm] {$\varepsilon$} (s2t1)
				
					(s2t1) edge [loop above] node [left, pos = 0.15] {$1$} (s2t1)
				;
				
				\node [below of = s0t0, node distance = 0.75cm, xshift = -0.5cm] {$\{a\}$, \textcolor{blue}{$1$}, \textcolor{red}{$e^{\frac{\delta}{2}}$}};
				\node [right of = s2t2, node distance = 1.5cm] {$\{b\}$, \textcolor{blue}{$e^{\frac{\delta}{2}}$}, \textcolor{red}{$e^{-\frac{\delta}{2}}$}};
				\node [above of = s1t0, node distance = 0.65cm, xshift = 0.95cm] {$\{a\}$, \textcolor{blue}{$e^{\delta}$}, \textcolor{red}{$e^{\frac{\delta}{2}}$}};
				\node [left of = s1t1, node distance = 1.5cm] {$\{c\}$, \textcolor{blue}{$e^{-\delta}$}, \textcolor{red}{$e^{-\delta}$}};
				\node [below of = s2t1, node distance = 0.65cm, xshift = 1.1cm] {$\{c\}$, \textcolor{blue}{$e^{-\delta}$}, \textcolor{red}{$e^{-\delta}$}};
				
			\end{tikzpicture}}
			\caption{The CTMCs used in \Cref{ex:same-structure}.}
			\label{fig:example-construction}
		\end{figure}
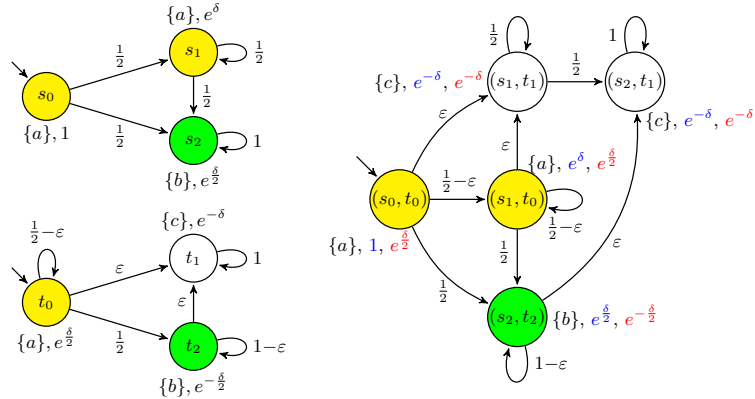
		
		\begin{example}\label{ex:same-structure}
			We illustrate the construction of \Cref{thm:same-structure-no-additional-tolerance} in \Cref{fig:example-construction}. The CTMCs $\mathcal{M}$ (top left) and $\mathcal{N}$ (bottom left) are $(\varepsilon, \delta)$-bisimilar, as all states with the same label are pairwise $(\varepsilon, \delta)$-bisimilar in $\mathcal{M} \oplus \mathcal{N}$. 
			The chains $\mathcal{M}'$ and $\mathcal{N}'$, which only differ in their exit rate functions $\rate^{\mathcal{M}'}$(first, in blue) and $\rate^{\mathcal{N}'}$(second, in red) are on the right of the figure. To construct $\mathcal{M}'$ and $\mathcal{N}'$ it is necessary to compute weight functions $\Delta_{s,t}$ for all $s \sim_{\varepsilon, \delta} t$. For example, a suitable choice for $(s_0, t_0)$ is $\Delta_{s_0, t_0}(s_1, t_0) = 1 - 2 \cdot \varepsilon, \Delta_{s_0, t_0}(s_1, t_1) = 2 \cdot \varepsilon, \Delta_{s_0, t_0}(s_2, t_2) = 1$ and $\Delta_{s_0, t_0}(\cdot, \cdot) = 0$ otherwise. It is easy to prove that $\mathcal{M} \sim_{\varepsilon, 0} \mathcal{M}' \sim_{0, \delta} \mathcal{N}' \sim \mathcal{N}$.
		\end{example}
		
	
	\section{Bounding Timed Reachability Probabilities}\label{sec:time-bounded-reachability-bounds}
	Let $\mathcal{M}$ be a CTMC and $t \in \mathbb{R}_{\geq 0}$. We are interested in bounds for the absolute difference of the probabilities of $(\varepsilon, \delta)$-bisimilar states $s$ and $s'$ to reach a unique \emph{goal state} $g$ until the deadline $t$, i.e., in bounds for $\vert \probMeasure_s(\lozenge^{\leq t} g) - \probMeasure_{s'}(\lozenge^{\leq t} g) \vert$. 
	
	For the rest of the paper we therefore assume that $\mathcal{M}$ has such a unique goal state $g$, which is w.l.o.g. \cite{MCACTMC} absorbing and uniquely labeled.
	 If there are multiple, potentially non-absorbing goal states, these states can be collapsed to a single absorbing goal state without affecting the time-bounded reachability probabilities.
	The same pre-processing is also applied in the computation of satisfaction probabilities for time-bounded until-formulas of the \emph{continuous stochastic logic CSL} \cite{MCCTMC,MCACTMC}.
	Similarly, we assume that all states from which $g$ is not reachable are collapsed to a single, uniquely labeled absorbing fail state. Hence, there are at most two absorbing states in $\mathcal{M}$ which are not $(\varepsilon, \delta)$-bisimilar to any other state (because of the unique labels) and are eventually reached almost surely.
	 Further, let $\uniformizationRate = \max_{s \in \stateSpace} \rate(s)$ be the smallest possible uniformization rate of $\mathcal{M}$.
	
	By considering the DTMC $\uniformization{\mathcal{M}}{\uniformizationRate}$ and applying the bound of \Cref{thm:main-result-rbbteapc} \cite{RBBTEAPC}, which is possible because of the preservation of approximate bisimilarity in the uniformization (see \Cref{lem:approximate-bisimulation-in-uniformization}), we derive an upper bound on the absolute difference of timed reachability probabilities of $(\varepsilon, \delta)$-bisimilar states. 
	\begin{restatable}{proposition}{PropTransientBoundsEpsilonDeltaGreateZero}\label{prop:transient-prob-bounds-epsilon-delta-greater-0}
		For $ s \sim_{\varepsilon, \delta} s'$: 
		$\vert \probMeasure_s(\reachability^{\leq t} g) - \probMeasure_{s'}(\reachability^{\leq t} g) \vert \leq 1 - e^{-\uniformizationRate  t  (e^\delta(1 + \varepsilon)-1)}$.
	\end{restatable}
	
	Considering $\varepsilon = 0$ or $\delta  = 0$ in \Cref{prop:transient-prob-bounds-epsilon-delta-greater-0} yields the following corollary. 
	\begin{corollary}\label{cor:transient-prob-bounds-epsilon-or-delta-0}
		Let $ s \sim_{\varepsilon, \delta} s'$.
		\begin{enumerate}
			\item If $\delta = 0$ then $\vert \probMeasure_s(\reachability^{\leq t} g) - \probMeasure_{s'}(\reachability^{\leq t} g) \vert \leq 1 - e^{-q  t  \varepsilon}.$
			\item If $\varepsilon = 0$ then $	\vert \probMeasure_s(\reachability^{\leq t} g) - \probMeasure_{s'}(\reachability^{\leq t} g) \vert \leq 1 - e^{-q  t  (e^\delta -1)}.$
			\item If $\varepsilon = \delta = 0$ then $\probMeasure_s(\reachability^{\leq t} g) = \probMeasure_{s'}(\reachability^{\leq t} g)$.
		\end{enumerate}
	\end{corollary}
	
	The third result of \Cref{cor:transient-prob-bounds-epsilon-or-delta-0} also directly follows from $s \sim_{0, 0} s'$ iff \mbox{$s \sim s'$}, which was proved in the third item of \Cref{thm:fundamental-properties}, and the fact that $\sim$ exactly preserves transient (reachability) probabilities \cite{MCACTMC,EOLFMC}. 
	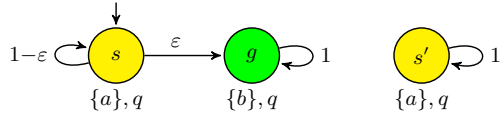
\begin{figure}[t!]
		\centering
		\resizebox{!}{0.09\textheight}{
		\begin{tikzpicture}[->,>=stealth',shorten >=1pt,auto, semithick]
			\tikzstyle{every state} = [text = black]
			
			\node[state] (s) [fill = yellow] {$s$}; 
			\node[state] (g) [fill = green, right of = s, node distance = 2cm] {$g$}; 
			\node[state] (s2) [right of = g, node distance = 2.5cm, fill = yellow] {$s'$}; 
			
			\node[] (sinit) [above of = s, node distance = 0.9cm] {}; 
			
			\path 
			(sinit) edge (s)
			(s) edge node [above, pos = 0.4] {$\varepsilon$} (g)
			(s) edge [loop left] node [pos = 0.5, left] {$1{-}\varepsilon$} (s)
			(g) edge [loop right] node {$1$} (g)
			(s2) edge [loop right] node {$1$} (s2)
			;

			\node[below of = s, node distance = 0.65cm] {$\{a\}, \uniformizationRate$}; 
			\node[below of = g, node distance = 0.65cm] {$\{b\}, \uniformizationRate$};
			\node[below of = s2, node distance = 0.65cm] {$\{a\}, \uniformizationRate$};
			
		\end{tikzpicture}}
		\caption{A CTMC with $s \sim_{\varepsilon, 0} s'$ and $\vert \probMeasure_s(\reachability^{\leq t} g) - \probMeasure_{s'}(\reachability^{\leq t} g) \vert = 1 - e^{-\uniformizationRate t  \varepsilon}$.}
		\label{fig:example-bound-is-tight-for-epsilon-0}
	\end{figure}

	We now show that the bound of \Cref{prop:transient-prob-bounds-epsilon-delta-greater-0} is tight if $\delta = 0$.
	\begin{example}[\cite{RBBTEAPC}]\label{ex:tightness-bound-epsilon-0}
		In \Cref{fig:example-bound-is-tight-for-epsilon-0}, $s \sim_{\varepsilon ,0} s'$, $\probMeasure_{s'}(\reachability^{\leq t} g) = 0$ and  
		$\probMeasure_s(\lozenge^{\leq t} g) = 1 - e^{- \uniformizationRate  t  \varepsilon}$, so $\vert \probMeasure_s(\reachability^{\leq t} g) - \probMeasure_{s'}(\reachability^{\leq t} g) \vert = \vert 1 - e^{- \uniformizationRate  t  \varepsilon} - 0 \vert = 1 - e^{-\uniformizationRate  t  \varepsilon}$ for all $t$.
	\end{example}
	
	A disadvantage of the bounds in \Cref{prop:transient-prob-bounds-epsilon-delta-greater-0} (and \Cref{cor:transient-prob-bounds-epsilon-or-delta-0}) is that they converge to $1$ exponentially fast for $t \to \infty$, while the maximal total difference in timed reachability probabilities is usually much smaller for $(\varepsilon, \delta)$-bisimilar states. Intuitively, the bounds converge to $1$ because the bound of \Cref{thm:main-result-rbbteapc} used in their derivation always assumes the maximal possible error. However it does not consider, e.g., that for growing $t$ some of the probability mass already reached the goal state and can thus not contribute to the total error anymore. 
	This disadvantage is particularly observable for $(0, \delta)$-bisimilar states, as in this case the actual error converges to $0$ for $t\to \infty$ (see also \Cref{fig:4-state-chain,fig:example-bounds-diag-queue} in \Cref{sec:errors-0-delta-states}). 
	In \Cref{sec:errors-0-delta-states}, we derive explicit formulas and better bounds for the absolute difference in timed reachability probabilities of $(0, \delta)$-bisimilar states. 
	
	For now we consider the problem of computing, for given $\theta \in [0,1)$ and time $t > 0$, values for $\varepsilon$ and $\delta$ that guarantee  the absolute difference in reachability probabilities of $g$ until $t$ for $(\varepsilon, \delta)$-bisimilar states to be $\leq \theta$. By using the bound of \Cref{prop:transient-prob-bounds-epsilon-delta-greater-0} and solving $1 - e^{-\uniformizationRate t ( e^{\delta}(1+\varepsilon) - 1)} \leq \theta$ w.r.t. $\varepsilon$ and $\delta$ we obtain:
	\begin{restatable}{theorem}{ThmParetoCurveUnif}\label{cor:admissible-values-of-varepsilon-and-delta-for-given-theta}
		Let $\theta \in [0,1)$, $t > 0$, and $q = \max_{p \in \stateSpace} \rate(p)$. Then, for all $\varepsilon, \delta$ with $\varepsilon \in \left[0, \frac{1}{e^{\delta}} \cdot \left(\frac{\uniformizationRate \cdot t - \ln(1-\theta)}{\uniformizationRate\cdot t} \right) - 1\right]$ and $\delta \in \left[0, \ln\left(\frac{q\cdot t - \ln(1-\theta)}{(\varepsilon + 1)\cdot q \cdot t}\right)\right]$, it holds that
		\begin{center} 
			$s\sim_{\varepsilon, \delta} s'$ \, implies \, $\vert \probMeasure_s(\lozenge^{\leq t} g) - \probMeasure_{s'}(\lozenge^{\leq t} g)\vert \leq \theta$.
		\end{center}
	\end{restatable}
	
	If $\delta = 0$, the upper bound for $\varepsilon$ collapses to $\varepsilon \leq - \frac{\ln(1 - \theta)}{q \cdot t}$, while for $\varepsilon = 0$ we obtain $\delta \leq \ln \left( 1 - \frac{\ln(1-\theta)}{q \cdot t} \right)$.
	As the bound of \Cref{prop:transient-prob-bounds-epsilon-delta-greater-0} is tight if $\delta = 0$ (see \Cref{ex:tightness-bound-epsilon-0}), the range of admissible values for $\varepsilon$ in \Cref{cor:admissible-values-of-varepsilon-and-delta-for-given-theta} is tight in this case as well, in the sense that for a given $\theta \in [0,1)$ there is a CTMC $\mathcal{M}$ with states $s \sim_{\varepsilon, 0} s'$ such that $\vert \probMeasure_s(\lozenge^{\leq t} g) - \probMeasure_{s'}(\lozenge^{\leq t} g)\vert \leq \theta$ iff $\varepsilon \leq - \frac{\ln(1 - \theta)}{q \cdot t}$.
	
	\section{The Effect of Changing Rates Only} \label{sec:errors-0-delta-states}
	The goal of this section is to obtain more refined bounds for the absolute difference of timed reachability probabilities of $(0, \delta)$-bisimilar CTMCs. This setting is relevant in different scenarios like, e.g., queuing theory. Consider a CTMC $\mathcal{M}$ modeling a single server queue with a buffer of size $n$, in which tasks arrive with a constant rate $\tau$ and are served with a constant rate $\mu$. If the buffer is full and a new task arrives, it has to be dropped and an absorbing fail state $f$ is entered. The timed reachability probability $\probMeasure^{\mathcal{M}}(\lozenge^{\leq t} f)$ describes the probability of failure until $t$. Now let $\mathcal{M}'$ be like $\mathcal{M}$, but with an increased service rate $\mu' = \mu \cdot c$ for some $c > 1$. Then $\mathcal{M} \sim_{0, \delta} \mathcal{M}'$ for $\delta \geq \ln(c)$ and $\vert \probMeasure^\mathcal{M}(\lozenge^{\leq t} f) - \probMeasure^{\mathcal{M}'}(\lozenge^{\leq t} f) \vert$ describes how much higher the probability of failure until $t$ is when using the slower server instead of the faster one, making it easy to compare their performance. 
	
	We first show how to compute the absolute difference in timed reachability probabilities of $(0, \delta)$-bisimilar CTMCs $\mathcal{M}$ and $\mathcal{M}'$ exactly, without relying on the uniformization method \cite{TSMQS,MCASMP} or the Kolmogorov forward equations \cite{ÜAMW}. We then derive two bounds for the difference: an easy-to-compute constant value (that depends on $t$ and $\delta$), and a bound that yields better results if $t \to \infty$.
	
	\begin{remark}
		Throughout this section we assume that the CTMCs are uniform, i.e., that all states have the same exit rates. 
		For \emph{transitive} $(0,\delta)$-bisimulations  $R$ on non-uniformized CTMCs, the following procedure allows us to construct uniformized CTMCs for the analysis of time-bounded reachability probabilities:
		Let $\mathcal{M}$ and $\mathcal{N}$ be related by such an $R$. 
		In all equivalence classes $C$ of $R$, we change the exit rates in $\mathcal{M}$ to the minimal exit rate $\rate_{\min}^\mathcal{M}(C)$ present in $C$
		 and all rates in  $\mathcal{N}$ to $\rate_{\min}^\mathcal{M}(C)\cdot e^\delta$. As the transition probabilities of the chains are not affected by this change, and the rates of related states differ by a factor of at most $e^{\delta}$, $R$ is also a $(0,\delta)$-bisimulation on the transformed chains. Afterwards, $\mathcal{M}$ and $\mathcal{N}$ can be uniformized with rates $\uniformizationRate$ and $\uniformizationRate\cdot e^\delta$, respectively, for a suitable $\uniformizationRate$.
		This again preserves the fact that $R$ is a $(0,\delta)$-bisimulation and results in uniform CTMCs. During the procedure, the difference $\vert \probMeasure^\mathcal{M}(\lozenge^{\leq t} g) - \probMeasure^\mathcal{N}(\lozenge^{\leq t} g) \vert$ in the original CTMCs is at most as big as in the resulting uniformized CTMCs. For more details on the transformation, see \Cref{app:uniformizing}.
	\end{remark}
	
	Our formula for $\vert \probMeasure^{\mathcal{M}}(\lozenge^{\leq t} g) - \probMeasure^{\mathcal{M}'}(\lozenge^{\leq t} g) \vert$ depends on the corresponding differences in CTMCs of a specific form, which we call \emph{Erlang CTMCs} \cite{SSPTPSATE}.
	\begin{definition}\label{def:erlang-ctmc}
		Let $n \in \mathbb{N}$. The \emph{$(n$-$)$Erlang CTMC} $\erlang_n$ is a CTMC with $n$ non-goal states $s_0, \ldots, s_{n-1}$ and a unique goal state $s_n = g$, such that $\labelFunction(s_i) = \labelFunction(s_j)$ for $0 \leq i, j < n$, $\labelFunction(g) \neq \labelFunction(s_i)$ for any $i < n$, $\rate(s_i) = 1$ for $0 \leq i \leq n$ and $\prob(s_{i}, s_{i+1}) = 1$ for $0 \leq i < n$, as well as $\prob(g, g) = 1$. 
	\end{definition}
	
	\begin{figure}[t!]
		\centering
		\resizebox{!}{0.06\textheight}{
		\begin{tikzpicture}[->,>=stealth',shorten >=1pt,auto, semithick]
			\tikzstyle{every state} = [text = black]
			
			\node[state] (s0) [fill = yellow] {$s_0$}; 
			\node[state] (s1) [fill = yellow, right of = s0, node distance = 2cm] {$s_1$}; 
			\node[state] (s2) [fill = yellow, right of = s1, node distance = 2cm] {$s_2$}; 
			\node[state] (s3) [fill = yellow, right of = s2, node distance = 2cm] {$s_3$}; 
			\node[state] (g) [fill = green, right of = s3, node distance = 2cm] {$g$}; 
			
			\node[] (sinit) [left of = s0, node distance = 1.1cm] {}; 
			
			\path 
			(sinit) edge (s0)
			(s0) edge node {$1$} (s1)
			(s1) edge node {$1$} (s2)
			(s2) edge node {$1$} (s3)
			(s3) edge node {$1$} (g)
			(g) edge [loop right] node {$1$} (g)
			;
			
			\node [below of = s0, node distance = 0.65cm] {$\{a\}, 1$};
			\node [below of = s1, node distance = 0.65cm] {$\{a\}, 1$};
			\node [below of = s2, node distance = 0.65cm] {$\{a\}, 1$};
			\node [below of = s3, node distance = 0.65cm] {$\{a\}, 1$};
			\node [below of = g, node distance = 0.65cm] {$\{g\}, 1$};
		\end{tikzpicture}}
		\caption{The Erlang CTMC $\erlang_4$.}
		\label{fig:erlang-4}
	\end{figure}
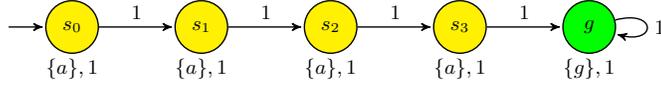

	$\erlang_4$ is illustrated in \Cref{fig:erlang-4}. Given $\mathcal{M}$ and $c \in \mathbb{R}_{> 0}$, let $c \cdot \mathcal{M}$ be the CTMC obtained from $\mathcal{M}$ by multiplying all rates with $c$. Then $\mathcal{M} \sim_{0, \delta} c \cdot \mathcal{M}$ for $\delta \geq \ln(c)$. For an absorbing, uniform CTMC $\mathcal{M}$ it is clear that an acceleration $c \cdot \mathcal{M}$ with $c = e^{\delta} \geq 1$ or a deceleration  $1/c\cdot \mathcal{M}$ induces the maximal possible difference in timed reachability probabilities compared to $\mathcal{M}$ among all $\mathcal{M}'$ with \mbox{$\mathcal{M} \sim_{0, \delta} \mathcal{M}'$}. 
	In the sequel, results are formulated for the comparison of $\mathcal{M}$ with $c\cdot \mathcal{M}$. By switching the roles of $\mathcal{M}$ and $c\cdot \mathcal{M}$, we obtain symmetric results for the deceleration.
	\begin{restatable}{proposition}{PropErrorErlangCTMC}\label{prop:error-erlang-ctmc}
		Let $n \in \mathbb{N}$, $t \geq 0$, and $\erlang'_n = c \cdot \erlang_n$ for some $c \geq 1$. Then 
		\begin{center}
			$\Diff_t(\erlang_n) \coloneqq \vert \probMeasure^{\erlang_n}(\lozenge^{\leq t} g) - \probMeasure^{\erlang'_n}(\lozenge^{\leq t} g) \vert = \sum\limits_{k = 0}^{n-1} \frac{t^k}{k!} \cdot \left(e^{-t} - c^ke^{-ct} \right)$.
		\end{center}
		If $c > 1$, $\Diff_t(\erlang_n)$ has a local maximum at $t = \frac{n \cdot \ln(c)}{c - 1}$ that is global on $[0, \infty)$. 
	\end{restatable}

	\begin{remark}\label{rem:arbitrary-rate}
		We could have also allowed arbitrary (but uniform) rates for $\erlang_n$ in \Cref{def:erlang-ctmc}, i.e., $\rate(s) = r$ for all $s \in \stateSpace$ and some $r > 0$. In this case,
		\begin{center}
			$\Diff_t(\erlang_n) = \sum\limits_{k=0}^{n-1} \frac{(r\cdot t)^k}{k!} \cdot \left(e^{-r\cdot t} - c^{k}e^{-c \cdot r\cdot t}\right)$, 
		\end{center}
		with a (local) maximum at $t^* = \frac{n \cdot \ln(c)}{r \cdot (c-1)}$ for $c > 1$. Note that $\Diff_{t^*}(\erlang_n)$ does not depend on $r$. Moreover, the assumption that $r = 1$ can be made w.l.o.g., as in the case of $r \neq 1$ the time-bounded reachability probabilities until $t$ are the same as those of the corresponding CTMC with uniform exit rate $1$ until $r \cdot t$. 
	\end{remark}
	
	\Cref{prop:error-erlang-ctmc} induces a way to exactly compute the absolute difference in timed reachability probabilities of a CTMC $\mathcal{M}$ and its acceleration $c \cdot \mathcal{M}$. 
	\begin{restatable}{theorem}{ThmExactBoundZeroDelta}\label{thm:exact-computation-error-0-delta}
		Let $\mathcal{M}$ be a CTMC with $\rate(s) = 1$ for all $s \in \stateSpace$, and let $g$ be a unique absorbing goal state. Let $\delta > 0$, $c = e^\delta > 1$, and $\mathcal{M}' = c \cdot \mathcal{M}$. Then \begin{center}
			$\Diff_t(\mathcal{M}) = \vert \probMeasure^{\mathcal{M}'}(\lozenge^{\leq t} g) - \probMeasure^{\mathcal{M}}(\lozenge^{\leq t}g) \vert = \sum\limits_{n=1}^{\infty} p_n \cdot \Diff_t(\erlang_n), $
		\end{center}
		where $p_n$ denotes the probability to reach $g$ after exactly $n$ (discrete) steps. 
	\end{restatable}
	\paragraph{Proof sketch.} Let $\Pi_n$ contain all finite paths entering $g$ after exactly $n$ steps, and let $\mathrm{Traj}^*(\pi)$ contain all timed versions of $\pi \in \Pi_n$ that enter $g$ in $(t, c \cdot t]$. Then $\Diff_t(\mathcal{M}) = \sum_{n=0}^{\infty} \sum_{\pi \in \Pi_n} \probMeasure^\mathcal{M}(\mathrm{Traj}^*(\pi)) = \sum_{n=0}^{\infty} \sum_{\pi \in \Pi_n} \mathrm{Pr}^\mathcal{M}(\mathrm{Traj}^*(\pi) \mid \pi) \cdot \mathrm{Pr}^\mathcal{M}(\pi)$ by Bayes' rule. As $\probMeasure^\mathcal{M}(\mathrm{Traj}^*(\pi) \mid \pi) = \Diff_t(\erlang_n)$ for all $\pi \in \Pi_n$, $\Diff_t(\erlang_0) = 0$ for all $t$, and $\sum_{\pi \in \Pi_n} \probMeasure^\mathcal{M}(\pi) = p_n$, the claim follows. \qed
	
	\medskip 
	
	In \Cref{thm:exact-computation-error-0-delta} (and all other results in this section) we could have also used any (uniform) exit rate $r > 0$ for the states of $\mathcal{M}$ since, as described in \Cref{rem:arbitrary-rate}, we can model different rates by scaling the time bound $t$.
	
	Next we show how to obtain, for given $t$ and $\delta$, a maximal $N$ (that depends on $t, \delta$) such that $\Diff_t(\mathcal{M}) \leq \Diff_t(\erlang_N)$. Using the explicit form of $\Diff_t(\erlang_N)$ from \Cref{prop:error-erlang-ctmc}, this yields an easy-to-compute upper bound for $\Diff_t(\mathcal{M})$. 
	\begin{restatable}{proposition}{PropImproveBoundZeroDeltaViaMaxErlang}\label{prop:erlang-bound-0-delta}
		Let $N = \left \lceil \frac{(e^{\delta} - 1) \cdot t}{\delta} \right \rceil$. Under the conditions of \Cref{thm:exact-computation-error-0-delta},
		\begin{center}
			$\Diff_t(\mathcal{M}) \leq \Diff_t(\erlang_N) = \sum\limits_{k=0}^{N-1} \frac{t^k}{k!} \cdot \left(e^{-t} - c^ke^{-ct} \right).$
		\end{center}
	\end{restatable}
	
	The bound of \Cref{prop:erlang-bound-0-delta} is usually tighter than that for $(0, \delta)$-bisimilar states in \Cref{cor:transient-prob-bounds-epsilon-or-delta-0}, since it does not converge to $1$ exponentially fast for increasing $t$. See \Cref{fig:4-state-chain,fig:example-bounds-diag-queue} for some examples of the behavior of the bound. 
	
	\begin{remark}\label{rem:adjustment-of-bound-via-Markov-inequality}
		Let $\mathcal{M}$ be as in \Cref{thm:exact-computation-error-0-delta} and let $X$ be the random variable that describes the probability to reach $g$ for the first time after exactly $n$ steps. Then $p_n \leq \probMeasure(X \leq n)$, so the Markov inequality implies $p_n \leq \frac{\mathbb{E}(X)}{n}$ for every $n$, where $\mathbb{E}(X)$ is the expected value of $X$. Together with \Cref{thm:exact-computation-error-0-delta} this yields
		\begin{center}
			$\Diff_t(\mathcal{M}) = \sum\limits_{n=1}^{\infty} p_n \cdot \Diff_t(\erlang_n) \leq 
			\mathbb{E}(X) \cdot \sum\limits_{n=1}^{\infty} \frac{1}{n} \cdot \Diff_t(\erlang_n).$ 		\end{center}
		This bound heavily depends on $\mathbb{E}(X)$, i.e., on the expected number of steps until reaching $g$. If $\mathbb{E}(X)$ is small the bound can be tighter than those of \Cref{prop:transient-prob-bounds-epsilon-delta-greater-0,prop:erlang-bound-0-delta}, while it can quickly become trivial if $\mathbb{E}(X)$ is large. 
	\end{remark}
	
	Next, we present  a way to compute the values of $p_n$ based on a spectral decomposition of the transition probability matrix $\probMatrix$ of CTMC $\mathcal{M}$. We start with the case of a  diagonalizable $\probMatrix$, and afterwards deal with arbitrary $\probMatrix$ by using the Jordan canonical form \cite{MAALA}. Moreover, we use the so obtained explicit formulas for $p_n$ to obtain upper bounds for $\Diff_t(\mathcal{M})$. In contrast to the results so far, these bounds converge to $0$ for $t \to \infty$, just like the actual difference does. 
	
	The transition probability matrix $\probMatrix$ of $\mathcal{M}$ is diagonalizable if there is a diagonal matrix $\diagonalMatrix$, with diagonal elements corresponding to the eigenvalues of $\probMatrix$ (repeated according to their multiplicities), and a regular matrix $\symmetryMatrix$ such that $\probMatrix = \symmetryMatrix \diagonalMatrix \symmetryMatrix^{-1}$ \cite{MAALA}. It is well-known that every eigenvalue $\lambda_i$ of a stochastic matrix satisfies $\vert \lambda_i \vert \leq 1$, and that $\lambda_1 = 1$ is an eigenvalue of any such matrix. 
	As we assume that an absorbing state is reached almost surely, $1$ is the only eigenvalue of modulus $1$, and it has a multiplicity $a_\probMatrix$ of at most $2$ (as there are at most two absorbing states in $\mathcal{M}$).
	Hence, we can w.l.o.g. assume that for the $m$ distinct eigenvalues of $\probMatrix$ we have $\lambda_1 = 1 > \vert \lambda_2 \vert \geq \ldots \geq \vert \lambda_m \vert$, and that the diagonal of $\diagonalMatrix$ is in descending order w.r.t. the absolute values of these eigenvalues. From now on we denote the second largest (absolute value wise) eigenvalue of $\probMatrix$ by $\lambda$.

	\begin{restatable}[\cite{MCASAGE}]{proposition}{PropExactComputationPnDiag}\label{prop:exact-computation-pn-diag}
		Let $\mathcal{M}$ be a CTMC such that $\probMatrix = \symmetryMatrix \diagonalMatrix\symmetryMatrix^{-1}$ is diagonalizable. Let $k \in \mathbb{N}$ and $n = \vert \stateSpace \vert$. Then $p_{k+1} = \sum_{j=a_\probMatrix+1}^{n} \symmetryMatrix_{1, j} \cdot \symmetryMatrix^{-1}_{j, n} \cdot (\lambda_j -1) \cdot \lambda_j^k.$
	\end{restatable}
	
	Combining \Cref{thm:exact-computation-error-0-delta} and \Cref{prop:exact-computation-pn-diag} yields a new bound for $\Diff_t(\mathcal{M})$. \unskip
	\begin{restatable}{proposition}{BoundPnDiag}\label{prop:improved-bounds-0delta-with-ev-diagonalizable}
		Under the conditions of \Cref{thm:exact-computation-error-0-delta} we have 
		\begin{center}
			$\Diff_t(\mathcal{M}) \leq (n-a_\probMatrix) \cdot C \cdot \sum_{k=1}^{\infty} \vert \lambda \vert^{k-1} \cdot \Diff_t(\erlang_k),$
		\end{center}
		where $n = \vert \stateSpace \vert$ and $C =  \max_{i = a_\probMatrix+1, \ldots, n} \vert \symmetryMatrix_{1,i} \cdot \symmetryMatrix^{-1}_{i,n} \cdot (\lambda_i - 1) \vert$. 
	\end{restatable}
	\paragraph{Proof sketch.} By \Cref{prop:exact-computation-pn-diag}, $p_{k+1} = \sum_{j=a_\probMatrix+1}^{n} \symmetryMatrix_{1, j} \cdot \symmetryMatrix^{-1}_{j, n} \cdot (\lambda_j -1) \cdot \lambda_j^k$, so by the triangle inequality $p_{k+1} \leq \sum_{j=a_\probMatrix+1}^{n} \vert \symmetryMatrix_{1, j} \cdot \symmetryMatrix^{-1}_{j, n} \cdot (\lambda_j -1) \vert \cdot \vert \lambda_j \vert^k $. Because $\vert \lambda_j \vert \leq \vert \lambda \vert$ for all $j > a_\probMatrix$ we get $p_{k+1} \leq \vert \lambda \vert^{k} \cdot \sum_{j=a_\probMatrix+1}^{n} C = \vert \lambda \vert^{k} \cdot (n-a_\probMatrix) \cdot C$ for $C = \max_{i = a_\probMatrix+1, \ldots, n} \vert \symmetryMatrix_{1,i} \cdot \symmetryMatrix^{-1}_{i,n} \cdot (\lambda_i - 1) \vert$. The claim follows from \Cref{thm:exact-computation-error-0-delta}.  \qed
	
	\medskip
	
	The bound obtained in \Cref{prop:improved-bounds-0delta-with-ev-diagonalizable} is tight: Consider a CTMC $\mathcal{M}$ with two states $s$ and $g$, such that $\rate(s) = \rate(g) = 1$ and $\prob(s, s) = p, \prob(s,g) = 1-p$ for a $p \in (0,1)$. Let $\mathcal{M}' = e^{\delta} \cdot \mathcal{M}$, so $\mathcal{M} \sim_{0, \delta} \mathcal{M}'$. The probability matrix $\probMatrix$ of $\mathcal{M}$ is diagonalizable with eigenvalues $\lambda_1= 1$ and $\lambda_2 = \lambda = p$. Then $p_k =\vert \sum_{j=2}^{2} \symmetryMatrix_{1, j} \cdot \symmetryMatrix^{-1}_{j, 2} \cdot (\lambda_j -1) \cdot \lambda_j^{k-1} \vert = (1-p) \cdot p^{k-1} = (n-1) \cdot C \cdot \lambda^{k-1}$, so here the inequalities in the proof  of \Cref{prop:improved-bounds-0delta-with-ev-diagonalizable} are  equalities. 
	Furthermore, since always $\vert \lambda \vert < 1$, the bounds converge to $0$ for $t \to \infty$, just like $\Diff_t(\mathcal{M})$.
	
	\begin{figure}[t]
		\begin{subfigure}{0.14\textwidth}
			\centering
			\begin{tikzpicture}[ ->,>=stealth',shorten >=1pt,auto, semithick]
				\tikzstyle{every state} = [text = black]
				
				\node[state, scale = 0.7] (s0) [] {$s_0$}; 
				\node[scale = 0.7] (temp) [below of = s0, node distance = 2cm] {};
				\node[state, scale = 0.7] (s1) [left of = temp, node distance = 0.8cm] {$s_1$}; 
				
				\node[state, scale = 0.7] (s2) [right of = temp, node distance = 0.8cm] {$s_2$}; 
				
				\node[state, scale = 0.7] (g) [below of = temp, node distance = 2cm, fill = green] {$g$}; 
				
				\node[, scale = 0.7] (sinit) [left of = s0, node distance = 1.1cm] {}; 
				\node[, scale = 0.7] () [below of = g, node distance = 0.7cm] {}; 
				
				\path 
				(sinit) edge (s0)
				
				(s0) edge [loop right] node [scale = 0.8] {$\frac{1}{4}$} (s0)
				(s0) edge [bend right = 20] node [scale = 0.8, left] {$\frac{1}{4}$} (s1)
				(s0) edge [bend left = 20] node [scale = 0.8, right] {$\frac{1}{2}$} (s2)
				
				(s2) edge node [scale = 0.8, pos = 0.4, above] {$\frac{1}{2}$} (s1)
				(s2) edge [bend left = 20] node [scale = 0.8, right] {$\frac{1}{2}$} (g)
				
				(s1) edge [bend right = 20] node [scale = 0.8, left] {$\frac{1}{2}$} (g)
				(s1) edge [loop left] node [scale = 0.8, below, pos = 0.15] {$\frac{1}{2}$} (s1)
				
				(g) edge [loop left] node [scale = 0.8] {$1$} (g)
				;
				
				
				\node[below of = s0, node distance = 0.45cm, scale = 0.8] {$\emptyset, 1$};
				\node[right of = s1, node distance = 0.4cm, scale = 0.8, yshift = -0.4cm] {$\emptyset, 1$};
				\node[right of = s2, node distance = 0.5cm, scale = 0.8] {$\emptyset, 1$};
				\node[right of = g, node distance = 0.65cm, scale = 0.8] {$\{g\}, 1$};
			\end{tikzpicture}
		\end{subfigure}
		\begin{subfigure}{0.85\textwidth}
			
			{\hspace{1.4cm}
				\includegraphics[width=0.83\textwidth]{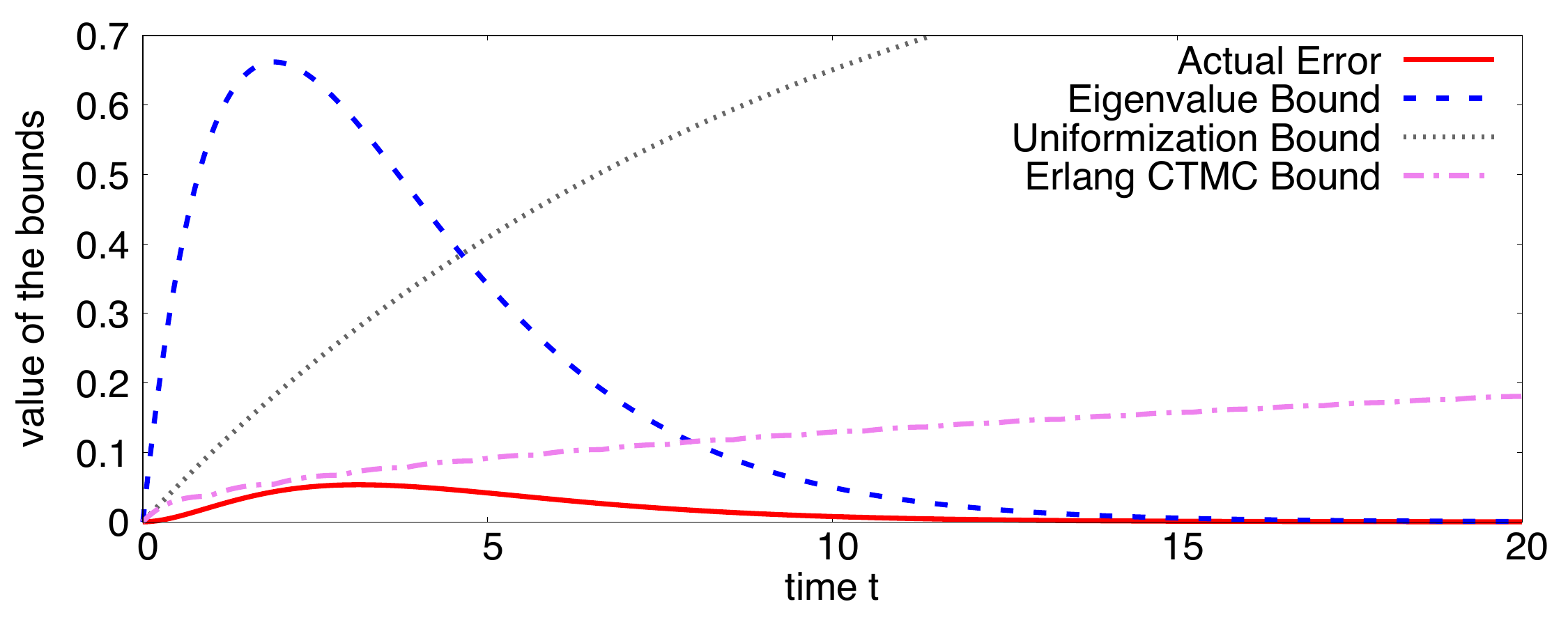}
			}
		\end{subfigure}
		\caption{A CTMC $\mathcal{M}$ with diagonalizable probability matrix (left) and a comparison of the different error bounds for $\Diff_t(\mathcal{M})$.}
		\label{fig:4-state-chain}
	\end{figure}
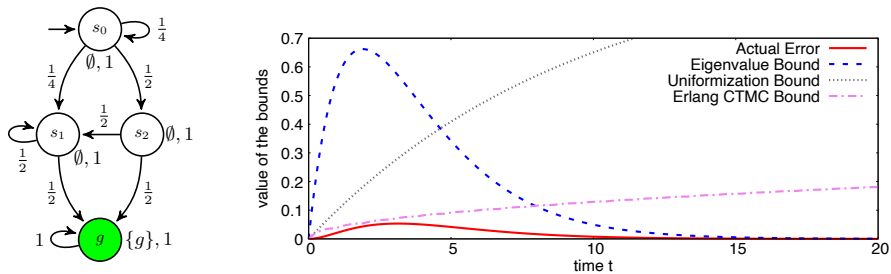
	\begin{example}\label{ex:4-state-chain}
		Consider the CTMC $\mathcal{M}$ on the left of \Cref{fig:4-state-chain}. The probability matrix $\probMatrix$ of $\mathcal{M}$ is diagonalizable, with second largest eigenvalue $\lambda = \frac{1}{2}$. The graphic on the right of the figure compares the different error bounds for $\Diff_t(\mathcal{M})$ when $\delta = 0.1$. We can observe that the bound from \Cref{prop:transient-prob-bounds-epsilon-delta-greater-0} based on the uniformization method quickly approaches $1$. The bound from \Cref{prop:erlang-bound-0-delta} based on the maximal value of $\Diff_t(\erlang_n)$ performs well if $t$ is small, but becomes imprecise if $t$ grows. On the other hand, the bound from \Cref{prop:improved-bounds-0delta-with-ev-diagonalizable} that utilizes the diagonalization of $\probMatrix$ shows the opposite behavior. It performs bad for small values of $t$ (as here the values of $\Diff(\erlang_k)$ are big for small $k$, for which $\vert \lambda\vert^k$ is not yet close to $0$), while it converges to the actual error with increasing $t$. 
	\end{example}
	
	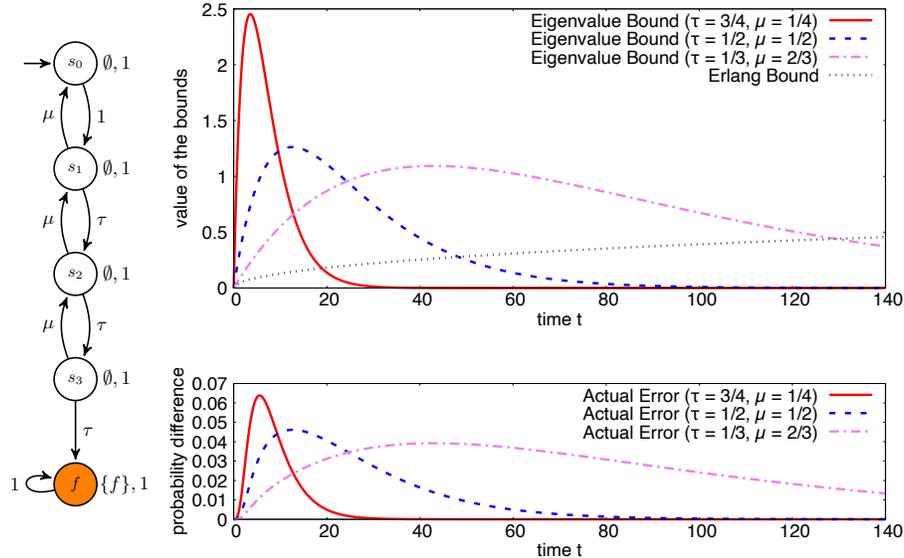
\begin{figure}[t]
		\begin{subfigure}{0.1\textwidth}
			\centering
			\begin{tikzpicture}[ ->,>=stealth',shorten >=1pt,auto, semithick]
				\tikzstyle{every state} = [text = black]
				
				\node[state, scale = 0.7] (s0) [] {$s_0$}; 
				\node[state, scale = 0.7] (s1) [below of = s0, node distance = 2cm] {$s_1$}; 
				 
				 \node[state, scale = 0.7] (s2) [below of = s1, node distance = 2cm] {$s_2$}; 
				 
				 \node[state, scale = 0.7] (s3) [below of = s2, node distance = 2cm] {$s_3$}; 
				 \node[state, scale = 0.7] (f) [fill = orange, below of = s3, node distance = 2cm] {$f$}; 
				
				\node[, scale = 0.7] (sinit) [left of = s0, node distance = 1.1cm] {}; 
				
				\node[scale = 0.7]  [below of = f, node distance = 1.5cm] {}; 
				
				\path 
				(sinit) edge (s0)
				
				(s0) edge [bend left = 20] node [scale = 0.8] {$1$} (s1)
				
				(s1) edge [bend left = 20] node [scale = 0.8] {$\mu$} (s0)
				(s1) edge [bend left = 20] node [scale = 0.8] {$\tau$} (s2)
				
				(s2) edge [bend left = 20] node [scale = 0.8] {$\mu$} (s1)
				(s2) edge [bend left = 20] node [scale = 0.8] {$\tau$} (s3)
				
				(s3) edge [bend left = 20] node [scale = 0.8] {$\mu$} (s2)
				(s3) edge node [scale = 0.8] {$\tau$} (f)
				
				(f) edge [loop left] node [scale = 0.8] {$1$} (f)
				;
				
				
				\node[right of = s0, node distance = 0.55cm, scale = 0.8] {$\emptyset, 1$};
				\node[right of = s1, node distance = 0.55cm, scale = 0.8] {$\emptyset, 1$};
				\node[right of = s2, node distance = 0.55cm, scale = 0.8] {$\emptyset, 1$};
				\node[right of = s3, node distance = 0.55cm, scale = 0.8] {$\emptyset, 1$};
				\node[right of = f, node distance = 0.65cm, scale = 0.8] {$\{f\}, 1$};
			\end{tikzpicture}
			\vfill
		\end{subfigure}
		\begin{subfigure}{0.95\textwidth}
			
			\begin{center}\includegraphics[width=0.85\linewidth]{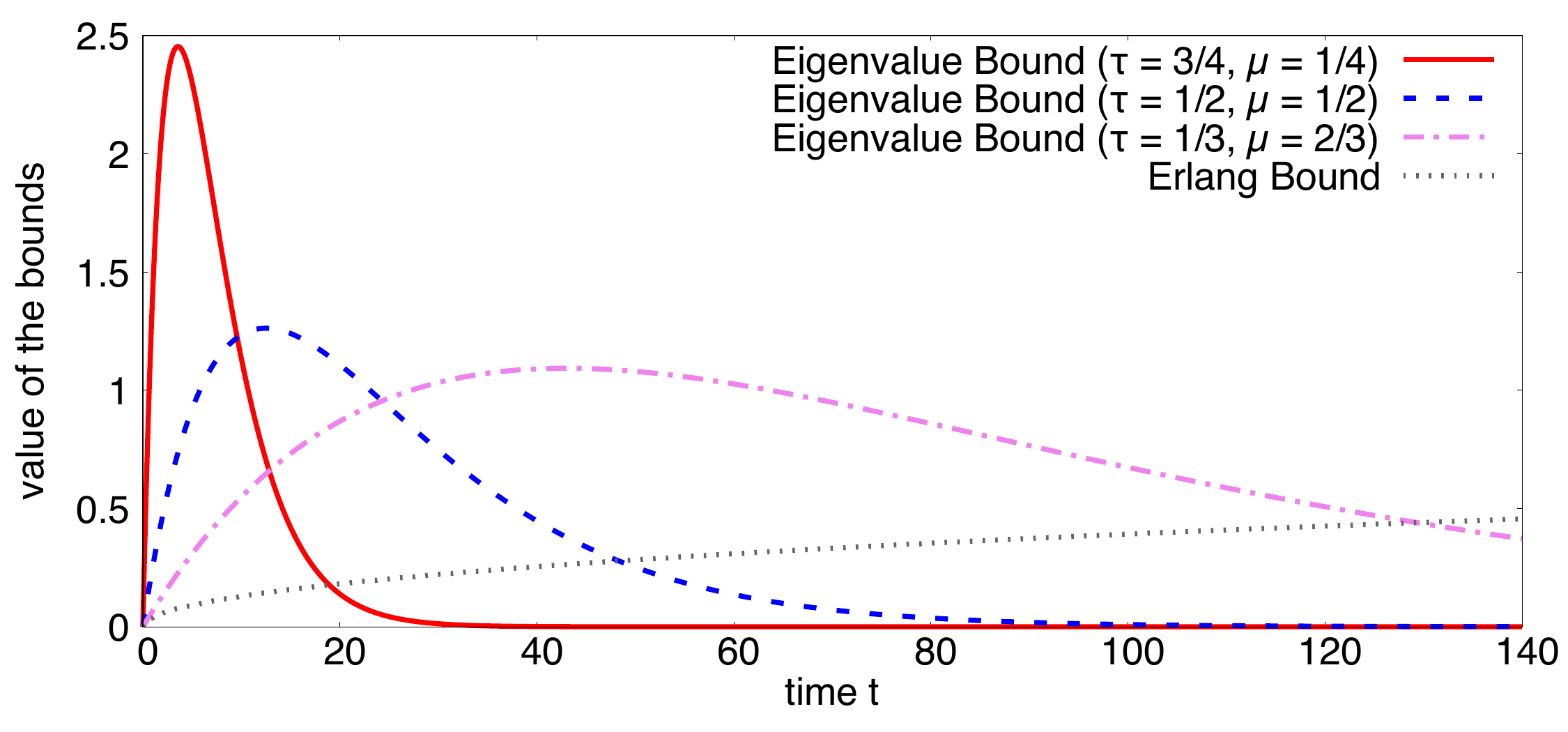}
		
		 \end{center}
			
			\begin{center}\includegraphics[width=0.85\linewidth]{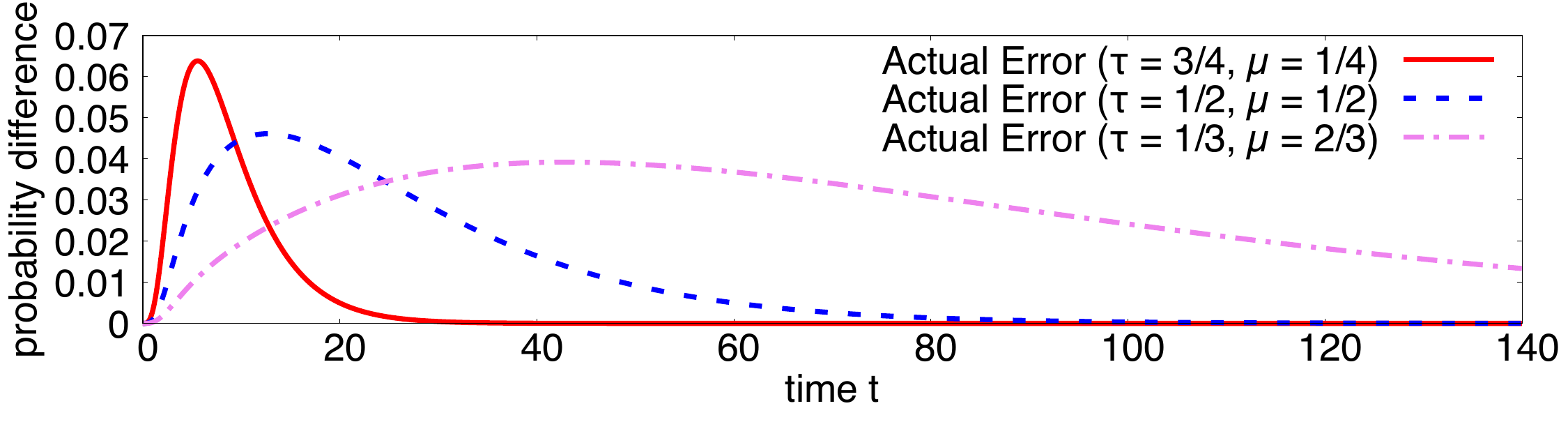} \end{center}
		\end{subfigure}
		\caption{The queue from \Cref{ex:bounds-diag} (left) and the bounds and errors for specific instances (right). The upper graphic depicts the bounds from \Cref{prop:erlang-bound-0-delta,prop:improved-bounds-0delta-with-ev-diagonalizable} for the considered scenarios, the lower one the actual errors.}
		\label{fig:example-bounds-diag-queue}
	\end{figure}
		
\begin{example}\label{ex:bounds-diag}
	Consider a single-server queue with capacity $4$, modeled by the CTMC $\mathcal{M}$ on the left of \Cref{fig:example-bounds-diag-queue}. Tasks are completed with rate $\mu$. If the queue is empty, the next task arrives with rate $1$, and otherwise tasks arrive with rate $\tau$. We assume that $\tau + \mu = 1$. If the queue is full, i.e., if it already contains four tasks, and a fifth one arrives it has to be dropped, leading to a fail-state $f$. This setting is close to the example discussed in the beginning of this section.
	
	Let $\mathcal{M}' = e^{0.1} \cdot \mathcal{M}$. In particular, $\mathcal{M} \sim_{0, 0.1} \mathcal{M}'$.
	We analyze the bounds from \Cref{prop:erlang-bound-0-delta,prop:improved-bounds-0delta-with-ev-diagonalizable} w.r.t. $\mathcal{M}, \mathcal{M}'$ and $f$ by considering three scenarios that are parameterized in $\tau$ and $\mu$.\footnote{We omit the bound based on uniformization from \Cref{prop:transient-prob-bounds-epsilon-delta-greater-0} as we have already seen in \Cref{ex:4-state-chain} that it does not perform well for $(0, \delta)$-bisimilar models.} In the first scenario, $\tau_1 = \frac{3}{4}$ and $\mu_1 = \frac{1}{4}$, i.e., the arrival rate of tasks is three times higher than their service rate. In the second scenario, arrival and service rates coincide ($\tau_2 = \frac{1}{2} = \mu_2$), while in the third scenario the service rate is twice the arrival rate, i.e., $\tau_3 = \frac{1}{3}$  and $\mu_3 = \frac{2}{3}$. We denote the CTMC corresponding to scenario $i \in \{1,2,3\}$ by $\mathcal{M}_i$. Note that for each $i$ the transition probability matrix $\probMatrix_i$ of $\mathcal{M}_i$ is diagonalizable.
	
	The resulting bounds, as well as the actual errors, are depicted on the right-hand side of \Cref{fig:example-bounds-diag-queue}, where the same-colored (and similarly dashed) lines correspond to the same scenario (red and solid for the first, blue and dashed for the second and magenta and dash-dotted for the third). The dotted line in the upper graphic represents the bound from  \Cref{prop:erlang-bound-0-delta}.
	
	We observe that in each scenario the actual error and the  bounds of \Cref{prop:improved-bounds-0delta-with-ev-diagonalizable} converge to $0$ for $t \to \infty$, while the bound of \Cref{prop:erlang-bound-0-delta} converges to $1$ for increasing $t$. The convergence rate to $0$ of the former bound can, however, be quite slow. In, e.g., the third scenario, the second largest eigenvalue of $\probMatrix_3$ has modulus $0.9778$, causing a slow convergence of $\vert \lambda \vert^k$ to $0$. On the other hand, in the first scenario $\mathcal{M}_1$ enters $f$ quickly, which is reflected in a second largest eigenvalue of $0.7334$ whose powers converge to $0$ fast.
	
	While the values of the bounds from \Cref{prop:improved-bounds-0delta-with-ev-diagonalizable} are, in particular for small $t$, significantly larger than the actual error, it is promising that their shape as a function of $t$ closely resembles the actual difference.
	This scaling of the bounds by a large factor is caused by the constant $C$, which is obtained by using the triangle inequality and a maximum.
	Refining this constant in future work might lead to bounds following the actual error much more closely. 
\end{example}

	 As the presented bounds work differently well for small and large $t$, the following formulation is useful.
	\begin{corollary}\label{cor:minimum-of-bounds-diag}
		Let $C$ be as in \Cref{prop:improved-bounds-0delta-with-ev-diagonalizable}, and $N$ as in \Cref{prop:erlang-bound-0-delta}. Then
		\begin{center}
			$\Diff_t(\mathcal{M}) \leq \min \left\{\Diff_t(\erlang_N), (n-a_\probMatrix) \cdot C \cdot \sum_{k=1}^{\infty} \vert \lambda \vert^{k-1} \cdot \Diff_t(\erlang_k)\right\}.$
		\end{center}
	\end{corollary}
	
	\Cref{prop:improved-bounds-0delta-with-ev-diagonalizable}  requires the transition probability matrix $\probMatrix$ of $\mathcal{M}$ to be diagonalizable. 
	If this is not the case, we can use the \mbox{\emph{Jordan canonical form} (\emph{JCF}) \cite{MAALA}} instead, which looks as follows:
	Given an $n \times n$ matrix $M$ over the field of complex numbers $\mathbb{C}$ with distinct eigenvalues $\lambda_1, \ldots, \lambda_m$, given in descending order w.r.t. their absolutes values, the JCF of $M$ is a decomposition of the form $M = \symmetryMatrix \jordanMatrix\symmetryMatrix^{-1}$ for matrices $\symmetryMatrix, \jordanMatrix \in \mathbb{C}^{n \times n}$. $\jordanMatrix$ is a block diagonal matrix
	with so-called \emph{Jordan blocks} $\jordanMatrix_i$, $i = 1, \ldots, k$, $m \leq k \leq n$, on the diagonal. The $i$-th Jordan block $\jordanMatrix_i$ corresponds to one of $M$'s eigenvalues $\tau$ and is a $r_i \times r_i$ matrix for some $r_i$ that is at most the algebraic multiplicity of $\tau$, and such that the sum over all $r_i$ for the Jordan blocks belonging to the same eigenvalue $\tau$ equals its geometric multiplicity. The Jordan blocks
	for eigenvalue $\tau$ have $\tau$ everywhere on the diagonal and entry $1$ directly above the diagonal. All other entries are $0$. We w.l.o.g. assume that the Jordan blocks are in descending order w.r.t. the absolute value of the corresponding eigenvalues, and that multiple Jordan blocks for the same eigenvalue are in descending order w.r.t. their size. We write $q_\jordanMatrix$ for the total number of Jordan blocks of $\jordanMatrix$ and $r_i$ for the size of the $i$-th Jordan block $\jordanMatrix_i$.
		
		It is well-known that every matrix over $\mathbb{C}$ has a JCF \cite{MAALA}, so the above decomposition exists for the transition probability matrix of \emph{any} CTMC. We now show how to compute $p_n$, i.e., the probability to reach the goal state after exactly $n$ discrete steps, in the case that $\probMatrix$ is not diagonalizable.
	\begin{restatable}{proposition}{ThmExactValuePnNonDiagonalizable}\label{thm:exact-value-pn-non-diagonalizable}
		Denote by $\lambda_i$ the eigenvalue corresponding to the $i$-th Jordan block $\jordanMatrix_i$ of $\probMatrix = \symmetryMatrix \jordanMatrix\symmetryMatrix^{-1}$. Let $N \in \mathbb{N}$, $z \in \{0, \ldots, q_{\jordanMatrix}-1\}$ the number of Jordan blocks for eigenvalue $0$ and $h_l = \sum_{i = 1}^{l-1} r_{i}$. Then $p_{N+1}$ equals
		\begin{align*}
			\sum_{l = a_\probMatrix + 1}^{q_{\jordanMatrix}-z} \sum_{j = 1}^{r_{l}} \sum_{k = 1}^{j} \symmetryMatrix_{1, k + h_l} \cdot \symmetryMatrix^{-1}_{j + h_l, n} \cdot \lambda_{l}^{N + k - j} \cdot \left(\lambda_{l} \cdot \binom{N+1}{j-k} - \binom{N}{j-k} \right) + R(N, z)
		\end{align*}
		where, for $\symmetryMatrix_{1, h_l + j}^* = \symmetryMatrix_{1, h_l, +j}$ if $j > 0$ and  $\symmetryMatrix_{1, h_l + j}^* = 0$ if $j = 0$,
		\begin{center}
			$R(N, z) = \sum_{l = q_{\jordanMatrix}-z+1}^{q_{\jordanMatrix}} \sum_{j=0}^{r_{l} - (N + 1)} (\symmetryMatrix_{1, h_l + j}^* - \symmetryMatrix_{1, h_l + j + 1}) \cdot \symmetryMatrix^{-1}_{h_l + N + j, n}$.
		\end{center} 
	\end{restatable}
	
	In particular, the term $R(N, z)$ introduced in \Cref{thm:exact-value-pn-non-diagonalizable} equals $0$ if $z = 0$ or if $N \geq \max_{q_{\jordanMatrix}-z+1 \leq i \leq q_{\jordanMatrix}} r_i$. Hence, if $N$ is at least the size of the largest Jordan block corresponding to eigenvalue $0$, the term $R(N, z)$ vanishes.
	
	\begin{remark}
	\label{rem:multiple absorbing states}
		The decompositions of $\probMatrix$ proved in \Cref{prop:exact-computation-pn-diag,thm:exact-value-pn-non-diagonalizable} also hold when the Markov chain contains more than $2$ (reachable) absorbing states, i.e, if $a_\probMatrix > 2$. In an absorbing Markov chain, the multiplicity (both geometric and algebraic) of the eigenvalue $1$ is  the number of absorbing states, and thus each Jordan block for eigenvalue $1$ has a size of $1$, causing them to get canceled out in the computations of the explicit forms of $p_{N+1}$. 
	\end{remark}
	
	Combining \Cref{thm:exact-computation-error-0-delta} and \Cref{thm:exact-value-pn-non-diagonalizable} yields another upper bound on the absolute difference in timed reachability probabilities of $(0, \delta)$-bisimilar CTMCs.
	
	\begin{restatable}{theorem}{PropJordanBound}\label{prop:error-bound-jordan}
		In the setting of \Cref{thm:exact-value-pn-non-diagonalizable} let $\vert \lambda \vert > 0$. Let $R$ be the maximal size of any Jordan block of $\probMatrix$, $r = \max_{i = a_\probMatrix+1, \ldots, q_\jordanMatrix-z} r_i$ the maximal size of any such block corresponding to an eigenvalue $\neq 0, 1$, $\lambda_{l}^* =\max\{\vert \lambda_{l} \vert,\vert 1 - \lambda_{l} \vert\}$ for $l = a_\probMatrix+1, \ldots, q_{\jordanMatrix}-z$ and $C = \sum_{l = a_\probMatrix+1}^{q_{\jordanMatrix}-z} \sum_{j = 1}^{r_{l}} \sum_{k = 1}^{j}  \vert \symmetryMatrix_{1, k + h_l} \cdot \symmetryMatrix^{-1}_{j + h_l, n} \vert \cdot \vert \lambda_{l}^* \vert$. Then 
		\begin{center}
			$\Diff_t(\mathcal{M}) \leq \sum_{k = 1}^{R-1} p_k \cdot \Diff_t(\erlang_k) + C \cdot  \sum_{k = R}^{\infty} \vert \lambda\vert^{k - r} \cdot  k^{r - 1} \cdot \Diff_t(\mathcal{E}_k).$
		\end{center}
	\end{restatable}
	
	The bound of \Cref{prop:error-bound-jordan} is a conservative extension of the one for diagonalizable $\probMatrix$ from \Cref{prop:improved-bounds-0delta-with-ev-diagonalizable}, as in this case $R = 1 = r$ and so the bound of \Cref{prop:error-bound-jordan} collapses to 
	$\Diff_t(\mathcal{M}) \leq C \cdot \sum_{k=1}^{\infty} \vert \lambda \vert^{k-1} \cdot \Diff_t(\erlang_k)$.

	\begin{remark}
		In \Cref{prop:error-bound-jordan} we have to exclude the case $\vert \lambda \vert = 0$, i.e., the case that $\probMatrix$ only has eigenvalues $0$ and $1$. This happens, e.g., if $\probMatrix$ (and hence $\mathcal{M}$) is acyclic. In an acyclic CTMC the goal state $g$ is, however, reached after \emph{at most} as many steps as the length of the longest path in the chain. Therefore, if $\mathcal{M}$ is acyclic, $\Diff_t(\mathcal{M})$ can easily be computed exactly, for example by using the identity in \Cref{thm:exact-computation-error-0-delta} and computing the relevant values of $p_{n} = \probMatrix^{n+1} - \probMatrix^n$, or by applying methods like the ACE algorithm of \cite{TAAMC} that allow computing the transient distribution functions of acyclic CTMCs directly.
	\end{remark}

	\begin{corollary}
		Under the conditions and with the notation of \Cref{prop:error-bound-jordan} and \Cref{prop:erlang-bound-0-delta}, $\Diff_t(\mathcal{M})$ is bounded from above by 
		\begin{center}
			$\min\left\{\Diff_t(\erlang_N), \sum_{k = 1}^{R-1} p_k \cdot \Diff_t(\erlang_k) + C \cdot  \sum_{k = R}^{\infty} \vert \lambda \vert^{k - r}  \cdot k^{r - 1} \cdot \Diff_t(\mathcal{E}_k) \right\}.$
		\end{center}
	\end{corollary}
	
	We finish the section by noting that the bounds for $(0, \delta)$-bisimilar CTMCs can also be used to derive bounds for the case that $\varepsilon > 0$. Given $\mathcal{M} \sim_{\varepsilon, \delta} \mathcal{N}$, by \Cref{thm:same-structure-no-additional-tolerance} there are CTMCs $\mathcal{M}'$ and $\mathcal{N}'$ with $\mathcal{M} \sim_{\varepsilon, 0} \mathcal{M}' \sim_{0, \delta} \mathcal{N}' \sim \mathcal{N}$. Hence, the triangle inequality together with the fact that strong bisimilarity preserves timed reachability probabilities \cite{MCACTMC,EOLFMC} implies
	$\vert \probMeasure^{\mathcal{M}}(\lozenge^{\leq t} g) - \probMeasure^{\mathcal{N}}(\lozenge^{\leq t} g) \vert \leq \vert \probMeasure^{\mathcal{M}}(\lozenge^{\leq t} g) - \probMeasure^{\mathcal{M}'} (\lozenge^{\leq t} g) \vert + \vert \probMeasure^{\mathcal{M}'}(\lozenge^{\leq t} g) - \probMeasure^{\mathcal{N}'}(\lozenge^{\leq t} g) \vert$. As $\mathcal{M} \sim_{\varepsilon, 0} \mathcal{M}'$, the bound from \Cref{prop:transient-prob-bounds-epsilon-delta-greater-0} for the special case of $\varepsilon = 0$ is applicable to the first term, and one of the bounds for $(0, \delta)$-bisimilar chains can be applied to the second term. This yields another upper bound on the absolute difference in timed reachability probabilities for $(\varepsilon, \delta)$-bisimilar CTMCs.

	\section{Bounds for Reward-Bounded Reachability Probabilities}\label{sec:reward-bounds}
	In practice, CTMCs are often extended with \emph{rewards} that allow to model, e.g., the accumulation of costs, the energy consumption, or different performance measures like the availability of the system \cite{LCPP,MRMMDPDCTPEO}. We consider state-based reward functions $\reward \colon \stateSpace \to \mathbb{R}$ that assign to every $s \in \stateSpace$ a reward $\rho(s)$. Rewards are accumulated in the states, and the cumulative reward along the timed path $\pi =s_0 t_0 s_1 t_1 \ldots \in \Paths(\mathcal{M})$ of $\mathcal{M}$ until time $t$ is given (for $\sigma @ t = s_m$) as \cite{LCPP}
	\begin{center}
		$\reward(\sigma, t) = \sum_{j=0}^{m-1} t_j \cdot \rho(s_j) + \left(t - \sum_{j = 0}^{m-1} t_j \right) \cdot \rho(s_m).$
	\end{center} 

	Similar to timed reachability probabilities $\probMeasure_s(\reachability^{\leq t} g)$, logics like CSRL \cite{LCPP} allow to specify properties in terms of \emph{reward-bounded} reachability probabilities $\probMeasure_s(\reachability_{\leq r} g)$, which denote the probability to reach, from state $s$, a goal state $g$ while accumulating at most reward $r \in \mathbb{R}$. 
	
	We extend the results of \Cref{sec:time-bounded-reachability-bounds,sec:errors-0-delta-states} to reward-bounded reachability probabilities. Here, we consider the special case of \emph{nonnegative} rewards, i.e., we assume that $\reward(s) \geq 0$ for all $s \in \stateSpace$. To accommodate for the addition of rewards to our model, we adjust the definition of $(\varepsilon, \delta)$-bisimulations to require related states to have the same reward, i.e., $s \sim_{\varepsilon, \delta} s'$ is possible only if $\reward(s) = \reward(s')$. An extension to a notion of, say, $(\varepsilon, \delta, \mu)$-bisimulation, where $\mu$ describes an allowed error in the rewards  in related states, is an interesting direction for future work. 
	
	First, assume $\rho(s) > 0$ for all $s$, a setting also analyzed in, e.g., \cite{LCPP}. There, based on ideas from \cite{PRRMCS}, a transformation is proposed that allows to compute reward-bounded reachability probabilities in $\mathcal{M}$ via \emph{time-bounded} reachability probabilities in a modified model $\widehat{\mathcal{M}}$. This model differs from $\mathcal{M}$ only in the exit rates $\widehat{\rate}(s) = \frac{\rate(s)}{\reward(s)}$ and rewards $\widehat{\reward}(s) = \frac{1}{\reward(s)}$ of states $s \in S$. A direct consequence of \cite[Lem. 1]{LCPP} is that 
$		\probMeasure_s^\mathcal{M}(\reachability_{\leq r} g) = \probMeasure_s^{\widehat{\mathcal{M}}}(\reachability^{\leq r} g)$.	

Thus, given $s \sim_{\varepsilon, \delta} s'$ we can directly apply the bounds obtained in \Cref{sec:time-bounded-reachability-bounds,sec:errors-0-delta-states}, provided that $s$ and $s'$ are still $(\varepsilon, \delta)$-bisimilar in $\widehat{\mathcal{M}}$. However, this is indeed the case as we require $\reward(s) = \reward(s')$, and so the transformation of \cite{LCPP} simply scales the exit rates (and rewards) of $s$ and $s'$ by a common factor.
	\begin{proposition}\label{prop:reward-bounds-all-positive}
		Let $s \sim_{\varepsilon, \delta} s'$, $r \geq 0$ and $\widehat{\uniformizationRate} \geq \max_{p \in S} \widehat{\rate}(p)$. Then 
		\begin{center}$\vert \probMeasure_s(\reachability_{\leq r} g) - \probMeasure_{s'}(\reachability_{\leq r} g)\vert \leq 1 - e^{-\widehat{\uniformizationRate \cdot }t\cdot(e^{\delta}(\varepsilon + 1)- 1)}$. \end{center}
	\end{proposition}
	
	If $\varepsilon = 0$, a similar result can be formulated for the bounds from \Cref{sec:errors-0-delta-states}.

	Next, we  allow $\rho(s) = 0$. As the reward of the goal state $g$ does not matter for our purpose, we can w.l.o.g. assume $\rho(g) > 0$. The transformation from \cite{LCPP} is not applicable, as it could require division by $0$. To circumvent this issue we first construct from $\mathcal{M}$ a CTMC $\mathcal{M}_{>0}$ without states with $\rho(s) = 0$. 
	Intuitively, this is done by redirecting the incoming transitions of states with reward $0$ to their successors, allowing us to remove any such state from the chain.

	\begin{restatable}{proposition}{PropRemoveRewardZeroStates}\label{prop:remove-reward-0-states}
		Let $\mathcal{M}$ be a CTMC. There is a CTMC $\mathcal{M}_{>0}$ with $\stateSpace^{\mathcal{M}_{>0}} = \stateSpace^\mathcal{M} \setminus \{x \mid \rho(x) = 0\}$ such that  $\probMeasure^\mathcal{M}_s(\reachability_{\leq r} g) = \probMeasure^{\mathcal{M}_{> 0}}_s(\reachability_{\leq r}g)$ for all $s \in \stateSpace^{\mathcal{M}_{>0}}$.
	\end{restatable}
	
	 Hence, by first constructing $\mathcal{M}_{>0}$ from $\mathcal{M}$ and afterwards applying the transformation from \cite{LCPP} to $\mathcal{M}_{>0}$, we obtain the following corollary.
	 \begin{corollary}
	 	\Cref{prop:reward-bounds-all-positive} also holds for CTMCs with nonnegative rewards.
	 \end{corollary}
	\section{Conclusion and Future Work}\label{sec:conclusion}
	We introduced $(\varepsilon, \delta)$-bisimulations, a novel type of approximate probabilistic bisimulation for CTMCs. 
	In contrast to notions from the literature such as quasi-lumpability \cite{EOLFMC,BQLMC,CBPIQLSWFN}, the separate bounds on changes in transition probabilities and changes in exit rates used in $(\varepsilon,\delta)$-bisimulation result in a 
	 more flexible, fine-grained notion of behavioral similarity.
	We established fundamental properties of the notion, and analyzed the difference in timed reachability probabilities of $(\varepsilon, \delta)$-bisimilar chains. 
	
	We obtained  bounds on this difference by uniformizing the chain and applying known results from \cite{RBBTEAPC} for the discrete-time setting. 
	Although tight if $\delta = 0$, the bounds usually do not perform well if $\delta > 0$. Hence, we investigated the special case of $(0, \delta)$-bisimilarity in more detail---also as a stepping stone towards a refined treatment of $(\varepsilon,\delta)$-bisimilarity in future work.
	For $(0,\delta)$-bisimilar chains, we provided bounds based on the error in Erlang CTMCs $\erlang_n$, as well as based on a spectral analysis of the probability matrix $\probMatrix$ of the CTMC in question. Furthermore, we extended our results to reward-bounded reachability probabilities, provided that all rewards in the model are nonnegative.
	
	From a theoretical point of view, the behavior of the bounds obtained from the spectral analysis of $\probMatrix$ is promising. They are tight and, as can be observed in \Cref{fig:example-bounds-diag-queue}, seem to evolve similar to the actual error. 
	The bounds are, however, stretched by a (large) factor, which happens since, to obtain the constants $C$, the triangle inequality and maxima are used. This causes the constants to become significantly larger than necessary, leading to poor performance of the bounds in particular if $t$ is small. Therefore, searching for smaller constants $C' \ll C$ in future work is in order.
	From a practical point of view, applicability of the bounds seems questionable. They tend to over-approximate the error and, depending on the second largest eigenvalue of $\probMatrix$, can take a long time (i.e., require large values of $t$) to converge to the actual error, even for small chains. Moreover, computing Jordan canonical forms (and diagonalizations) is computationally expensive and numerically unstable \cite{BCJNF,NDWCEM,CME}. Consequently, the search for more practically relevant bounds is an important direction for future work.

	 Other open questions include an extension of $(\varepsilon, \delta)$-bisimulation to a notion of $(\varepsilon, \delta, \mu)$-bisimulation that allows a tolerance of $\mu$ in the rewards of related states, an analysis of the achievable state space reduction in quotients w.r.t. (transitive) $(\varepsilon, \delta)$-bisimulations, and a closer look at the (approximate) preservation of logical properties expressed by, e.g., CSL-formulas \cite{MCCTMC,MCACTMC}, between $(\varepsilon, \delta)$-bisimilar states.
	
	\begin{credits}
		\subsubsection*{Acknowledgements.} This work was partly funded by the DFG through
		the DFG grant 389792660 as part of TRR~248 (see \url{https://perspicuous-computing.science})
		and the Cluster of Excellence EXC 2050/1 (CeTI, project ID 390696704, as part of Germany's Excellence Strategy).
		
		\subsubsection*{Disclosure of Interests.} The authors have no competing interests to declare that are relevant to the content of this article.
	\end{credits}
	
	\newpage
	\bibliography{references}{}
		
	\newpage
	\appendix
	
	\section{Proofs of \Cref{sec:definition-fundamental-properties}}
	\begin{restatable}{lemma}{LemCharacterizationTransitive}\label{lem:characterization-transitive}
		Let $R \subseteq \stateSpace \times \stateSpace$ be an equivalence. $R$ satisfies the $\varepsilon$-condition of  \Cref{def:epsilon-delta-bisimulation} iff for every $R$-closed set $A \subseteq \stateSpace$: $\vert \prob(s, A) - \prob(s', A) \vert \leq \varepsilon.$
	\end{restatable}
	\begin{proof}
		Let $(s,s') \in R$. For the direction from left to right, let $A \subseteq \stateSpace$ be $R$-closed. Then $R(A) = A$, so by the $\varepsilon$-condition of \Cref{def:epsilon-delta-bisimulation} we have 
		\begin{align*}
			\prob(s, A) \leq \prob(s', A) + \varepsilon \qquad \text{ and } \qquad \prob(s', A) \leq \prob(s, A) + \varepsilon,
		\end{align*}
		and it follows that $\vert \prob(s, A) - \prob(s', A) \vert \leq \varepsilon.$
		
		For the reverse direction let $A \subseteq \stateSpace$ (not necessarily $R$-closed). Because $R$ is an equivalence $R(A)$ is the union of all equivalence classes containing any element of a $A$, and hence $R$-closed. Thus, we have $\vert \prob(s, R(A)) - \prob(s', R(A)) \vert \leq \varepsilon$, so in particular $\prob(s, R(A)) \leq \prob(s', R(A)) + \varepsilon$. Since $A \subseteq R(A)$ it follows that 
		\begin{align*}
			\prob(s, A) \leq \prob(s, R(A)) \leq \prob(s', R(A)) + \varepsilon.
		\end{align*}
	\end{proof}
	
	\ThmFundamentalProperties*
	\begin{proof}
		\begin{enumerate}
			\item We start by showing that $\sim_{\varepsilon, \delta}$ is itself an $(\varepsilon, \delta)$-bisimulation. To this end we first observe that, as a countable union of symmetric and reflexive relations, $\sim_{\varepsilon, \delta}$ is itself both reflexive and symmetric. Now let $(s,t) \in {\sim_{\varepsilon, \delta}}$. Then there is some $(\varepsilon, \delta)$-bisimulation $R$ such that $(s,t) \in R$. It follows that $\labelFunction(s) = \labelFunction(t)$ and $\vert \ln(\rate(s)) - \ln(\rate(t)) \vert \leq \delta$ (both of these implications are independent of the choice of $R$). 
			
			\quad We now show that for all $A \subseteq \stateSpace$, it holds that $\prob(s, A) \leq \prob(t, {\sim_{\varepsilon, \delta}}(A)) + \varepsilon$. To this end observe that, for any two relations $R_1, R_2 \subseteq \stateSpace \times \stateSpace$ with $R_1 \subseteq R_2$ and all $A \subseteq \stateSpace$ it holds that $R_1(A) \subseteq R_2(A)$. 
			Because $R$ is an $(\varepsilon, \delta)$-bisimulation, $R \subseteq {\sim_{\varepsilon, \delta}}$, and since $(s,t) \in R$ it follows that 
			\begin{align*}
				\prob(s, A) \overset{(s,t) \in R}{\leq} \prob(t, R(A)) + \varepsilon \overset{R(A) \, \subseteq \, {\sim_{\varepsilon, \delta}(A)}}{\leq} \prob(t, {\sim_{\varepsilon, \delta}}(A)) + \varepsilon
			\end{align*}
			for every $A \subseteq \stateSpace$, so $\sim_{\varepsilon, \delta}$ is indeed an $(\varepsilon, \delta)$-bisimulation. 
			
			\quad Now assume that there is a larger $(\varepsilon, \delta)$-bisimulation $R^*$ on $\mathcal{M}$, i.e., an $(\varepsilon, \delta)$-bisimulation $R^*$ such that ${\sim_{\varepsilon, \delta}} \subsetneq R^*$. Then there is some $(s,t) \in R^* \setminus {\sim_{\varepsilon, \delta}}$. However, as $R^*$ is an $(\varepsilon, \delta)$-bisimulation on $\mathcal{M}$, $R^* \subseteq {\sim_{\varepsilon, \delta}}$, and so $R^* \setminus {\sim_{\varepsilon, \delta}} = \emptyset$. Contradiction.
			\item 	It is clear that $\labelFunction(s) = \labelFunction(s') = \labelFunction(s'')$. Define 
			\begin{align*}
				R = \{(x, z) \in \stateSpace \times \stateSpace \mid \exists \, y \in \stateSpace \colon x \sim_{\varepsilon_1, \delta_1} y \text{ and } y \sim_{\varepsilon_2, \delta_2} z\}.
			\end{align*}
			Then it follows from the respective properties of $\sim_{\varepsilon_1, \delta_1}$ and $\sim_{\varepsilon_2, \delta_2}$ that $R$ is a symmetric and reflexive relation, and it holds that $(s,s'') \in R$. Moreover, for every pair of states $(x,z) \in R$, for which $y \in \stateSpace$ is a state satisfying the defining conditions of $R$ w.r.t. $x$ and $z$, we have that 
			\begin{align*}
				\vert \ln(\rate(x)) - \ln(\rate(z)) \vert &= \vert \ln(\rate(x)) - \ln(\rate(y)) + \ln(\rate(y)) - \ln(\rate(z)) \vert \\& \leq \vert \ln(\rate(x)) - \ln(\rate(y)) \vert + \vert \ln(\rate(y)) - \ln(\rate(z)) \vert \\&\leq \delta_1 + \delta_2.
			\end{align*}
			Now let $A \subseteq \stateSpace$. We first observe that ${\sim_{\varepsilon_2, \delta_2}}({\sim_{\varepsilon_1, \delta_1}}(A)) \subseteq R(A)$, since 
			\begin{align*}
				{\sim_{\varepsilon_2, \delta_2}}({\sim_{\varepsilon_1, \delta_1}}(A)) &= \{z \in \stateSpace \mid \exists \, y \in {\sim_{\varepsilon_1, \delta_1}}(A) \colon y \sim_{\varepsilon_2, \delta_2} z\} \\
				&= \{z \in S \mid \exists \, y \in \{x \in \stateSpace \mid \exists \, a \in A \colon a \sim_{\varepsilon_1, \delta_1} x \} \colon y \sim_{\varepsilon_2, \delta_2} z\} \\
				&= \{z \in S \mid \exists \, y \in \stateSpace \colon \exists \, a \in A \colon \underbrace{a \sim_{\varepsilon_1, \delta_1} y \text{ and } y \sim_{\varepsilon_2, \delta_2} z}_{(a, z) \in R}\} \\
				&\subseteq \{z \in S \mid \exists \, a \in A\colon (a, z) \in R \} = R(A).
			\end{align*}
			Hence, we get for all $(x, z) \in R$ and $y$ as above that
			\begin{align*}
				\prob(x, A)& \leq \prob(y, {\sim_{\varepsilon_1, \delta_1}}(A)) + \varepsilon_1 \\&\leq \prob(z, {\sim_{\varepsilon_2, \delta_2}}({\sim_{\varepsilon_1, \delta_1}}(A))) + \varepsilon_1 + \varepsilon_2 \\&\leq \prob(z, R(A)) + \varepsilon_1 + \varepsilon_2,
			\end{align*}
			so $R$ is a $(\varepsilon_1 + \varepsilon_2, \delta_1+ \delta_2)$-bisimulation. As $(s, s'') \in R$, the claim follows.
			\item We first observe that $\sim_{0, 0}$ is an equivalence. To see this, let $s \sim_{0,0} s'$ and $s' \sim_{0,0} s''$. By the additivity of $(\varepsilon, \delta)$-bisimulations proved in item 2 it follows that $s \sim_{0 + 0, 0 + 0} s''$, i.e., $s \sim_{0, 0} s''$, so $\sim_{0, 0}$ is transitive. As any $(\varepsilon, \delta)$-bisimulation is reflexive and symmetric by definition, and because $\sim_{\varepsilon, \delta}$ is itself an $(\varepsilon, \delta)$-bisimulation for any $\varepsilon, \delta \geq 0$ by item 1, $\sim_{0, 0}$ is an equivalence. 
			
			\quad Now let $s \sim_{0, 0} s'$. Then $\labelFunction(s) = \labelFunction(s')$ because $(\varepsilon, \delta)$-bisimulations only relate states with the same label. Furthermore, $\vert \ln(\rate(s)) - \ln(\rate(s')) \vert \leq 0$, so $\ln(\rate(s)) = \ln(\rate(s'))$ which, because of the strong monotonicity of $x \mapsto \ln(x)$ on $\mathbb{R}_{>0}$, is equivalent to $\rate(s) = \rate(s')$. By \Cref{lem:characterization-transitive}, the $\varepsilon$-condition for $\sim_{0,0}$ holds iff $\vert \prob(s, A) - \prob(s', A) \vert \leq 0$ iff $\prob(x, A) = \prob(y, A)$ for every $\sim_{0,0}$-closed set $A$ and any two states $(s,s') \in {\sim_{0,0}}$. As $s \sim_{0,0} s'$ and any equivalence class $C$ of $\sim_{0,0}$ is $\sim_{0, 0}$-closed it follows that for all these $C$, $\prob(s, C) = \prob(s', C)$. Thus, $\sim_{0, 0}$ is a strong bisimulation and hence $s \sim_{0,0} s'$ implies $s \sim s'$. 
			
			\quad In the other direction, let $s \sim s'$. The validity of the labeling- and the $\delta$-condition (for $\delta = 0$) follows as in the previous case. Again by \Cref{lem:characterization-transitive}, the only thing left to prove is that for every $\sim$-closed set $A$ we have $\prob(s, A) = \prob(s', A)$. As $\sim$ is an equivalence, the $\sim$-closed sets are precisely the (unions of) equivalence classes of $\sim$ \cite{RBBTEAPC}. Let $C = \bigcup_{i \in I} C_i$ be such a union for some countable index set $I \subseteq \mathbb{N}$ and equivalence classes $C_i \in \xfrac{\stateSpace}{\sim}$. Then 
			\begin{align*}
				\prob(s, C) = \sum_{i \in I} \prob(s, C_i) \overset{s \sim s'}{=} \sum_{i \in I} \prob(s', C_i) = \prob(s', C). 
			\end{align*}
			Hence, $\sim$ is a $(0, 0)$-bisimulation, and thus $s \sim s'$ implies $s \sim_{0, 0} s'$. 
		\end{enumerate}
	\end{proof}

	\CorAlgorithmicComplexity*
	\begin{proof}
		We adjust the algorithm proposed in \cite{AAPP} for computing $\sim_\varepsilon$ on DTMC in $\mathcal{O}(\vert \stateSpace \vert^7)$. To this end, we first observe that the $\varepsilon$-condition of $(\varepsilon, \delta)$-bisimulation is as for $\varepsilon$-bisimulations on DTMC (see \Cref{def:epsilon-bisimulation,def:epsilon-delta-bisimulation}). The $\varepsilon$-condition can hence be checked by solving specific maximum flow problems, just like in the algorithm of \cite{AAPP}. Furthermore, the $\delta$-condition is a purely state based property, and can hence be checked before considering the $\varepsilon$-condition at all. Thus, altering the algorithm of \cite{AAPP} such that the initial relation $R_0$ equals 
		\begin{align*}
			R_0 = \{(s,s') \in \stateSpace \times \stateSpace \mid \labelFunction(s) = \labelFunction(s') \text{ and } \vert \ln(\rate(s)) - \ln(\rate(s')) \vert \leq \delta\}
		\end{align*}
		causes the relation returned by the algorithm to equal $\sim_{\varepsilon, \delta}$, as it only relates states that satisfy both the labeling and the $\delta$-condition, and only removes states from one of the intermediate relations $R_i$ if they do not satisfy the $\varepsilon$-condition. Since the construction of $R_0$ as above requires at most $\vert \stateSpace \vert^2$ comparisons of state labels and $\vert \stateSpace \vert^2$ comparisons of (the logarithms of)  exit rates, the dominating part of the algorithm is still the (unchanged) loop, leading to the same algorithmic complexity as that of the computation of $\sim_\varepsilon$ on DTMC.
	\end{proof}

	\LemApproximatBisimulationInUniformization*
	\begin{proof}
		Let $s \sim_{\varepsilon, \delta} s'$, and let $\uniformizationRate \geq \max_{p \in \stateSpace} \rate(p)$ be a uniformization rate of $\mathcal{M}$. Recall that the probability distribution $\unifProb$ of $\uniformization{\mathcal{M}}{q}$ is given as \cite{MCASMP}
		\begin{align*}
			\unifProb(s, t) = 
			\begin{cases}
				\frac{\prob(s,t) \cdot \rate(s)}{\uniformizationRate}, & \text{if } s \neq t\\
				1 + \frac{\prob(s,s) \cdot \rate(s)}{\uniformizationRate} - \frac{\rate(s)}{\uniformizationRate}, & \text{if} s = t
			\end{cases}.
		\end{align*}
		In particular, this definition implies $\frac{\prob(s,s) \cdot \rate(s)}{\uniformizationRate} = \unifProb(s, s) + \frac{\rate(s)}{q} - 1$.
		
		The goal is to prove that $\sim_{\varepsilon, \delta}$ induces a $\tau$-bisimulation on $\uniformization{\mathcal{M}}{\uniformizationRate}$ for $\tau = e^\delta \cdot (1 +\varepsilon) - 1$. By definition, $\sim_{\varepsilon, \delta}$ only relates states with the same label. Thus, what remains to be checked is condition (ii) of \Cref{def:epsilon-bisimulation}. To this end, let $A \subseteq \stateSpace$. We distinguish three cases. 
		
		If $s \notin A$ and $s' \notin {\sim_{\varepsilon, \delta}} (A)$ then 
		\begin{align*}
			\unifProb(s, A) &= \frac{\prob(s, A) \cdot \rate(s)}{\uniformizationRate} \\
			&\leq \frac{(\prob(s', {\sim_{\varepsilon, \delta}}(A)) + \varepsilon) \cdot \rate(s') \cdot e^{\delta}}{\uniformizationRate} \\
			&= \underbrace{\frac{\prob(s', {\sim_{\varepsilon, \delta}}(A)) \cdot \rate(s')}{\uniformizationRate}}_{= \unifProb(s', {\sim_{\varepsilon, \delta}}(A))}  \cdot e^{\delta} + \underbrace{\frac{\varepsilon \cdot e^{\delta} \cdot \rate(s')}{\uniformizationRate}}_{\leq \varepsilon \cdot e^{\delta} \text{ as } \rate(s') \leq \uniformizationRate} \\
			&\leq \unifProb(s', {\sim_{\varepsilon, \delta}}(A)) \cdot e^{\delta} + \varepsilon \cdot e^{\delta} \\
			&= \unifProb(s', {\sim_{\varepsilon, \delta}}(A)) + (e^{\delta} - 1) \cdot \unifProb(s', {\sim_{\varepsilon, \delta}}(A)) + \varepsilon \cdot e^{\delta} \\
			&\leq \unifProb(s', {\sim_{\varepsilon, \delta}}(A)) + e^{\delta} - 1 + \varepsilon \cdot e^{\delta} \\
			&= \unifProb(s', {\sim_{\varepsilon, \delta}}(A)) + e^{\delta} \cdot (1 + \varepsilon) - 1. 
		\end{align*}
		Otherwise, if $s \notin A$ but $s' \in {\sim_{\varepsilon, \delta}}(A)$ then 
		\begin{align*}
			\unifProb(s, A) &= \frac{\prob(s, A) \cdot \rate(s)}{q} \\&\leq \frac{(\prob(s', {\sim_{\varepsilon, \delta}}(A))+ \varepsilon)\cdot \rate(s') \cdot e^{\delta}}{\uniformizationRate} \\
			&= e^{\delta} \cdot \underbrace{\frac{\prob(s', {\sim_{\varepsilon, \delta}}(A) \setminus \{s'\}) \cdot \rate(s')}{\uniformizationRate}}_{= \unifProb(s', {\sim_{\varepsilon, \delta}}(A) \setminus \{s'\})} + \ e^{\delta} \cdot \underbrace{\frac{\prob(s', s')\cdot \rate(s')}{\uniformizationRate}}_{= \unifProb(s', s') + \frac{\rate(s')}{\uniformizationRate} - 1} + \frac{\varepsilon \cdot e^{\delta} \cdot\rate(s')}{\uniformizationRate} \\
			&= e^{\delta} \cdot \unifProb(s', {\sim_{\varepsilon, \delta}}(A) \setminus \{s'\}) + e^{\delta} \left(\unifProb(s', s') + \frac{\rate(s')}{\uniformizationRate} - 1 \right) + \frac{\varepsilon \cdot e^{\delta} \cdot \rate(s')}{\uniformizationRate} \\
			&= e^{\delta} \cdot \unifProb(s', {\sim_{\varepsilon, \delta}}(A)) + (1+\varepsilon)\cdot e^{\delta} \cdot \frac{\rate(s')}{\uniformizationRate} - e^{\delta} \\
			&= \unifProb(s', {\sim_{\varepsilon, \delta}}(A)) + (e^{\delta} - 1)\cdot \underbrace{\unifProb(s', {\sim_{\varepsilon, \delta}}(A))}_{\leq 1} + (1 + \varepsilon)\cdot e^{\delta} \cdot \underbrace{\frac{\rate(s')}{\uniformizationRate}}_{\leq 1} - e^{\delta} \\
			&\leq \unifProb(s', {\sim_{\varepsilon, \delta}}(A)) + e^{\delta} - 1 + (1+\varepsilon) \cdot e^\delta - e^{\delta} \\
			&= \unifProb(s', {\sim_{\varepsilon, \delta}}(A)) + e^{\delta} \cdot (1+\varepsilon) - 1.
		\end{align*}
		Lastly, if $s \in A$ then it immediately follows that $s' \in {\sim_{\varepsilon, \delta}}(A)$ as $s \sim_{\varepsilon, \delta} s'$ and 
		\begin{align*}
			\unifProb(s, A) = \frac{\prob(s, A) \cdot \rate(s)}{\uniformizationRate} + 1 - \frac{\rate(s)}{q} = \frac{(\prob(s, A) - 1) \cdot \rate(s)}{q} + 1.
		\end{align*}
		As $\prob(s, A) - 1 \leq 0$, it follows that
		\begin{align*}
			\unifProb(s, A) &\leq \frac{(\prob(s', {\sim_{\varepsilon, \delta}}(A)) + \varepsilon - 1) \cdot e^{-\delta}\cdot \rate(s')}{\uniformizationRate} + 1 \\
			&= e^{-\delta} \cdot \underbrace{\frac{\prob(s', {\sim_{\varepsilon, \delta}}(A) \setminus \{s'\}) \cdot \rate(s')}{\uniformizationRate}}_{= \unifProb(s', {\sim_{\varepsilon, \delta}}(A) \setminus \{s'\})} + e^{-\delta} \cdot \underbrace{\frac{\prob(s', s') \cdot \rate(s')}{\uniformizationRate}}_{= \unifProb(s', s') + \frac{\rate(s')}{\uniformizationRate} - 1} \\&\qquad + (\varepsilon-1) \cdot \frac{e^{-\delta} \cdot \rate(s')}{\uniformizationRate} + 1 \\
			&= \unifProb(s', {\sim_{\varepsilon, \delta}}(A)) + (e^{-\delta} - 1) \cdot \unifProb(s', {\sim_{\varepsilon, \delta}}(A)) \\& \qquad + e^{-\delta} \cdot \frac{\rate(s')}{\uniformizationRate} - e^{-\delta} + (\varepsilon-1) \cdot \frac{e^{-\delta} \cdot \rate(s')}{\uniformizationRate} + 1 \\
			&= \unifProb(s', {\sim_{\varepsilon, \delta}}(A)) + \underbrace{\underbrace{(e^{-\delta} - 1)}_{\leq 0} \cdot \unifProb(s', {\sim_{\varepsilon, \delta}}(A))}_{\leq 0} - e^{-\delta} + \varepsilon \cdot e^{-\delta} \cdot \underbrace{\frac{\rate(s')}{\uniformizationRate}}_{\leq 1}  + 1  \\
			&\leq \unifProb(s', {\sim_{\varepsilon, \delta}}(A)) - e^{-\delta} + \varepsilon \cdot e^{-\delta} + 1 \\
			&= \unifProb(s', {\sim_{\varepsilon, \delta}}(A)) + e^{-\delta} \cdot ( \varepsilon - 1) + 1
		\end{align*}
		Hence, $\sim_{\varepsilon, \delta}$ induces a $\max\{1 + e^{-\delta} \cdot (\varepsilon -1), e^{\delta} \cdot (1 + \varepsilon) -1 \}$-bisimulation on the uniformized DTMC $\uniformization{\mathcal{M}}{\uniformizationRate}$. As this maximum is, for any $\varepsilon, \delta \geq 0$, attained at $e^{\delta} \cdot ( 1 + \varepsilon) - 1$, the claim follows. 
	\end{proof}
	
	\PropConnectionTauLump*
	\begin{proof}
		Let $C \in \xfrac{\stateSpace}{R}$ and $(s, s') \in R$. Then 
		\begin{align*}
			\prob(s,C) \cdot \rate(s) &\leq (\prob(s', C) + \varepsilon)\cdot \rate(s') \cdot e^{\delta} \\
			&= \prob(s', C) \cdot \rate(s') \cdot e^{\delta} + \varepsilon \cdot \rate(s') \cdot e^{\delta} \\
			&= \prob(s', C) \cdot \rate(s') + (e^{\delta} - 1) \cdot \prob(s', C) \cdot \rate(s') + \varepsilon \cdot \rate(s') \cdot e^{\delta} \\
			&= \prob(s', C) \cdot \rate(s') + \rate(s') \cdot ((e^{\delta} - 1)\cdot \prob(s', C) + \varepsilon \cdot e^{\delta}) \\
			&\leq \prob(s', C) \cdot \rate(s') + \rate(s') \cdot (e^{\delta} -1 + \varepsilon \cdot e^{\delta}) \\
			& \leq \prob(s', C) \cdot \rate(s') + q \cdot (e^{\delta} \cdot (1 + \varepsilon) - 1).
		\end{align*}
		With the same calculations, $\prob(s', C) \cdot \rate(s') \leq \prob(s, C) \cdot \rate(s) + q \cdot (e^{\delta} \cdot (1 + \varepsilon) -1)$, and so $\vert \rate(s, C) - \rate(s', C) \vert \leq q \cdot (e^\delta \cdot (1+\varepsilon) - 1)$ for all equivalence classes $C$ of $R$ and all $(s,s') \in R$. Thus, $\xfrac{S}{R}$ is a $q \cdot(e^{\delta} \cdot (1 + \varepsilon) - 1)$-lumpability of $\stateSpace$.
	\end{proof}
	
	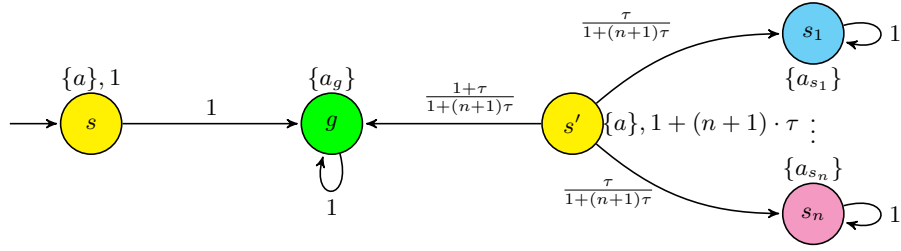
\begin{figure}[t]
		\centering
		\resizebox{!}{!}{
			\begin{tikzpicture}[->,>=stealth',shorten >=1pt,auto, semithick]
				\tikzstyle{every state} = [text = black]
				
				\node[state] (s) [fill = yellow] {$s$};
				\node[state] (g) [right of = s, node distance = 3.2cm, fill = green] {$g$};
				\node[state] (s') [right of = g, node distance = 3.2cm, fill = yellow] {$s'$};
				\node(temp) [right of = s', node distance = 3.2cm] {$\vdots$}; 
				\node[state] (s1) [above of = temp, node distance = 1.2cm, fill = cyan!50] {$s_1$}; 
				\node[state] (sn) [below of = temp, node distance = 1.2cm, fill = magenta!50] {$s_n$}; 
				\node (init) [left of = s, node distance = 1.2cm] {}; 
				
				\path 
					(init) edge (s)
					(s) edge node {$1$} (g)
					(g) edge [loop below] node {$1$} (g)
					(s') edge node [above, pos = 0.4] {$\frac{1 + \tau}{1 + (n+1)\tau}$} (g)
					(s') edge [bend left = 20] node {$\frac{\tau}{1 + (n+1)\tau}$} (s1)
					(s') edge [bend right = 20] node [below, pos = 0.2, xshift = -0.3cm] {$\frac{\tau}{1 + (n+1)\tau}$} (sn)
					(s1) edge [loop right] node {$1$} (s1)
					(sn) edge [loop right] node {$1$} (sn)
				; 
				
				\node[above of = s, node distance = 0.6cm] {$\{a\}, 1$};
				\node[right of = s', node distance = 1.7cm, align=left] {$\{a\}, 1 + (n+1) \cdot \tau$};
				\node[above of = g, node distance = 0.6cm] {$\{a_g\}$};
				\node[below of = s1, node distance = 0.6cm] {$\{a_{s_1}\}$};
				\node[above of = sn, node distance = 0.6cm] {$\{a_{s_n}\}$};
				
		\end{tikzpicture}}
		\caption{The CTMC $\mathcal{M}$ used in the proof of \Cref{prop:quasi-lumpability-does-not-imply-epsilon-delta}.}
		\label{fig:cex-lumpability}
	\end{figure}
	
	\PropTauLumpDoesNotImplyEpsilonDelta*
	\begin{proof}
		Let $\varepsilon \in (0, 1)$, $\delta > 0, \tau > 0$ and $n \in \mathbb{N}$ with \mbox{$n > \max\{\frac{\varepsilon\cdot(\tau + 1)}{\tau\cdot(1 - \varepsilon)}, \frac{e^{\delta}}{\tau} - 1\}$.} We consider a CTMC $\mathcal{M} = \mathcal{M}(\varepsilon, \delta, \tau)$. It has states $s, s', g, s_1, \ldots, s_n$ and transition probabilities $\prob(s, g) = 1, \prob(s', g) = \frac{1 + \tau}{1 + (n+1)\tau}, \prob(s', s_i) = \frac{\tau}{1 + (n+1)\tau}$ for all $i = 1, \ldots, n$, and self-loops in the remaining states. For our purpose, only the exit rates of states $s$ and $s'$ matter, which are $\rate(s) = 1$ and $\rate(s') = 1 + (n+1) \cdot \tau$. The labeling function $\labelFunction$ is such that $\labelFunction(s) = \labelFunction(s') = a$ and it assigns unique labels $a_x$ to all other states $x$. In particular, besides $s$ and $s'$ no two states have the same label. A depiction of $\mathcal{M}$ can be found in \Cref{fig:cex-lumpability}.
		
		In $\mathcal{M}$, the pair $(s,s')$ does not satisfy the $\delta$-condition since
		\begin{align*}
			\vert \rate(s) - \rate(s') \vert = \vert 1 - (1 + (n+1) \cdot \tau) \vert = (n+1) \cdot \tau > \left(\frac{e^{\delta}}{\tau} - 1 + 1 \right) \cdot \tau = e^{\delta}.
		\end{align*} 
		Moreover, the $\varepsilon$-condition is violated because 
		\begin{align*}
			\prob(s', \{g\}) + \varepsilon &= \frac{1 + \tau}{1 + (n+1)\cdot \tau} + \varepsilon 
			\\&< \frac{1 + \tau}{1 + (\frac{\varepsilon\cdot(\tau +1)}{\tau\cdot (1-\varepsilon)} + 1) \cdot \tau} + \varepsilon 
			\\&= \frac{1 + \tau}{1 + \frac{\varepsilon \cdot(\tau + 1) + \tau \cdot (1 - \varepsilon)}{1 - \varepsilon}} + \varepsilon 
			\\&= \frac{1+ \tau}{1 + \frac{\varepsilon + \tau}{1 - \varepsilon}} + \varepsilon 
			\\&= \frac{1 + \tau}{\frac{\varepsilon + \tau + 1 - \varepsilon}{1 - \varepsilon}} + \varepsilon \\&= \frac{1+\tau}{1 + \tau} \cdot (1 - \varepsilon) + \varepsilon 
			\\&= 1 = \prob(s, \{g\})
		\end{align*}
		The partition $\Omega = \{\{s, s'\}, \{g\}\} \cup \{\{s_i\} \mid i = 1, \ldots, n\}$ that relates $s$ and $s'$ is, however, a $\tau$-quasi lumpability since 
		\begin{align*}
			\vert \prob(s, \{g\})  \rate(s)- \prob(s', \{g\})  \rate(s') \vert &= \left \vert 1 - \left(\frac{1 + \tau}{1 + (n+1) \cdot \tau} \right) \cdot (1 + (n+1) \cdot \tau)  \right \vert \\&= \left \vert 1 - (1 + \tau) \right \vert = \tau
		\end{align*} 
		and, for all $i$, 
		\begin{align*}
			\vert \prob(s, \{s_i\})  \rate(s) - \prob(s', \{s_i\})  \rate(s') \vert = \left\vert \frac{\tau}{1 + (n+1)\tau} \cdot (1 + (n+1) \cdot \tau) \right \vert = \tau.
		\end{align*}
	\end{proof}
	
	\ThmSameStructureConstruction*
	\begin{proof}
		We start by defining the models $\mathcal{M}'$ and $\mathcal{N}'$. They both share the same state space $\stateSpace' = \stateSpace^\mathcal{M} \times \stateSpace^\mathcal{N}$, with initial state $\initialState' = (\initialState^\mathcal{M}, \initialState^\mathcal{N})$. The label function $\labelFunction'$ of both models is defined to be $\labelFunction'((s,t)) = \labelFunction^\mathcal{N}(t)$, and the exit rate functions are defined via 
		\begin{align*}
			{\rate^\mathcal{M}}'((s,t)) &= \begin{cases}
				\rate^\mathcal{M}(s), &\text{if } s \sim_{\varepsilon, \delta} \! \! t \\
				\rate^\mathcal{N}(t), &\text{if} s \nsim_{\varepsilon, \delta} t
			\end{cases} \qquad \text{and} \qquad
			{\rate^\mathcal{N}}'((s,t)) = \rate^\mathcal{N}(t).
		\end{align*}
		Lastly, the common transition probability function is defined as 
		\begin{align*}
			\prob'((s,t), (s',t')) &= \begin{cases}
				\Delta_{s,t}(s',t') \prob^\mathcal{M}(s,s'), &\text{if } s \sim_{\varepsilon, \delta} t, s' \in \Succ(s), t' \in \Succ(t) \\
				\prob^\mathcal{M}(s,s')  \prob^\mathcal{N}(t,t') &\text{if } s \nsim_{\varepsilon, \delta} t \\
				0, &\text{otherwise}
			\end{cases}
		\end{align*}
		where, for a given pair $(s,t) \in \stateSpace'$ with $s \sim_{\varepsilon, \delta} t$, $\Delta_{s,t} \colon \Succ(s) \to \Distr(\Succ(t))$ is a weight function as in \Cref{lem:characterization-weight-functions-concur}. Note that $\Delta_{s,t}$ exists for every such pair $(s,t)$ as any two $(\varepsilon, \delta)$-bisimilar states are $\varepsilon$-bisimilar in the DTMCs underlying $\mathcal{M}$ and $\mathcal{N}$ by the definition of $(\varepsilon, \delta)$-bisimulations. All in all, we thus have $\mathcal{M}' = (\stateSpace', \prob', {\rate^\mathcal{M}}', \initialState', \labelFunction')$ and $\mathcal{N}' = (\stateSpace', \prob', {\rate^\mathcal{N}}', \initialState', \labelFunction')$.
		
		To ensure well-definedness of $\mathcal{M}'$ and $\mathcal{N}'$, we first show that $\prob'$ is actually a probability distribution on $\stateSpace'$. To this end, let $(s,t) \in \stateSpace'$. If $s \sim_{\varepsilon, \delta} t$ then, by definition of $\prob'$, we have $\prob'((s,t), (s',t')) = 0$ if $s' \notin \Succ(s)$ or $t' \notin \Succ(t)$, i.e., $\prob'((s,t), (s',t')) = 0$ for every $(s',t') \notin (\Succ(s) \times \Succ(t))$. Hence, 
		\begin{align*}
			\prob'((s,t), \stateSpace') &= \sum_{(s', t') \in (\Succ(s) \times \Succ(t))} \prob'((s,t), (s',t')) \\
			&= \sum_{(s', t') \in (\Succ(s) \times \Succ(t))} \Delta_{s,t}(s', t') \cdot \prob^\mathcal{M}(s,s') \\
			&= \sum_{s' \in \Succ(s)} \prob^\mathcal{M}(s,s') \cdot \underbrace{\sum_{t' \in \Succ(t)} \Delta_{s,t}(s',t')}_{=1 \text{ as} \Delta_{s,t}(s', \cdot) \in \Distr(Succ(t))} \\
			&= \sum_{s' \in \Succ(s)} \prob^\mathcal{M}(s,s') = 1.  
		\end{align*}
		Otherwise, if $s \nsim_{\varepsilon, \delta} t$ then 
		\begin{align*}
			\prob'((s,t), \stateSpace') &= \sum_{(s',t') \in \stateSpace'}\prob^\mathcal{M}(s,s') \cdot \prob^\mathcal{N}(t,t') \\&= \sum_{s' \in \stateSpace^\mathcal{M}} \prob^\mathcal{M}(s,s') \cdot \sum_{t' \in \stateSpace^\mathcal{N}}  \prob^\mathcal{N}(t,t') = 1.
		\end{align*}
		
		We now prove that $\mathcal{M} \sim_{\varepsilon, 0} \mathcal{M}' \sim_{0, \delta} \mathcal{N}' \sim \mathcal{N}$. 
		
		\medskip
		
		\textbf{$\mathcal{M}' \sim_{0, \delta} \mathcal{N}'$: } \qquad Consider the identity relation $R_{id} = \{((s,t), (s,t)) \mid (s,t) \in \stateSpace'\}$ on $\stateSpace'$. $R_{id}$ is reflexive, symmetric and only relates states with the same label. Furthermore, since both $\mathcal{M}'$ and $\mathcal{N}'$ have the same transition probability function $\prob'$, the $\varepsilon$-condition of \Cref{def:epsilon-delta-bisimulation} is satisfied for $\varepsilon = 0$. Lastly, regarding the $\delta$-condition, we distinguish two cases. If $s \sim_{\varepsilon, \delta} t$ then 
		\begin{align*}
			\vert \ln ({\rate^\mathcal{M}}'((s,t))) - \ln({\rate^\mathcal{N}}'((s,t))) \vert = \vert \ln(\rate^\mathcal{M}(s)) - \ln(\rate^\mathcal{N}(t)) \vert \leq \delta.
		\end{align*}
		Otherwise, if $s \nsim_{\varepsilon, \delta} t$ then ${\rate^\mathcal{M}}'((s,t)) = \rate^\mathcal{N}(t) = {\rate^\mathcal{N}}'((s,t))$ and the $\delta$-condition holds trivially. Hence, $R_{id}$ is a $(0, \delta)$-bisimulation, and since it relates the initial states of both $\mathcal{M}'$ and $\mathcal{N}'$ it follows that $\mathcal{M}' \sim_{0, \delta}\mathcal{N}'$.
		
		\medskip 
		
		\textbf{$\mathcal{M} \sim_{\varepsilon, 0} \mathcal{M}'$: } \qquad Denote by $R \subseteq \stateSpace^\mathcal{M} \times \stateSpace'$ the reflexive and symmetric closure of 
		\begin{align*}
			\{(s,(s,t)) \mid s \in \stateSpace^\mathcal{M}, (s,t) \in \stateSpace' \text{ such that } s \sim_{\varepsilon, \delta} t\}.
		\end{align*}
		Note that, in particular, $(\initialState^\mathcal{M}, \initialState') \in R$ as $\initialState^\mathcal{M} \sim_{\varepsilon, \delta} \initialState^\mathcal{N}$ by assumption.
		
		Let $(s,(s,t)) \in R$ (the case for $((s,t), s) \in R$ works analogously, and the cases $(s,s) \in R$ resp. $((s,t), (s,t)) \in R$ are trivial). Then $\labelFunction'((s,t)) = \labelFunction^\mathcal{N}(t) = \labelFunction^\mathcal{M}(s)$ since $s \sim_{\varepsilon, \delta} t$, and 
		\begin{align*}
			\vert \ln (\rate^\mathcal{M}(s)) - \ln ({\rate^\mathcal{M}}'((s,t))) \vert = \vert \ln (\rate^\mathcal{M}(s)) - \ln (\rate^\mathcal{M}(s)) \vert = 0,
		\end{align*}
		so $R$ satisfies both the labeling- and the $\delta$-condition (for $\delta = 0$) of \Cref{def:epsilon-delta-bisimulation}. 
		
		The only thing left to prove for $R$ to be an $(\varepsilon, 0)$-bisimulation is thus the $\varepsilon$-condition. To see that this conditions holds, let $A \subseteq \stateSpace^\mathcal{M} \cup \stateSpace'$. Then 
		\begin{align*}
			\prob'((s,t), R(A)) = \sum_{(s', t') \in R(A)} \prob'((s,t), (s',t')) = \sum_{\substack{(s',t') \in R(A) \\ s' \in \Succ(s) \\ t' \in \Succ(t)}} \prob'((s,t), (s',t')) 
		\end{align*}
		where the second equality follows from the definition of $\prob'$ since $s \sim_{\varepsilon, \delta} t$. 
		
		By the definition of $R$ every $(s', t')$ with $s' \in A \cap \Succ(s)$ and $t' \in {\sim_{\varepsilon, \delta}}{\{s'\}} \cap \Succ(t)$ is contained in $R(A)$, i.e., $\{(s',t') \in R(A) \mid s' \in \Succ(s), t' \in \Succ(t)\}$ is a superset of $\{(s',t') \in \stateSpace' \mid s' \in A \cap \Succ(s), t' \in \Succ(t) \cap {\sim_{\varepsilon, \delta}}(\{s'\})\}$. Thus,
		\begin{align*}
			\prob'((s,t), R(A)) &= \sum_{\substack{(s',t') \in R(A) \\ s' \in \Succ(s) \\ t' \in \Succ(t)}} \prob'((s,t), (s',t')) \\
			&\geq \sum_{s' \in \Succ(s) \cap A} \sum_{t'\in \Succ(t) \cap {\sim_{\varepsilon,\delta}}(\{s'\})} \prob'((s,t), (s',t')) \\
			&= \sum_{s' \in \Succ(s) \cap A} \sum_{t'\in \Succ(t)\cap {\sim_{\varepsilon,\delta}}(\{s'\})} \prob^\mathcal{M}(s,s') \cdot \Delta_{s,t}(s',t') \\
			&= \sum_{s' \in \Succ(s) \cap A} \sum_{t'\in \Succ(t)} \prob^\mathcal{M}(s,s') \cdot \Delta_{s,t}(s',t') \\& \qquad - \sum_{s' \in \Succ(s) \cap A} \sum_{t'\in {\Succ(t) \setminus \sim_{\varepsilon,\delta}}(\{s'\})} \prob^\mathcal{M}(s,s') \cdot \Delta_{s,t}(s',t') \\
			&= \prob^\mathcal{M}(s, A) - \sum_{s' \in \Succ(s) \cap A} \sum_{t'\in {\Succ(t) \setminus \sim_{\varepsilon,\delta}}(\{s'\})} \prob^\mathcal{M}(s,s') \cdot \Delta_{s,t}(s',t')
		\end{align*}
		as $\sum_{t' \in \Succ(t)} \Delta_{s,t}(s',t') = 1$ since $\Delta_{s,t}(s', \cdot) \in \Distr(\Succ(t))$. The sum in the previous term can then be bounded from above via
		\begin{align}
			&\textcolor{white}{=} \sum_{s' \in \Succ(s) \cap A} \sum_{t'\in {\Succ(t) \setminus \sim_{\varepsilon,\delta}}(\{s'\})} \prob^\mathcal{M}(s,s') \cdot \Delta_{s,t}(s',t') \nonumber \\&\leq \sum_{s' \in \Succ(s)} \sum_{t'\in {\Succ(t) \setminus \sim_{\varepsilon,\delta}}(\{s'\})} \prob^\mathcal{M}(s,s') \cdot \Delta_{s,t}(s',t') \nonumber \\
			&= \sum_{s'\in \Succ(s)} \prob^\mathcal{M}(s,s') \cdot \left(1 - \sum_{t' \in \Succ(t) \cap {\sim_{\varepsilon, \delta}}(\{s'\})} \Delta_{s,t}(s',t') \right) \nonumber \\
			&= 1 - \underbrace{\sum_{s' \in \Succ(s)}\sum_{t' \in \Succ(t) \cap {\sim_{\varepsilon, \delta}}(\{s'\})} \prob^\mathcal{M}(s,s') \cdot  \Delta_{s,t}(s',t')}_{\geq 1 - \varepsilon \text{ by condition (2) of  \Cref{lem:characterization-weight-functions-concur}} } \nonumber \\
			&\leq 1 - (1-\varepsilon) = \varepsilon. \label{eq:same-structure-2}
		\end{align}
		All in all, this yields that 
		\begin{align*}
			\prob'((s,t), R(A)) &\geq  \prob^\mathcal{M}(s, A) - \underbrace{\sum_{s' \in \Succ(s) \cap A} \sum_{t'\in {\Succ(t) \setminus \sim_{\varepsilon,\delta}}(\{s'\})} \prob^\mathcal{M}(s,s') \cdot \Delta_{s,t}(s',t')}_{\leq \varepsilon \text{ by \Cref{eq:same-structure-2}}} \\&\geq \prob^\mathcal{M}(s,A) - \varepsilon,
		\end{align*}
		which is equivalent to the $\varepsilon$-condition
		\begin{align*}
			\prob^\mathcal{M}(s,A) \leq \prob'((s,t), R(A)) + \varepsilon.
		\end{align*}
		
		To prove that $\prob'((s,t), A) \leq \prob^\mathcal{M}(s,R(A)) + \varepsilon$ for all $A \subseteq S$, we first observe that since $(s,(s,t)) \in R$ it again has to hold that $s \sim_{\varepsilon, \delta} t$, so we get 
		\begin{align}
			\prob'((s,t), A) &= \sum_{(s', t') \in A} \prob'((s,t), (s',t')) \nonumber \\
			&= \sum_{\substack{(s',t') \in A \\ s' \in \Succ(s) \\ t' \in \Succ(t)}} \prob^\mathcal{M}(s,s') \cdot \Delta_{s,t}(s',t') \nonumber\\
			&= \sum_{\substack{(s',t') \in A \\ s' \in \Succ(s) \\ t' \in \Succ(t) \cap {\sim_{\varepsilon, \delta}}(\{s'\})}} \prob^\mathcal{M}(s,s') \cdot \Delta_{s,t}(s',t') \nonumber\\& \qquad+ \sum_{\substack{(s',t') \in A \\ s' \in \Succ(s) \\ t' \in \Succ(t) \setminus {\sim_{\varepsilon, \delta}}(\{s'\})}} \prob^\mathcal{M}(s,s') \cdot \Delta_{s,t}(s',t') \label{eq:same-struct-1}
		\end{align}
		
		As $(s', t') \in A$ with $s' \sim_{\varepsilon, \delta} t'$ implies $s' \in R(A)$ it follows that the first sum in \Cref{eq:same-struct-1} can be bounded from above via
		\begin{align*}
			&\sum_{\substack{(s',t') \in A \\ s' \in \Succ(s) \\ t' \in \Succ(t) \cap {\sim_{\varepsilon, \delta}}(\{s'\})}} \prob^\mathcal{M}(s,s') \cdot \Delta_{s,t}(s',t') \\&\leq \sum_{s' \in \Succ(s) \cap R(A)} \sum_{t'\in \Succ(t) \cap {\sim_{\varepsilon, \delta}}(\{s'\}) }\prob^\mathcal{M}(s,s') \cdot \Delta_{s,t}(s',t') \\
			&\leq \sum_{s' \in \Succ(s) \cap R(A)} \prob^\mathcal{M}(s,s') \cdot \sum_{t'\in \Succ(t)} \Delta_{s,t}(s',t') \\
			&= \prob^\mathcal{M}(s,R(A)).
		\end{align*}
		Furthermore, similar to \Cref{eq:same-structure-2}, the second sum of \Cref{eq:same-struct-1} satisfies 
		\begin{align*}
			&\sum_{\substack{(s',t') \in A \\ s' \in \Succ(s) \\ t' \in \Succ(t) \setminus {\sim_{\varepsilon, \delta}}(\{s'\})}} \prob^\mathcal{M}(s,s') \cdot \Delta_{s,t}(s',t') \leq \varepsilon.
		\end{align*}
		Hence, \Cref{eq:same-struct-1} becomes
		\begin{align*}
			\prob'((s,t), A) &= \sum_{\substack{(s',t') \in A \\ s' \in \Succ(s) \\ t' \in \Succ(t) \cap {\sim_{\varepsilon, \delta}}(\{s'\})}} \prob^\mathcal{M}(s,s') \cdot \Delta_{s,t}(s',t') \nonumber\\& \qquad+ \sum_{\substack{(s',t') \in A \\ s' \in \Succ(s) \\ t' \in \Succ(t) \setminus {\sim_{\varepsilon, \delta}}(\{s'\})}} \prob^\mathcal{M}(s,s') \cdot \Delta_{s,t}(s',t') \\
			&\leq \prob^\mathcal{M}(s, R(A)) + \varepsilon.
		\end{align*}
		All in all, this proves that $R$ is an $(\varepsilon,0)$-bisimulation, and thus $\mathcal{M} \sim_{\varepsilon, 0} \mathcal{M}'$. 
		
		\medskip 
		
		\textbf{$\mathcal{N}' \sim \mathcal{N}$: } \qquad Let $R'$ be the smallest equivalence on $\stateSpace^\mathcal{N} \cup \stateSpace'$ that contains all pairs $(t, (s,t))$ with $t \in \stateSpace^\mathcal{N}$ and $s \in \stateSpace^\mathcal{M}$. It is easy to see that every equivalence class of $R'$ is of the form $C = \{t\} \cup (\stateSpace^\mathcal{M} \times \{t\})$ for some $t \in \stateSpace^\mathcal{N}$. 
		
		Now, let $(x,y) \in R$. In the case that $x = y = t$ or $x = y = (s,t)$ for some $t \in \stateSpace^\mathcal{N}$ resp. some $(s,t) \in \stateSpace'$ it is clear that $R'$ satisfies all conditions of a strong bisimulation w.r.t. the pair $(x,y)$. 
		
		Otherwise, we consider two cases. If, on the one hand, $x = t\in \stateSpace^\mathcal{N}$ and $y \in \stateSpace'$ (or the other way around, which can be proved analogously) it follows from the definition of $R'$ that $y = (s,t)$ for some $s \in \stateSpace^\mathcal{M}$, i.e., that $(x,y) = (t, (s,t))$. By definition we have $\labelFunction'((s,t)) = \labelFunction^\mathcal{N}(t)$ and ${\rate^\mathcal{N}}'((s,t)) = \rate^\mathcal{N}(t)$, so we only need to check the transition probabilities. To this end, let $C = \{t'\} \cup (\stateSpace^\mathcal{M} \times \{t'\})$ be some equivalence class of $R'$. If $s \sim_{\varepsilon, \delta} t$ in $\mathcal{M} \oplus \mathcal{N}$ then there is a weight function $\Delta_{s,t}$ as in \Cref{lem:characterization-weight-functions-concur} for the pair $(s,t)$. If additionally $t' \in \Succ(t)$ we get
		\begin{align*}
			\prob^\mathcal{N}(t,C) &= \prob^\mathcal{N}(t,t') \\
			&= \sum_{s' \in \Succ(s)} \prob^\mathcal{M}(s,s') \cdot \Delta_{s,t}(s',t') \\
			&= \sum_{s' \in \Succ(s)} \prob^\mathcal{M}(s,s') \cdot \Delta_{s,t}(s',t') + \underbrace{ \sum_{s' \in \stateSpace^\mathcal{M} \setminus \Succ(s)} \prob'((s,t), (s',t'))}_{= 0} \\
			&= \sum_{s' \in \Succ(s)} \prob'((s,t), (s',t')) + \sum_{s'\in \stateSpace^\mathcal{M} \setminus \Succ(s)} \prob'((s,t), (s',t')) \\
			&= \sum_{s' \in \stateSpace^\mathcal{M}} \prob'((s,t), (s', t')) \\
			&= \sum_{(s', t') \in C} \prob'((s,t), (s',t')) = \prob'((s,t), C),
		\end{align*}
		while in the case that $t' \notin \Succ(t)$ we have 
		\begin{align*}
			\prob^\mathcal{N}(t,C) &= \prob^\mathcal{N}(t,t') = 0 = \sum_{(s',t') \in C} \prob'((s,t), (s',t')).
		\end{align*}
		Moreover, if $s \nsim_{\varepsilon, \delta} t$, we get 
		\begin{align*}
			\prob^\mathcal{N}(t,C) &= \prob^\mathcal{N}(t,t') \\
			&= \prob^\mathcal{N}(t,t') \cdot \sum_{s' \in \stateSpace^\mathcal{M}}\prob^\mathcal{M}(s,s') \\
			&= \sum_{s' \in \stateSpace^\mathcal{M}} \prob^\mathcal{M}(s,s') \cdot \prob^\mathcal{N}(t,t') \\
			&= \sum_{s' \in \stateSpace^\mathcal{M}} \prob'((s,t),(s',t')) \\
			&= \sum_{(s',t') \in C} \prob'((s,t), (s',t')) = \prob'((s,t), C). 
		\end{align*}
		
		If, on the other hand, $x,y \in \stateSpace'$ it follows from the definition of $R'$ that there is some $t \in \stateSpace^\mathcal{N}$ such that $x = (s, t)$ and $y = (s', t)$. Let $C = \{t'\} \cup (\stateSpace^\mathcal{M} \times \{t'\})$ be an equivalence class of $R'$. 
		
		We again do a case distinction, depending on the $(\varepsilon, \delta)$-bisimilarity of $t$ and $s$ resp. $s'$ in $\mathcal{M} \oplus \mathcal{N}$. In the case that $s \sim_{\varepsilon, \delta} t$, $s' \sim_{\varepsilon, \delta} t$ and $t' \in \Succ(t)$ we have 
		\begin{align*}
			\prob'((s,t), C) &= \sum_{(p,t') \in C} \prob'((s,t), (p,t')) \\&= \sum_{p \in \Succ(s)} \prob^\mathcal{M}(s,p) \cdot \Delta_{s,t}(p,t') \\&= \prob^\mathcal{N}(t,t') \\
			&= \sum_{p \in \Succ(s')} \prob^\mathcal{M}(s', p) \cdot \Delta_{s', t}(p, t')\\&= \sum_{(p, t') \in C} \prob'((s', t), (p, t')) \\&= \prob'((s',t), C),
		\end{align*}
		while if $t' \notin \Succ(t)$ it holds that 
		\begin{align*}
			\prob'((s,t), C) = 0 = \prob'((s', t), C).
		\end{align*}
		Otherwise, if $s \nsim_{\varepsilon, \delta} t$ and $s' \nsim_{\varepsilon, \delta} t$ then 
		\begin{align*}
			\prob'((s,t), C) &= \sum_{(p, t') \in C} \prob^\mathcal{M}(s,p) \cdot \prob^\mathcal{N}(t,t') \\&= \sum_{p \in S^\mathcal{M}} \prob^\mathcal{M}(s,p) \cdot \prob^\mathcal{N}(t,t') \\&= \prob^\mathcal{N}(t,t') \\
			&= \sum_{p \in S^\mathcal{M}} \prob^\mathcal{M}(s',p) \cdot \prob^\mathcal{N}(t,t') \\&= \prob'((s', t), C). 
		\end{align*}
		
		Lastly, if $s \sim_{\varepsilon, \delta} t$ and $s' \nsim_{\varepsilon, \delta} t$ (or the other way around), we have if $t' \notin \Succ(t)$ that 
		\begin{align*}
			\prob'((s,t), C) = 0 &= \sum_{p \in S^\mathcal{M}} \prob^\mathcal{M}(s',p) \cdot \underbrace{\prob^\mathcal{N}(t,t')}_{=0} = \sum_{(p,t') \in C} \prob'((s',t),(p,t')) \\&=  \prob'((s', t), C),
		\end{align*}
		while for $t' \in \Succ(t)$ we get
		\begin{align*}
			\prob'((s,t), C) &= \sum_{(p, t') \in C} \prob'((s,t), (p, t')) \\&= \sum_{p \in \Succ(s)} \prob^\mathcal{M}(s,p) \cdot \Delta_{s,t}(p,t') \\&= \prob^\mathcal{N}(t, t') \\
			&=  \prob^\mathcal{N}(t,t') \cdot \sum_{p \in \stateSpace^\mathcal{M}}\prob^\mathcal{M}(s',p) \\&= \sum_{p \in \stateSpace^\mathcal{M}}\prob^\mathcal{M}(s',p) \cdot \prob^\mathcal{N}(t,t') \\&= \prob'((s', t), C)
		\end{align*}
		
		All in all, the previous calculations proved that, for every pair of states $(x,y) \in R'$ and any $R'$-equivalence class $C$ we have $\prob'(x, C) = \prob'(y, C)$, and hence $R'$ is a strong bisimulation between $\mathcal{N}$ and $\mathcal{N}'$. As $R'$ in particular relates the initial states of $\mathcal{N}$ and $\mathcal{N}'$ it follows that $\mathcal{N}' \sim \mathcal{N}$. 
	\end{proof}

	\section{Proofs of \Cref{sec:time-bounded-reachability-bounds}}

	\PropTransientBoundsEpsilonDeltaGreateZero*
	\begin{proof}
		Let $\uniformizationRate$ be a uniformization rate of $\mathcal{M}$, i.e., $\uniformizationRate \geq \max_{s \in \stateSpace} \rate(s)$. From \Cref{lem:approximate-bisimulation-in-uniformization} we know that in $\uniformization{\mathcal{M}}{\uniformizationRate}$ we have $s \sim_{e^\delta \cdot(1 + \varepsilon) - 1} s'$. 
		
		\Cref{thm:main-result-rbbteapc} \cite{RBBTEAPC} together with the fact that the probability to reach $g$ from $s$ resp. $s'$ after at most $k$ steps in $\uniformization{\mathcal{M}}{\uniformizationRate}$ coincides with the entries $\unifProbMatrix^k_{s,g}$ resp. $\unifProbMatrix^k_{s', g}$ at positions $[s,g]$ resp. $[s', g]$ of $\unifProbMatrix^k$ \cite{PoMC} yields the desired bound via
		\begin{align*}
			\vert \probMeasure_s(\reachability^{\leq t} g) - \probMeasure_{s'}(\reachability^{\leq t} g) 	\vert &= e^{-\uniformizationRate \cdot t} \cdot \left \vert \sum_{k = 0}^\infty \frac{(\uniformizationRate \cdot t)^k}{k!} \cdot (\unifProbMatrix^k_{s,g} - \unifProbMatrix^k_{s', g}) \right \vert \\
			&= e^{-\uniformizationRate \cdot t} \cdot \left \vert \sum_{k = 0}^\infty \frac{(\uniformizationRate \cdot t)^k}{k!} \cdot (\probMeasure_s(\reachability^{\leq k} g) - \probMeasure_{s'}(\reachability^{\leq k} g)) \right \vert \\
			&\leq e^{-\uniformizationRate \cdot t} \cdot \sum_{k=0}^\infty \frac{(\uniformizationRate \cdot  t)^k}{k!} \cdot \vert 	\probMeasure_s(\reachability^{\leq k} g) - \probMeasure_{s'}(\reachability^{\leq k} g) \vert \\
			&\leq e^{-\uniformizationRate\cdot  t} \cdot\sum_{k=0}^\infty \frac{(\uniformizationRate \cdot t)^k}{k!} \cdot (1 - 	(1-(e^\delta \cdot(1 + \varepsilon) - 1))^k) \\
			&= e^{-\uniformizationRate \cdot  t} \cdot \sum_{k=0}^\infty \frac{(\uniformizationRate \cdot  t)^k}{k!} \cdot (1 - (2- e	^\delta \cdot(1 + \varepsilon))^k) \\
			&= e^{-\uniformizationRate\cdot t} \cdot \underbrace{\sum_{k = 0}^{\infty} \frac{(q\cdot t)^k}{k!}}_{= e	^{\uniformizationRate\cdot t}} - e^{- \uniformizationRate\cdot t} \cdot \underbrace{\sum_{k=0}^{\infty} \frac{(\uniformizationRate \cdot t \cdot (2-e^{\delta} \cdot (1 + \varepsilon)))^k}{k!}}_{= e^{\uniformizationRate\cdot t \cdot (2-e^{\delta} \cdot ( 1+ \varepsilon))}} \\ 
			&= e^{-\uniformizationRate \cdot t} \cdot e^{\uniformizationRate \cdot t} - e^{-\uniformizationRate \cdot t} \cdot e^{\uniformizationRate\cdot t \cdot (2-e^{\delta} \cdot (1 + \varepsilon))} \\
			&= 1 - e^{-\uniformizationRate \cdot t + 2 \cdot \uniformizationRate \cdot t - \uniformizationRate \cdot t \cdot e^{\delta} \cdot (1 + \varepsilon)} \\
			&= 1 - e^{-\uniformizationRate \cdot t \cdot (e^\delta \cdot(1 + \varepsilon)-1)}.
		\end{align*}
	\end{proof}

	\ThmParetoCurveUnif*
	\begin{proof}
		By \Cref{prop:transient-prob-bounds-epsilon-delta-greater-0}, $\vert \probMeasure_s(\reachability^{\leq t} g) - \probMeasure_{s'}(\reachability^{\leq t} g) \vert \leq 1 - e^{- q \cdot t \cdot (e^{\delta}\cdot ( 1+ \varepsilon) - 1)}$. Hence, if $1 -e^{- q \cdot t \cdot (e^{\delta}\cdot ( 1+ \varepsilon) - 1)} \leq \theta$ it immediately follows that also the absolute difference in the probabilities to reach $g$ after at most time $t$ between $s$ and $s'$ is $\leq \theta$. It is easy to see that 
		\begin{align*}
			1 -e^{- q \cdot t \cdot (e^{\delta}\cdot ( 1+ \varepsilon) - 1)} \leq \theta \qquad \text{ iff } \qquad e^{\delta}\cdot(1+\varepsilon) \leq 1 - \frac{\ln(1 - \theta)}{q \cdot t}.
		\end{align*}
		For fixed $\delta$, this yields 
		\begin{align*}
			e^{\delta} \cdot (1 + \varepsilon) \leq 1 - \frac{\ln(1 - \theta)}{q \cdot t} \qquad \text{ iff } \qquad \varepsilon \leq \frac{1}{e^{\delta}} \cdot \left(\frac{\uniformizationRate \cdot t - \ln(1-\theta)}{q \cdot t} \right) - 1,
		\end{align*}
		while for fixed $\varepsilon$ we get 
		\begin{align*}
			e^{\delta} \cdot (1 + \varepsilon) \leq 1 - \frac{\ln(1 - \theta)}{q \cdot t} \qquad \text{ iff } \qquad \delta \leq \ln \left( \frac{q \cdot t - \ln(1 - \theta)}{(1+\varepsilon) \cdot q \cdot t} \right).
		\end{align*}
		If both inequalities hold simultaneously, i.e., if $\varepsilon \in \left[0, \frac{1}{e^{\delta}} \cdot \left(\frac{\uniformizationRate \cdot t - \ln(1-\theta)}{\uniformizationRate\cdot t} \right) - 1\right]$ and $\delta \in \left[0, \ln\left(\frac{q\cdot t - \ln(1-\theta)}{(\varepsilon + 1)\cdot q \cdot t}\right)\right]$, then $1 -e^{- q \cdot t \cdot (e^{\delta}\cdot ( 1+ \varepsilon) - 1)} \leq \theta$, and by \Cref{prop:transient-prob-bounds-epsilon-delta-greater-0} it follows that $\vert \probMeasure_s(\reachability^{\leq t} g) - \probMeasure_{s'}(\reachability^{\leq t} g) \vert \leq \theta$, too.
	\end{proof}
	
\section{Additional Material for \Cref{sec:errors-0-delta-states}}
	\label{app:uniformizing}

\subsection{Uniformizing $(0,\delta)$-Bisimilar CTMCs}

For transitive $(0,\delta)$-bisimulations, i.e.,  $(0,\delta)$-bisimulations that are equivalence relations, the following construction shows
that we can restrict our analysis to uniform CTMCs when comparing time-bounded reachability probabilities in related states. 

Let $\mathcal{M}=(\stateSpace^{\mathcal{M}}, \prob^{\mathcal{M}}, \rate^{\mathcal{M}}, \initialState^\mathcal{M}, \labelFunction^\mathcal{M})$ and $\mathcal{N}=(\stateSpace^{\mathcal{N}}, \prob^{\mathcal{N}}, \rate^{\mathcal{N}}, \initialState^\mathcal{N}, \labelFunction^\mathcal{N})$ be two CTMCs with goal states $g$ and $g'$, respectively, and $R$
a transitive $(0,\delta)$-bisimulation on $\mathcal{M} \oplus \mathcal{N}$ with $(\initialState^\mathcal{M},\initialState^\mathcal{N})\in R$. Further, assume that $(g,g')\in R$ and that $g$ and $g'$ are not related to any other states by $R$. (This can be achieved, e.g., by giving both states the same, fresh label.) Let $t>0$.
Assume
\[
 \probMeasure^{\mathcal{M}}(\lozenge^{\leq t} g) \leq \probMeasure^{\mathcal{N}}(\lozenge^{\leq t} g').
 \]
 \paragraph{Claim:}
There are uniform CTMCs $\mathcal{M}'=(\stateSpace^{\mathcal{M}}, \prob^{\mathcal{M}'}, \rate^{\mathcal{M}'}, \initialState^\mathcal{M}, \labelFunction^{\mathcal{M}})$ and $\mathcal{N}'=(\stateSpace^{\mathcal{N}}, \prob^{\mathcal{N}'}, \rate^{\mathcal{N}'}, \initialState^\mathcal{N}, \labelFunction^{\mathcal{N}})$ with the same state spaces as $\mathcal{M}$ and $\mathcal{N}$, such that $R$ is still a $(0,\delta)$-bisimulation on $\mathcal{M}' \oplus \mathcal{N}'$ and such that 
\[
\probMeasure^{\mathcal{M}'}(\lozenge^{\leq t} g) \leq  \probMeasure^{\mathcal{M}}(\lozenge^{\leq t} g) \leq \probMeasure^{\mathcal{N}}(\lozenge^{\leq t} g') \leq \probMeasure^{\mathcal{N}'}(\lozenge^{\leq t} g').
\]
Furthermore, for all states $p \in \stateSpace^{\mathcal{M}}$ and $p'\in  \stateSpace^{\mathcal{N}}$ with $(p,p')\in R$, 
we have $\rate^{\mathcal{N}'}(p')= \rate^{\mathcal{M}'}(p) \cdot e^\delta$.

\begin{proof}
W.l.o.g., we assume that all states in $\mathcal{M}$ and $\mathcal{N}$ are reachable.
First, observe that all (reachable) states have to be part of an $R$-equivalence class as $R$ is an exact probabilistic bisimulation on the embedded DTMCs of $\mathcal{M}$ and $\mathcal{N}$, and the initial states $\initialState^\mathcal{M}$ and $\initialState^\mathcal{N}$ are related by $R$.
Furthermore, all states in $\mathcal{M}$ have to be related by $R$ to at least one state in $\mathcal{N}$, and vice versa.

To obtain $\mathcal{M}'$ and $\mathcal{N}'$ we first construct 
$\mathcal{M}''$ and $\mathcal{N}''$  that differ from $\mathcal{M}$ and $\mathcal{N}$ only in the exit rates, given by new functions $\rate^{\mathcal{M}''}$ and $\rate^{\mathcal{N}''}$ on 
$\stateSpace^{\mathcal{M}}$ and $\stateSpace^{\mathcal{N}}$. In particular, the transition probabilities are not affected when moving from $\mathcal{M}$ (resp. $\mathcal{N}$) to $\mathcal{M}''$ (resp. $\mathcal{N}''$).

For the construction of $\mathcal{M}''$ and $\mathcal{N}''$, let $C$ be an $R$-equivalence class.
Define 
\begin{align*}
\rate_{\min}^\mathcal{M}(C) &= \min_{p\in \stateSpace^{\mathcal{M}} \cap C} \rate^{\mathcal{M}}(p) \qquad \text{ and } \qquad 
\rate_{\max}^\mathcal{N}(C) = \max_{p\in \stateSpace^{\mathcal{N}} \cap C} \rate^{\mathcal{N}}(p).
\end{align*}
As all states in $C$ are $(0,\delta)$-bisimilar, we know that $\rate_{\max}^\mathcal{N}(C) \leq  \rate_{\min}^\mathcal{M}(C) \cdot e^\delta$.
Now, let
\[
\rate^{\mathcal{M}''}(p) = \rate_{\min}^\mathcal{M}([p]_R) \leq  \rate^{\mathcal{M}}(p) 
\]
for all $p\in \stateSpace^{\mathcal{M}}$ 
and let 
\[
\rate^{\mathcal{N}''}(p') = \rate_{\min}^\mathcal{M}([p']_R) \cdot e^\delta \geq \rate_{\max}^\mathcal{N}([p']_R) \geq \rate^{\mathcal{N}}(p')
\]
for all $p'\in \stateSpace^{\mathcal{N}}$. 
As rates in $\mathcal{M}''$ are all smaller than or equal to the corresponding ones in $\mathcal{M}$ and rates in $\mathcal{N}''$ are all larger than or equal to the corresponding ones in $\mathcal{N}$, we have
\[
\probMeasure^{\mathcal{M}''}(\lozenge^{\leq t} g) \leq  \probMeasure^{\mathcal{M}}(\lozenge^{\leq t} g) \leq \probMeasure^{\mathcal{N}}(\lozenge^{\leq t} g') \leq \probMeasure^{\mathcal{N}''}(\lozenge^{\leq t} g').
\]
Furthermore, $R$ is still a $(0,\delta)$-bisimulation on $\mathcal{M}''\oplus \mathcal{N}''$ as the transition probabilities have not been changed and the rates of related states differ at most by a factor of $e^\delta$ by construction.

Now, we let $\mathcal{M}'$ be the uniformization of $\mathcal{M}''$ with uniformization rate $\uniformizationRate_{\mathcal{M}}=\max_{p\in \stateSpace^{\mathcal{M}}} \rate^{\mathcal{M}''}(p)$
and likewise $\mathcal{N}'$  the uniformization of $\mathcal{N}''$ with uniformization rate $\uniformizationRate_{\mathcal{N}}=\max_{p'\in \stateSpace^{\mathcal{N}}} \rate^{\mathcal{N}''}(p')$.
We view these uniformizations as CTMCs.
For each state $p \in \stateSpace^{\mathcal{M}}$, there is a state $p'\in \stateSpace^{\mathcal{N}}$, namely any state related to $p$ by $R$, such that 
\[
e^\delta \cdot \rate^{\mathcal{M}''}(p) = \rate^{\mathcal{N}''}(p').
\]
Analogously, for each $p'\in \stateSpace^{\mathcal{N}}$, there is a $p \in  \stateSpace^{\mathcal{M}}$ such that this equation holds.
So, $\uniformizationRate_{\mathcal{M}}\cdot e^\delta = \uniformizationRate_\mathcal{N}$.

Now, let $p \in \stateSpace^{\mathcal{M}}$ and  $p'\in \stateSpace^{\mathcal{N}}$ be such that $(p,p')\in R$.
The rates of $p$ and $p'$ in $\mathcal{M}' \oplus \mathcal{N}'$ differ by a factor of $e^\delta$, just like the uniformization rates.
So, for all states $p\not=r\in \stateSpace^{\mathcal{M}}$, the probability to move from $p$ to $r$ in the uniformization $\mathcal{M}'$
is 
\[
\prob^{\mathcal{M}'}(p,r) = \frac{ \rate^{\mathcal{M}''}(p)}{\uniformizationRate_\mathcal{M}} \cdot \prob^{\mathcal{M}}(p,r) .
\]
Likewise, for all $p',r'\in \stateSpace^{\mathcal{N}}$ with $r' \neq p'$ and $(p,p')\in R$,
\[
\prob^{\mathcal{N}'}(p',r') = \frac{ \rate^{\mathcal{N}''}(p')}{\uniformizationRate_\mathcal{N}} \cdot \prob^{\mathcal{N}}(p',r') =  \frac{ \rate^{\mathcal{M}''}(p)}{\uniformizationRate_\mathcal{M}} \cdot \prob^{\mathcal{N}}(p',r')
\]
 The probability to move from $p$ to $p$ and from $p'$ to $p'$ in the respective uniformizations increases by
\[
\frac{\uniformizationRate_\mathcal{M}-\rate^{\mathcal{M}''}(p)}{\uniformizationRate_\mathcal{M}} = \frac{\uniformizationRate_\mathcal{N}-\rate^{\mathcal{N}''}(p')}{q_\mathcal{N}} .
\]
So, in $p$ and $p'$ the outgoing transition probabilities change by the same factor in the uniformization, and the remaining probability mass is redirected to the state itself.
As this holds for all states, $R$ is still an exact bisimulation on the DTMCs embedded in $\mathcal{M}'$ and $\mathcal{N}'$, and so $R$ is a $(0,\delta)$-bisimulation on the CTMCs $\mathcal{M}'$ and $\mathcal{N}'$ 
because all states in $\mathcal{M}'$ have exit rate $q_{\mathcal{M}}$ and all states in $\mathcal{N}'$ have exit rate $q_{\mathcal{N}}=e^\delta \cdot q_{\mathcal{M}}$.
As uniformization does not affect time-bounded reachability probabilities, we also have
\[
\probMeasure^{\mathcal{M}'}(\lozenge^{\leq t} g) \leq  \probMeasure^{\mathcal{M}}(\lozenge^{\leq t} g) \leq \probMeasure^{\mathcal{N}}(\lozenge^{\leq t} g') \leq \probMeasure^{\mathcal{N}'}(\lozenge^{\leq t} g').
\]\qed
\end{proof}

\subsection{Proofs of \Cref{sec:errors-0-delta-states}}
	
	\begin{lemma}\label{lem:binom-coeff}
		Let $m \in \mathbb{N}_{> 0}$ and $c > 1$. Then 
		\begin{align*}
			\sum_{i=0}^{m-1} \frac{1}{i!(m-i-1)!} \cdot (c-1)^{m-i-1} = \frac{1}{(m-1)!}\cdot c^{m-1}.
		\end{align*}
	\end{lemma}
	\begin{proof}
		We start by observing that 
		\begin{align*}
			\sum_{i=0}^{m-1} \frac{1}{i!(m-i-1)!} \cdot (c-1)^{m-i-1} = \frac{1}{(m-1)!}c^{m-1}, 
		\end{align*}
		is equivalent to 
		\begin{align*}
			\sum_{i=0}^{m-1} \frac{(m-1)!}{i! (m-i-1)!} \cdot (c-1)^{-i} = \left( \frac{c}{c-1} \right)^{m-1}
		\end{align*}
		which, after rewriting, is equivalent to
		\begin{align*}
			\sum_{i=0}^{m-1} \binom{m-1}{i} \cdot (c-1)^{-i} = \left(\frac{c}{c-1} \right)^{m-1}.
		\end{align*}
		
		We prove the latter identity by induction on $m$, where the induction base $m = 1$ is trivial. For the induction step $m \mapsto m+1$, we make use of the well-known fact that 
		\begin{align*}
			\binom{m}{i+1} = \binom{m-1}{i} + \binom{m-1}{i+1}
		\end{align*}
		to obtain 
		\begin{align*}
			&\textcolor{white}{=} \sum_{i=0}^{m} \binom{m}{i} (c-1)^{-i} \\&= 1 + \sum_{i=1}^m \binom{m}{i} \cdot (c-1)^{-i} \\&= 1 + \sum_{i=0}^{m-1} \binom{m}{i+1} \cdot  (c-1)^{-i-1} \\&= 1 + \frac{1}{c-1} \cdot \sum_{i=0}^{m-1} \binom{m}{i+1} \cdot (c-1)^{-i}  \\
			&= 1 + \frac{1}{c-1} \cdot \sum_{i = 0}^{m-1} \left(\binom{m-1}{i} + \binom{m-1}{i+1} \right) \cdot (c-1)^{-i} \\
			&= 1 + \frac{1}{c-1} \cdot \big(\underbrace{\sum_{i = 0}^{m-1} \binom{m-1}{i} \cdot (c-1)^{-i}}_{\text{(IH): } = \left(\frac{c}{c-1}\right)^{m-1}} + \sum_{i = 0}^{m-1} \binom{m-1}{i+1} \cdot (c-1)^{-i} \big) \\
			&= 1 + \frac{c^{m-1}}{(c-1)^m} + \sum_{i = 0}^{m-1} \binom{m-1}{i+1} \cdot (c-1)^{-(i+1)} \\
			&= 1 + \frac{c^{m-1}}{(c-1)^m} + \sum_{i = 0}^{m-2} \binom{m-1}{i+1} \cdot (c-1)^{-(i+1)} + \underbrace{\binom{m-1}{m} \cdot (c-1)^{-m}}_{= 0}\\
			&= 1 + \frac{c^{m-1}}{(c-1)^m} + \sum_{i=1}^{m-1} \binom{m-1}{i} \cdot (c-1)^{-i} \\
			&= \frac{c^{m-1}}{(c-1)^m} + \underbrace{\sum_{i=0}^{m-1} \binom{m-1}{i} \cdot (c-1)^{-i}}_{\text{(IH): } = \left(\frac{c}{c-1}\right)^{m-1}} \\
			&= \frac{c^{m-1}}{(c-1)^m} + \frac{c^{m-1}}{(c-1)^{m-1}} \\
			&= \frac{c^{m-1} + c^{m-1} \cdot (c-1)}{(c-1)^{m}} \\
			&= \frac{c^m}{(c-1)^m}
		\end{align*}
		and so the claim follows by induction.
	\end{proof}
	
	\PropErrorErlangCTMC*
	\begin{proof}
		We start by observing that if $t = 0$ then the absolute difference in reachability probabilities of the goal state $g$ between $\erlang_n$ and $\erlang_n'$ until time $t$ equals $0$, as in this case either both have $g$ as initial state (if $n = 0$), yielding $\probMeasure^{\erlang_n} (\lozenge^{\leq 0} g) = 1 = \probMeasure^{\erlang_n'}(\lozenge^{\leq 0} g)$, or both have an initial state $s_0 \neq g$, yielding $\probMeasure^{\erlang_n} (\lozenge^{\leq 0} g) = 0 = \probMeasure^{\erlang_n'}(\lozenge^{\leq 0} g)$. Since, independent of $n$, we have 
		\begin{align*}
		\sum_{k=0}^{n-1} \frac{0^k}{k!} \cdot \left(e^{-0} - c^ke^{-c \cdot 0} \right) = \frac{0^0}{0!} \cdot (1 - c^0 \cdot 1) + \sum_{k=1}^{n-1} \frac{0^k}{k!} \cdot \left(e^{-0} - c^ke^{-c \cdot 0} \right) = 0,
		\end{align*}
		 it follows that the claim holds if $t = 0$. 
		
		Moreover, if $c = 1$ then $\erlang_n' = \erlang_n$ and thus
		\begin{align*}
			\vert \probMeasure^{\erlang_n}(\lozenge^{\leq t} g) - \probMeasure^{\erlang_n'}(\lozenge^{\leq t} g) \vert = 0 = \sum_{k=0}^{n-1} \frac{t^k}{k!} \left(e^{-t} - 1^ke^{-1 \cdot t} \right),
		\end{align*}
		so the claim also holds if $c = 1$, independent of $t$ and $n$. 
		
		Now assume that $t > 0$, $c > 1$, and denote, for a given $n$, the states of $\erlang_n$ by $s_0, \ldots, s_{n-1}, g$. As Erlang CTMC are, by definition, acyclic it is possible to compute their transient reachability probabilities explicitly by using, e.g., the ACE algorithm proposed in \cite{TAAMC}. We follow a more direct approach and prove the claim by induction on $n$. 
		
		\smallskip
		\noindent
		\textbf{Induction Base:} \qquad For $n = 0$, $\erlang_0$ consists of the goal state $g$ only. Hence, $\vert \probMeasure_{\erlang_0}(\lozenge^{\leq t} g) - \probMeasure_{\erlang'_0}(\lozenge^{\leq t} g) \vert = \vert 1 - 1 \vert = 0$ for any $t \geq 0$. This matches the empty sum $\sum_{k = 0}^{-1} \frac{t^k}{k!} \cdot \left(e^{-t} - c^ke^{-ct} \right)$. 
		
		If $n = 1$, the probability to reach $g$ in $\erlang_1'$ until at most time $t$ is
		\begin{align}
			\probMeasure_{\erlang_1'}(\lozenge^{\leq t} g) = \probMeasure_{\erlang_1}(\lozenge^{\leq t} g) + \probMeasure_{\erlang_1}(\lozenge^{= t} s_0) \cdot \probMeasure_{s_0}(\lozenge^{\leq t\cdot(c-1)} g), \label{eq:error-erlang-ctmc-eq-1}
		\end{align}
		where $\probMeasure_{\erlang_1}(\lozenge^{= t} x)$ denotes the probability to be in state $x$ of $\erlang_1$ after \emph{exactly} time $t$. 
		Intuitively, \Cref{eq:error-erlang-ctmc-eq-1} holds because accelerating $\erlang_1$ by $c$ to obtain $\erlang_1'$ yields a uniform CTMC with the same structure as $\erlang_1$ and modified rates $c$ for which, as described in \Cref{rem:arbitrary-rate}, the probability to reach the goal state $g$ until time point $t$ is the same as reaching $g$ in $\erlang_1$ until time point $c \cdot t$. This corresponds, however, to either reaching $g$ until $t$ in $\erlang_1$, or still being in $s_0$ at $t$ and reaching $g$ in the remaining time, i.e., in $[t, c \cdot t]$. Because we consider time-homogeneous CTMC, the latter probability is the same as that of reaching $g$ from $s_0$ in the interval $[0, c \cdot t - t] = [0, t \cdot(c-1)]$, i.e., the same as $\probMeasure_{s_0}(\lozenge^{\leq t\cdot(c-1)} g)$. 
		
		From \Cref{eq:error-erlang-ctmc-eq-1} it now follows that
		\begin{align*}
			&\textcolor{white}{=} \vert \probMeasure_{\erlang_1}(\lozenge^{\leq t} g) - \probMeasure_{\erlang'_1}(\lozenge^{\leq t} g) \vert \\&= \left \vert \probMeasure_{\erlang_1}(\lozenge^{\leq t} g) - \left( \probMeasure_{\erlang_1}(\lozenge^{\leq t} g) + \probMeasure_{\erlang_1}(\lozenge^{= t} s_0) \cdot \probMeasure_{s_0}(\lozenge^{\leq t\cdot(c-1)} g) \right) \right \vert\\
			&= \probMeasure_{\erlang_1}(\lozenge^{= t} s_0) \cdot \probMeasure_{s_0}(\lozenge^{\leq t\cdot(c-1)} g) \\
			&= (1 - (1-e^{-t}))\cdot (1-e^{-t(c-1)}) \\
			&= e^{-t} \cdot (1-e^{-t \cdot (c-1)}) \\
			&= e^{-t} - e^{-t -t \cdot (c-1)} \\
			&= \frac{t^0}{0!} \cdot \left(e^{-t} - c^0 e^{-c \cdot t} \right) \\
			&= \sum_{k=0}^{0} \frac{t^k}{k!} \left(e^{-t} - c^k e^{-c \cdot t} \right).
		\end{align*}
		
		The function $e^{-t} - e^{-ct}$ is differentiable in $t$ with $\frac{d}{dt} (e^{-t} - e^{-ct}) = -e^{-t} + ce^{-ct}$. By computing the roots via 
		\begin{align*}
			ce^{-ct} - e^{-t} = 0 \qquad \text{ iff } \qquad \ln(c)-ct = -t \qquad \text{ iff } \qquad t = \frac{\ln(c)}{c-1}
		\end{align*}
		and plugging into the second derivative of $e^{-t} - e^{-ct}$ we observe that the function has a local maximum at this position. Furthermore, $e^{-t} - e^{-ct} \geq 0$ for all $t \geq 0$ because $c > 1$. Since, for $t \geq 0$,  $\frac{d}{dt}(e^{-t} - e^{-ct}) > 0$ iff $ t < \frac{\ln(c)}{c-1}$, $\frac{d}{dt}(e^{-t} - e^{-ct}) < 0$ iff $t > \frac{\ln(c)}{c-1}$, $\lim_{t \to 0} e^{-t} - e^{-ct} = 0 = \lim_{t \to \infty} e^{-t} - e^{-ct}$ and $e^{-t} - e^{-ct}$ is continuous, it follows that the maximum is global on the interval $[0, \infty)$.
		
		\smallskip
		\noindent
		\textbf{Induction Step: } \qquad Assume that the claim holds for $n-1$. Using a similar argument as in the induction base, we can find an analogue to \Cref{eq:error-erlang-ctmc-eq-1} for $\erlang_n'$ and express the reachability probability in the accelerated chain by probabilities in the original chain $\erlang_n$ via 
		\begin{align*}
			\probMeasure_{\erlang_n'}(\lozenge^{\leq t}g) = \probMeasure_{\erlang_n}(\lozenge^{\leq t} g) + \sum_{i = 0}^{n-1} \probMeasure_{\erlang_n}(\lozenge^{= t} s_i) \cdot \probMeasure_{s_i}(\lozenge^{\leq t \cdot (c-1)} g).
		\end{align*}
		Intuitively, this identity expresses that reaching $g$ until time $t$ in $\erlang_n'$ is the same as either reaching $g$ in $\erlang_n$ until time $t$, or being in any of the remaining states of $\erlang_n$ at time point $t$ and reaching $g$ in the remaining time of $t \cdot (c-1)$. For the absolute difference of the reachability probabilities in $\erlang_n$ and $\erlang_n'$ this yields 
		\begin{align*}
			\Diff_t(\erlang_n) = \vert \probMeasure_{\erlang_n}(\lozenge^{\leq t} g) - \probMeasure_{\erlang'_n}(\lozenge^{\leq t} g) \vert = \sum_{i = 0}^{n-1} \pi_{t}(s_0, s_i) \cdot \pi_{t \cdot (c-1)}(s_i, g),
		\end{align*}
		where $\pi_x(s, s')$ denotes the transient probability to be in state $s'$ after time $x$ when starting in state $s$. 
		
		Because of the shape of $\erlang_n$, the probabilities $\pi_t(s_0, s_i)$ and $\pi_{t \cdot (c-1)}(s_i, g)$ for $i = 0, \ldots, n-1$ can be computed explicitly by applying the well-known Poisson distribution: $\pi_t(s_0, s_i)$ describes the probability to make exactly $i$ jumps until time $t$ and is thus given as $\pi_t(s_0, s_i) = \frac{1}{i!} \cdot t^i \cdot e^{-t}$, while $\pi_{t \cdot (c-1)} (s_i, g)$ equals the probability to make the remaining number of jumps required to get from $s_i$ to $g$ until time $t \cdot (c-1)$, so
		\begin{align*}
			\pi_{t \cdot(c-1)}(s_i, g) = 1 - \sum_{j=0}^{n-i-1} \frac{1}{j!} \cdot e^{-t\cdot(c-1)} \cdot (t \cdot(c-1))^j.
		\end{align*}
		When plugging this into the sum, we get
		\begin{align}
			&\textcolor{white}{=} \Diff_t(\erlang_n) \nonumber
			\\&= \sum_{i = 0}^{n-1} \frac{t^i}{i!} \cdot e^{-t} \cdot \left(1 - \sum_{j=0}^{n-i-1} \frac{1}{j!} \cdot e^{-t\cdot(c-1)} \cdot (t \cdot(c-1))^j\right) \label{eq:erlang-bound-eq1} \\
			&= \sum_{i=0}^{n-1}\frac{1}{i!} \cdot \left( t^i e^{-t} - \sum_{j=0}^{n-i-1} \frac{1}{j!} \cdot e^{-c t} \cdot t^{j+i} \cdot (c-1)^j \right) \nonumber \\
			&= \sum_{i=0}^{n-2}\frac{1}{i!} \cdot \left( t^i e^{-t} - \sum_{j=0}^{n-i-1} \frac{1}{j!} \cdot e^{-c t} \cdot t^{j+i} \cdot (c-1)^j \right) \nonumber \\& \qquad + \frac{1}{(n-1)!}\cdot \left (t^{n-1}e^{-t} - \sum_{j=0}^0 \frac{1}{j!} e^{-ct} \cdot t^{j + n-1} \cdot (c-1)^j \right) \nonumber \\
			&= \sum_{i=0}^{n-2}\frac{1}{i!} \cdot \left( t^i e^{-t} - \sum_{j=0}^{n-i-1} \frac{1}{j!} \cdot e^{-c t} \cdot t^{j+i} \cdot (c-1)^j \right) + \frac{t^{n-1}}{(n-1)!} \cdot (e^{-t} - e^{-ct}) \nonumber \\
			&= \sum_{i=0}^{n-2}\frac{1}{i!} \cdot \left( t^i e^{-t} - \sum_{j=0}^{n-i-2} \left(\frac{1}{j!} \cdot e^{-c t} \cdot t^{j+i} \cdot (c-1)^j \right) - \frac{1}{(n-i-1)!} \cdot e^{-ct} \cdot t^{n-1} \cdot  (c-1)^{n-i-1} \right) \nonumber  \\& \qquad + \frac{t^{n-1}}{(n-1)!} \cdot (e^{-t} - e^{-ct}) \nonumber  \\
			&= \sum_{i=0}^{n-2}\frac{1}{i!} \cdot \left( t^i e^{-t} - \sum_{j=0}^{n-i-2} \left(\frac{1}{j!} \cdot e^{-c t} \cdot t^{j+i} \cdot (c-1)^j \right) \right) \nonumber \\& \qquad- \sum_{i=0}^{n-2} \frac{1}{i!(n-i-1)!} \cdot e^{-ct} \cdot t^{n-1} \cdot (c-1)^{n-i-1} + \frac{t^{n-1}}{(n-1)!} \cdot (e^{-t} - e^{-ct}) \nonumber \\
			&= \sum_{i=0}^{n-2} \frac{t^i}{i!} \cdot e^{-t} \cdot \left(1 - \sum_{j=0}^{n-i-2} \left( \frac{1}{j!} \cdot e^{-t \cdot (c-1)} \cdot (t \cdot (c-1))^j \right) \right) + \frac{t^{n-1}}{(n-1)!} \cdot e^{-t} \nonumber \\ & \qquad - t^{n-1} \cdot \sum_{i=0}^{n-1} \frac{1}{i!(n-i-1)!}\cdot e^{-ct} \cdot  (c-1)^{n-i-1}. \nonumber
		\end{align}
		When comparing the first sum of the last term with the equality 
		\begin{align*}
			\Diff_t(\erlang_n) = \sum_{i = 0}^{n-1} \frac{t^i}{i!} \cdot e^{-t} \cdot \left(1 - \sum_{j=0}^{n-i-1} \frac{1}{j!} \cdot e^{-t\cdot(c-1)} \cdot (t \cdot(c-1))^j\right)
		\end{align*}
		from \Cref{eq:erlang-bound-eq1} it is easy to see that they almost coincide, with the only difference being that $n$ in the latter is replaced with $n-1$ in the former. Thus, 
		\begin{align*}
			\sum_{i=0}^{n-2} \frac{t^i}{i!} \cdot e^{-t} \cdot \left(1 - \sum_{j=0}^{n-i-2} \left( \frac{1}{j!} \cdot e^{-t \cdot (c-1)} \cdot (t \cdot (c-1))^j \right) \right) = \Diff_t(\erlang_{n-1})
		\end{align*}
		and so it follows by  induction hypothesis that 
		\begin{align*}
			\sum_{i=0}^{n-2} \frac{t^i}{i!} \cdot e^{-t} \cdot \left(1 - \sum_{j=0}^{n-i-2} \left( \frac{1}{j!} \cdot e^{-t \cdot (c-1)} \cdot (t \cdot (c-1))^j \right) \right) = \sum_{k=0}^{n-2} \frac{t^k}{k!} \cdot \left(e^{-t} - c^ke^{-ct} \right),
		\end{align*}
		yielding
		\begin{align*}
			\Diff_t(\erlang_n) &= \sum_{k=0}^{n-2} \frac{t^k}{k!} \cdot \left(e^{-t} - c^ke^{-ct} \right) + \frac{t^{n-1}}{(n-1)!} \cdot e^{-t} \\& \qquad - t^{n-1} \cdot \sum_{i=0}^{n-1} \frac{1}{i!(n-i-1)!}\cdot e^{-ct} \cdot (c-1)^{n-i-1} \\
			&= \sum_{k=0}^{n-2} \frac{t^k}{k!} \cdot \left(e^{-t} - c^ke^{-ct} \right) + \frac{t^{n-1}}{(n-1)!} \cdot e^{-t} \\& \qquad - t^{n-1} \cdot e^{-ct} \cdot \underbrace{\sum_{i=0}^{n-1} \frac{1}{i!(n-i-1)!} \cdot (c-1)^{n-i-1}}_{= \frac{1}{(n-1)!}\cdot c^{n-1} \text{ by \Cref{lem:binom-coeff}}} \\
			&= \sum_{k=0}^{n-2} \frac{t^k}{k!} \cdot \left(e^{-t} - c^ke^{-ct} \right) + \frac{t^{n-1}}{(n-1)!} \cdot \left( e^{-t} - c^{n-1}e^{-ct} \right) \\ 
			&= \sum_{k=0}^{n-1} \frac{t^k}{k!} \cdot \left(e^{-t} - c^ke^{-ct} \right)
		\end{align*}
		and completing the induction step. 
		
		As a finite sum of differentiable function we can differentiate $\Diff_t(\erlang_n)$ w.r.t. $t$ to obtain 
		\begin{align*}
			&\textcolor{white}{=}\frac{d}{dt} \Diff_t(\erlang_n) \\&= \sum_{k=1}^{n-1} \left(\frac{t^{k-1}}{(k-1)!} \cdot (e^{-t} - c^ke^{-ct}) + \frac{t^k}{k!}\cdot(-e^{-t} + c^{k+1}e^{-ct}) \right) + (ce^{-ct} - e^{-t})\\
			&= \sum_{k=1}^{n-1} \frac{t^{k-1}}{(k-1)!} \cdot (e^{-t} - c^ke^{-ct}) + \sum_{k=1}^{n-1} \frac{t^k}{k!}\cdot(c^{k+1}e^{-ct} - e^{-t}) + (ce^{-ct} - e^{-t}) \\
			&= \sum_{k=1}^{n-1} \frac{t^{k-1}}{(k-1)!} \cdot (e^{-t} - c^ke^{-ct}) + \sum_{k=0}^{n-1} \frac{t^k}{k!}\cdot(c^{k+1}e^{-ct} - e^{-t}) \\ 
			&= \sum_{k=1}^{n-1} \frac{t^{k-1}}{(k-1)!} \cdot (e^{-t} - c^ke^{-ct}) + \sum_{k=1}^{n} \frac{t^{k-1}}{(k-1)!}\cdot(c^{k}e^{-ct} - e^{-t}) \\
			&= \frac{t^{n-1}}{(n-1)!}(c^ne^{-ct} - e^{-t}).
		\end{align*}
		Now it is easy to see that $\frac{t^{n-1}}{(n-1)!} \cdot (c^ne^{-ct} - e^{-t}) = 0$ iff $t = 0$ or $c^ne^{-ct} - e^{-t}$. The former case if of no concern since $\Diff_{t}(\erlang_n) \geq 0 = \Diff_{0}(\erlang_n)$ for all $t \geq 0$. In the latter case we get
		\begin{align*}
			\frac{d}{dt} \Diff_t(\erlang_n) = 0 \qquad \text{ iff } \qquad c^ne^{-ct} - e^{-t} = 0 \qquad \text{ iff } \qquad t = \frac{n \cdot \ln(c)}{(c-1)}.
		\end{align*} 
		By computing the second derivative of $\Diff_t(\erlang_n)$ w.r.t. $t$ and plugging in $t^* = \frac{n \cdot \ln(c)}{(c-1)}$ one can observe that $\Diff_{t^*}(\erlang_n) < 0$, yielding a local maximum in $t^*$. Similar to the induction base, on the interval $[0, \infty)$ the function $\Diff_t(\erlang_n)$ is continuous and nonnegative, has a first derivative $> 0$ iff $t < t^*$ and $<0$ iff $t > t^*$, and converges to $0$ for $t \to 0$ and $t \to \infty$. All in all, this implies that the obtained maximum in $t^*$ is global on $[0, \infty)$. 
	\end{proof}

	\ThmExactBoundZeroDelta*
	\begin{proof}
		If $\initialState = g$ then $\Diff_t(\mathcal{M}) = \vert \probMeasure^{\mathcal{M}'}(\lozenge^{\leq t} g) - \probMeasure^{\mathcal{M}}(\lozenge^{\leq t}g) \vert = 0$ for any $t$. Otherwise we start by noting that, by the choice of $\mathcal{M}' = c \cdot \mathcal{M}$,
		\begin{align*}
			\Diff_t(\mathcal{M}) = \vert \probMeasure^{\mathcal{M}'}(\lozenge^{\leq t} g) - \probMeasure^{\mathcal{M}}(\lozenge^{\leq t} g) \vert = \probMeasure^\mathcal{M}(\lozenge^{(t, c \cdot t]} g), 
		\end{align*}
		where $\probMeasure^\mathcal{M}(\lozenge^{(t, c \cdot t]} g)$ denote the probability of the chain to \emph{enter}  $g$ in the interval $(t, c \cdot t]$. We introduce some notation used throughout the proof. Denote by 
		\begin{align*}
			\Pi_n = \{s_0 s_1 \ldots s_{n-1} s_n \in \stateSpace^{n} \mid s_0 = \initialState, s_n = g \text{ and} s_i \neq g \text{ for } 0 \leq i \leq n-1\}
		\end{align*}	
		the set of all finite (untimed) paths of length $n$ (i.e., with a total of $n$ transitions) from the initial state of $\mathcal{M}$ to the goal state $g$. For $\pi =  \initialState s_1\ldots s_{n-1} g \in \Pi_n$, we additionally define 
		\begin{align*}
			\mathrm{Traj}(\pi) &= \{\initialState t_0 s_1 t_1 \ldots s_{n-1}t_{n-1} g \in S \cdot (\mathbb{R}_{>0} \cdot S)^n\} \\
			\mathrm{Traj}^*(\pi) &= \{\overline{\pi} \in \mathrm{Traj}(\pi) \mid \overline{\pi} \text{ reaches} g \text{ in} (t, c\cdot t]\}.
		\end{align*}
		In words, $\mathrm{Traj}(\pi)$ is the set of all timed paths that follow $\pi$, and $\mathrm{Traj}^*(\pi)$ is the set of all timed paths of $\mathcal{M}$ that follow $\pi$ and enter $g$ in the time interval $(t, c \cdot t]$. 
		Lastly, define the probability of $\mathcal{M}$ to follow a (timed version of a) path in $\Pi_n$ as
		\begin{align*}
			p_n = \probMeasure^{\mathcal{M}}(\Pi_n) = \sum_{\pi \in \Pi_n} \probMeasure^{\mathcal{M}}(\pi).
		\end{align*}
	
		For two events $A, B$, denote by $\probability(A \mid B)$ the \emph{conditional probability} of $A$ given $B$ w.r.t probability measure $\probability$.
		By Bayes' rule, which states that
		\begin{align*}
			\probability(A \mid B) = \frac{\mathrm{Prob}(B \mid A) \cdot \probability(A)}{\probability(B)}
		\end{align*} 
		if $\probability(B) > 0$, it follows that
		\begin{align*}
			\Diff_t(\mathcal{M}) &=  \probMeasure^\mathcal{M}(\lozenge^{(t, c \cdot t]} g) \\
			&= \sum_{n = 0}^{\infty} \sum_{\pi \in \Pi_n} \probMeasure^\mathcal{M}(\mathrm{Traj}^*(\pi)) \\
			&= \sum_{n=0}^\infty \sum_{\pi \in \Pi_n} \frac{\probMeasure^\mathcal{M}(\mathrm{Traj}^*(\pi) \mid \pi) \cdot \probMeasure^\mathcal{M}(\pi)}{\probMeasure^\mathcal{M}(\pi \mid \mathrm{Traj}^*(\pi))} \\
			&= \sum_{n = 0}^{\infty} \sum_{\pi \in \Pi_n} \probMeasure^\mathcal{M}(\mathrm{Traj}^*(\pi) \mid \pi) \cdot \probMeasure^\mathcal{M}(\pi)
		\end{align*}
		where the last equality holds since the probability to follow (an untimed) path $\pi$ when knowing that one follows a timed version of $\pi$, i.e., the probability of $\pi$ in $\mathcal{M}$ given $\mathrm{Traj}^*(\pi)$, equals $1$.
		
		We observe that for any $\pi \in \Pi_n$, $ \probMeasure^\mathcal{M}(\mathrm{Traj}^*(\pi) \mid \pi)$ equals the error obtained at time point $t$ when considering, instead of $\mathcal{M}$, the Erlang-CTMCs $\erlang_n$ and $c \cdot \erlang_n$, respectively, i.e., that $ \probMeasure^\mathcal{M}(\mathrm{Traj}^*(\pi) \mid \pi) = \Diff_t(\erlang_n)$ for $\Diff_t(\erlang_n)$ as in \Cref{prop:error-erlang-ctmc}. Intuitively, this follows from the fact that conditioning on $\pi$ dictates a path through the chain that must be taken, so there are no probabilistic choices involved anymore. Hence, the only possibility in any state $s_i$, $i < n$, along $\pi = s_0 s_1 \ldots$ is to move to the next state $s_{i+1}$ with probability $1$. But this results in precisely the form of the Erlang CTMC $\erlang_n$. Because $\mathrm{Traj}^*(\pi)$ contains only those timed versions of $\pi$ that enter $g$ exactly in the time interval $(t, c\cdot t]$ we can argue as done in the proof of \Cref{prop:error-erlang-ctmc} to obtain \Cref{eq:error-erlang-ctmc-eq-1} and describe the probabilities as the difference of transient probabilities in the corresponding chains $\erlang_n$ and $\erlang_n'$. 
		
		Using this observation, we get the desired identity via
		\begin{align*}
			\Diff_t(\mathcal{M}) &= \sum_{n = 0}^{\infty} \sum_{\pi \in \Pi_n} \probMeasure^\mathcal{M}(\mathrm{Traj}^*(\pi) \mid \pi) \cdot \probMeasure^\mathcal{M}(\pi) \\&= \sum_{n = 0}^{\infty} \Diff_t(\erlang_n) \cdot \sum_{\pi \in \Pi_n} \probMeasure^\mathcal{M}(\pi) \\&= \sum_{n=0}^{\infty} \Diff_t(\erlang_n) \cdot p_n \\&= \sum_{n=1}^{\infty} \Diff_t(\erlang_n) \cdot p_n
		\end{align*}
		where the second to last equality follows from the definition of $p_n$, and the last one holds since $\Diff_t(\erlang_0) = 0$ for all $t$.
	\end{proof}

	\PropImproveBoundZeroDeltaViaMaxErlang*
	\begin{proof}
		From \Cref{thm:exact-computation-error-0-delta} we know that $\Diff_t(\mathcal{M}) = \sum_{n=0}^{\infty} p_n \cdot \Diff_t(\erlang_n)$, where $p_n$ is the probability of $\mathcal{M}$ to enter the goal state $g$ after \emph{exactly} $n$ discrete steps and $\Diff_t(\erlang_n)$ is as in \Cref{prop:error-erlang-ctmc}. In particular, it holds for all $n \in \mathbb{N}$ that $p_n \in [0,1]$, and additionally that $\sum_{n \in \mathbb{N}} p_n \leq 1$ since the total probability to enter $g$ is at most $1$. Thus, we can bound the value of $\Diff_t(\mathcal{M})$ for above via 
		\begin{align}\label{thm:exact-computation-error-0-delta-eq-1}
			\Diff_t(\mathcal{M}) = \sum_{n=0}^{\infty} p_n \cdot \Diff_t(\erlang_n) \leq \max_{n \in \mathbb{N}} \Diff_t(\erlang_n).
		\end{align}
		Since we know from \Cref{prop:error-erlang-ctmc} that $\Diff_t(\erlang_n) = \sum_{k=0}^{n-1} \frac{t^k}{k!} \cdot \left(e^{-t} - c^k e^{-ct} \right)$ for every $t \geq 0$ and all $n \in \mathbb{N}$ it follows that 
		\begin{align*}
			\Diff_t(\mathcal{M}) \leq \max_{n \in \mathbb{N}} \sum_{k=0}^{n-1} \frac{t^k}{k!} \cdot \left(e^{-t} - c^k e^{-ct} \right).
		\end{align*}
		We now show that for every $t \geq 0$ this maximum is attained at $N = \left \lceil \frac{(e^{\delta} - 1) \cdot t}{\delta} \right \rceil$ where, for $x \in \mathbb{R}$,  $\lceil x \rceil$ is defined to be the smallest integer $\geq x$.
		
		If $t = 0$ then $\Diff_t(\mathcal{M}) = 0 = \Diff_t(\erlang_n)$ for all $n$, and so in particular for $N = 0 = \left \lceil \frac{(e^{\delta}-1) \cdot 0}{\delta} \right \rceil$. Thus, the maximum is attained at $N$.
		
		If $t > 0$ we directly obtain from \Cref{prop:error-erlang-ctmc} that
		\begin{align*}
			\Diff_t(\erlang_{n+1}) = \Diff_t(\erlang_n) + \frac{t^n}{n!}\cdot(e^{-t} - c^ne^{-ct})
		\end{align*}
		from which, as $c = e^{\delta} > 1$, we get
		\begin{align}\label{thm:exact-computation-error-0-delta-eq-2}
			\Diff_t(\erlang_{n+1}) < \Diff_t(\erlang_n) \qquad \text{iff} \qquad e^{(c-1)t} < c^n \qquad \text{iff} \qquad \frac{(e^{\delta} - 1) \cdot t}{\delta} < n. 
		\end{align}
		Let $N = \left \lceil \frac{(e^{\delta} - 1) \cdot t}{\delta} \right \rceil$, i.e., $N$ is the smallest positive integer such that $N \geq \frac{(e^{\delta} - 1) \cdot t}{\delta}$ or equivalently (by \Cref{thm:exact-computation-error-0-delta-eq-2}) such that $e^{(c-1)\cdot N} \leq c^N$. To prove that the maximum in \Cref{thm:exact-computation-error-0-delta-eq-1} is obtained for $N$, we do a case distinction. 
		\begin{itemize}
			\item If $n > N$ then $e^{(c-1)t} \leq c^N < \ldots < c^n$, where the first inequality follows from the choice of $N$ as the smallest integer $\geq \frac{(e^\delta -1) \cdot t}{\delta}$, which is equivalent to $e^{(c-1)t}\leq c^N$ by \Cref{thm:exact-computation-error-0-delta-eq-2}, and the remaining inequalities follow by monotonicity and the fact that $c > 1$. But then $\Diff_t(\erlang_n) < \Diff_t(\erlang_N)$, as for all $n > N$ we have $\frac{(e^{\delta} -1) \cdot t}{\delta} \leq N < n$ and hence 
			\begin{align*}
				\Diff_t(\erlang_n) < \Diff_t(\erlang_{n-1}) < \ldots < \Diff_t(\erlang_{N+1}) < \Diff_t(\erlang_N).
			\end{align*}
			again by \Cref{thm:exact-computation-error-0-delta-eq-2} (applied iteratively).
			\item Otherwise, if $n < N$, then $e^{(c-1)t} > c^{N-1} > \ldots  > c^n$, where the first inequality follows from the choice of $N$ as the \emph{smallest} positive integer such that $e^{(c-1)t} \leq c^N$. Applying \Cref{thm:exact-computation-error-0-delta-eq-2} iteratively, it follows that 
			\begin{align*}
				\Diff_t(\erlang_n) < \Diff_t(\erlang_{n+1}) < \ldots < \Diff_t(\erlang_{N-1}) < \Diff_t(\erlang_N).
			\end{align*}
		\end{itemize}
		All in all this shows that indeed 
		\begin{align*}
			\Diff_t(\mathcal{M}) \overset{\text{Eq. } (\ref{thm:exact-computation-error-0-delta-eq-1})}{\leq} \max_{n \in \mathbb{N}} \Diff_t(\erlang_n) = \Diff_t(\erlang_N) = \sum_{k=0}^{N-1} \frac{t^k}{k!}\cdot \left(e^{-t} - c^ke^{-ct} \right).
		\end{align*}
	\end{proof}

	\PropExactComputationPnDiag*
	\begin{proof}
		As $\probMatrix$ is diagonalizable there are matrices $\symmetryMatrix, \diagonalMatrix \in \mathbb{C}^{n \times n}$ such that $\diagonalMatrix$ is a diagonal matrix whose diagonal contains, in descending order w.r.t. their absolute value, the eigenvalues $\lambda_1, \lambda_2, \ldots, \lambda_n$ of $\probMatrix$ (repeated according to their multiplicities), and such that $\probMatrix = \symmetryMatrix \diagonalMatrix\symmetryMatrix^{-1}$. Since $\probMatrix$ is stochastic, $\lambda_1 = 1$\cite{MAALA}. 
		
		The probability to move into the goal state $g = s_n$, which we w.l.o.g. associate with the last row of $\probMatrix$, after \emph{exactly} $k+1$ steps when starting in the initial state $\initialState = s_1$, which is associated with the first row of $\prob$, is given as the entry at position $[1,n]$ of $\probMatrix^{k+1} - \probMatrix^{k}$. If $a_{\probMatrix} = 1$, i.e., if $\mathcal{M}$ only has a single accepting state, we get 
		\begin{align*}
			p_{k+1} &= (\probMatrix^{k+1} - \probMatrix^{k})_{1, n} = \mathbf{\dirac}_{s_1} \cdot ((\symmetryMatrix \diagonalMatrix\symmetryMatrix^{-1})^{k+1} - (\symmetryMatrix \diagonalMatrix\symmetryMatrix^{-1})^k) \cdot \mathbf{\dirac}_{s_n}^{\top} \\
			&= \mathbf{\dirac}_{s_1} \cdot \symmetryMatrix(\diagonalMatrix^{k+1} - \diagonalMatrix^k) \symmetryMatrix^{-1} \cdot \mathbf{\dirac}_{s_n}^\top \\
			&= \begin{pmatrix}
				\symmetryMatrix_{1, 1} &\symmetryMatrix_{1, 2} &\ldots &\symmetryMatrix_{1, n}
			\end{pmatrix} \cdot \begin{pmatrix} 0 & 0 & & & \ldots \\
				0 & \lambda_2^{k} \cdot(\lambda_2 - 1) & 0 & & \ldots  \\
				0 & 0 & \lambda_3^{k} \cdot (\lambda_3 -1) & 0 & \ldots  \\
				\vdots & \vdots & & \ddots & \vdots \\
				0 & 0  & \ldots & 0 & \lambda_n^{k} \cdot (\lambda_n - 1)
			\end{pmatrix}\cdot \begin{pmatrix} \symmetryMatrix^{-1}_{1, n} \\ \symmetryMatrix^{-1}_{2, n}\\ \vdots \\ \symmetryMatrix^{-1}_{n, n}  \end{pmatrix} \\
			\\ &= \sum_{j = 2}^{n} \symmetryMatrix_{1, j} \cdot \symmetryMatrix^{-1}_{j, n} \cdot (\lambda_j - 1) \cdot \lambda_j^k.
		\end{align*}
	
		In the case that $a_{\probMatrix} = 2$ we have $\lambda_2 = 1$, and so the entry in the second row of $\diagonalMatrix^{k+1} - \diagonalMatrix^k$ vanishes as well. Hence, the resulting sum starts from $j = 3$, but stays the same otherwise. Combining both results we obtain the desired identity.
	\end{proof}
	
	\BoundPnDiag*
	\begin{proof}
		By \Cref{prop:exact-computation-pn-diag} we know that, if $\probMatrix$ is diagonalizable, i.e., if $\probMatrix = \symmetryMatrix \diagonalMatrix \symmetryMatrix^{-1}$ for a diagonal matrix $\diagonalMatrix$ and a suitable symmetry matrix $\symmetryMatrix$, 
		we have for every $k \in \mathbb{N}$ that $p_{k+1} = \sum_{j = a_{\probMatrix} + 1}^{n} \symmetryMatrix_{1, j} \cdot \symmetryMatrix^{-1}_{j, n} \cdot (\lambda_j - 1) \cdot \lambda_j^k$. 
		
		Let $\lambda_1', \ldots, \lambda_m'$ be the distinct eigenvalue of $\probMatrix$, given in descending order w.r.t. their absolute values, i.e., such that $1 = \lambda_1' > \vert \lambda_2' \vert > \ldots > \vert \lambda_m' \vert$. 
		
		We first assume that the only absorbing state of $\mathcal{M}$ is $g$, i.e., that $a_{\probMatrix} = 1$. Then the eigenvalues of $\mathcal{M}$, repeated according to their multiplicities, are $\lambda_1 = 1 > \vert \lambda_2 \vert \geq \ldots \geq \vert \lambda_n \vert$, and so $\lambda = \lambda_2$. 
		As $p_{k+1} \geq 0$ for every $k \in \mathbb{N}$ we have $p_{k+1} = \vert p_{k+1}\vert$. Therefore,
		\begin{align*}
			p_{k+1} &= \left \vert \sum_{j=2}^{n} \symmetryMatrix_{1, j} \cdot  \symmetryMatrix^{-1}_{j, n} \cdot (\lambda_j -1) \cdot \lambda_j^k \right \vert \\
			&\leq \sum_{j=2}^{n} \underbrace{\vert \symmetryMatrix_{1,j} \cdot \symmetryMatrix^{-1}_{j,n} \cdot (\lambda_j - 1) \vert}_{\leq \max_{j = 2, \ldots, n} \vert \symmetryMatrix_{1,j} \cdot \symmetryMatrix^{-1}_{j,n} \cdot (\lambda_j - 1) \vert \eqcolon C} \cdot \underbrace{\vert \lambda_j \vert^k}_{\leq \vert \lambda \vert^k} \\
			&\leq \vert \lambda \vert^k \cdot \sum_{j=2}^{n} C = \vert \lambda \vert^k \cdot (n-1) \cdot C.
		\end{align*}
		By plugging this inequality into the result of \Cref{thm:exact-computation-error-0-delta} we obtain
		\begin{align*}
			\Diff_t(\mathcal{M}) &= \sum_{k=1}^{\infty} p_k \cdot \Diff_t(\erlang_k) \\&\leq (n-1) \cdot C \cdot \sum_{k=1}^{\infty} \vert \lambda \vert^{k-1} \cdot \Diff_t(\erlang_k)
			\\&= (n-a_{\probMatrix}) \cdot C \cdot \sum_{k=1}^{\infty} \vert \lambda \vert^{k-1} \cdot \Diff_t(\erlang_k).
		\end{align*}
	
		Using similar calculations it is easy to show that the claim still holds if $\mathcal{M}$ has an additional absorbing fail state, i.e., if $a_{\probMatrix} = 2$.
	\end{proof}
		
	\ThmExactValuePnNonDiagonalizable* 
	\begin{proof}
		The proof proceeds mostly similar to that of \Cref{prop:exact-computation-pn-diag}, with the main difference being that here the JCF of $\probMatrix = \symmetryMatrix \jordanMatrix\symmetryMatrix^{-1}$ is used instead of the (probably not existing) diagonalization. Let $\tau_1, \ldots, \tau_m$ be the distinct eigenvalues of $\probMatrix$, given in descending order w.r.t. their absolute values. $\jordanMatrix$ can contain multiple blocks $\jordanMatrix_{i,1}, \ldots, \jordanMatrix_{i,q}$ that correspond to a single eigenvalue $\tau_i$. To avoid clutter in the notation we do not use double indices, but instead count the Jordan blocks individually, i.e., instead of $\jordanMatrix_{i,1}, \ldots, \jordanMatrix_{i, q}$ we write $\jordanMatrix_{i}, \jordanMatrix_{i+1}, \ldots, \jordanMatrix_{i +q}$. Let $r_{i}$ be the size (i.e., the number of rows resp. columns) of $\jordanMatrix_{i}$, and let $\lambda_i$ be the eigenvalue corresponding to $\jordanMatrix_i$.
		
		We assume w.l.o.g. that the initial state $\initialState$ of $\mathcal{M}$ corresponds to the first row of $\probMatrix$, and that the unique goal state $g$ corresponds to the last, i.e., the $n$-th row of $\probMatrix$. Furthermore, we assume that the Jordan blocks of $\jordanMatrix$ are ordered decreasingly w.r.t. the absolute values of the corresponding eigenvalues, starting with the blocks corresponding to eigenvalue $1$ in the top left of $\jordanMatrix$, and that if there are multiple Jordan blocks for the same eigenvalue, these occur in decreasing order w.r.t their sizes. 
		
		We first consider the case where $g$ is the only absorbing state in $\mathcal{M}$. 
		Since $g$ is reachable from every state in $\mathcal{M}$, the eigenvalue $\lambda_1 = 1$ has (algebraic and geometric) multiplicity $1$, i.e., $a_{\probMatrix} = 1$, and so $\jordanMatrix$ can be written as 
		\begin{align}
			\jordanMatrix = \begin{pmatrix}
				1 & 0  &\ldots & 0 \\
				0 & \jordanMatrix_{2} & \ldots & 0  \\
				\vdots  && \ddots  & \vdots \\
				0 &\ldots  & &\jordanMatrix_{q_{\jordanMatrix}} 
			\end{pmatrix}.\label{eq:jordanForm}
		\end{align}
		where $q_{\jordanMatrix}$ is the total number of Jordan blocks of $\jordanMatrix$. Each $\jordanMatrix_i$ with corresponding eigenvalue $\lambda_i$ is an $r_i \times r_i$ matrix of the form 
		\begin{align*}
			\jordanMatrix_i = \begin{pmatrix}
				\lambda_i & 1 & 0 & 0 & \ldots \\
				0 & \lambda_i & 1 & 0 & \ldots \\
				\vdots & & \ddots & \ddots & \vdots \\
				0 & \ldots & 0& \lambda_i & 1 \\
				0 & \ldots  & & 0 & \lambda_i
			\end{pmatrix}.
		\end{align*}
		The probability to reach $g$ after exactly $N+1$ discrete time steps is
		\begin{align}
			p_{N+1} &= \mathbf{\dirac}_{\initialState} \cdot (\probMatrix^{N+1} - \probMatrix^{N}) \cdot \mathbf{\dirac}_{g}^{\top}\nonumber \\
			&= \mathbf{\dirac}_{\initialState} \cdot \cdot ((\symmetryMatrix \jordanMatrix \symmetryMatrix^{-1})^{N+1} - (\symmetryMatrix \jordanMatrix \symmetryMatrix^{-1})^{N}) \cdot \mathbf{\dirac}_{g}^{\top} \nonumber \\
			&= \mathbf{\dirac}_{\initialState} \cdot \symmetryMatrix \cdot (\jordanMatrix^{N+1} - \jordanMatrix^{N}) \cdot \symmetryMatrix^{-1}  \cdot \mathbf{\dirac}_{g}^{\top} \nonumber \\
			&= \begin{pmatrix} \symmetryMatrix_{1, 1}  & \symmetryMatrix_{1, 2} & \ldots & \symmetryMatrix_{1, n} \end{pmatrix} \cdot (\jordanMatrix^{N+1} - \jordanMatrix^{N}) \cdot  \begin{pmatrix} \symmetryMatrix^{-1}_{1, n}\\ \symmetryMatrix^{-1}_{2, n}\\ \vdots \\ \symmetryMatrix^{-1}_{n, n} \end{pmatrix}. \label{eq:pN+1-written-our-via-Jordan}
		\end{align}
		We first consider the case that $\lambda_i \neq 0$ for all $i$. Then the $N$-th power of $\jordanMatrix_{i}$ has the form
		\begin{align}
			\jordanMatrix_i^N = \begin{pmatrix}
				\lambda_i^N & \binom{N}{1} \lambda_i^{N-1} & \binom{N}{2} \lambda_i^{N-2} & \ldots & \binom{N}{r_i - 1} \lambda_i^{N - r_i + 1} \\
				0 & \lambda_i^N & \binom{N}{1} \lambda_i^{N-1} & \ldots & \binom{N}{r_i - 2} \lambda_i^{N-r_i + 2} \\ \\
				\vdots & &  \ddots & &  \vdots \\ \\
				0 & \ldots & 0 & \lambda_i^N & \binom{N}{1} \lambda_i^{N-1} \\
				0 & \ldots & 0 & 0 & \lambda_i^N
			\end{pmatrix}, \label{eq:power-of-jordan}
		\end{align}
		yielding
		\begin{align*}
			&\jordanMatrix_i^{N+1} - \jordanMatrix_i^N = \jordanMatrix_i^N \cdot (\jordanMatrix_i - \mathbf{I}_{r_i}) \\&= \begin{pmatrix}
				\lambda_i^N ( \lambda_i - 1) & \lambda_i^{N-1}  (\lambda_i  \binom{N+1}{1} - \binom{N}{1}) & \lambda_i^{N-2}  (\lambda_i  \binom{N+1}{2} - \binom{N}{2})  & \ldots & \lambda_i^{N - r_i + 1}  ( \lambda_i  \binom{N+1}{r_i - 1} - \binom{N}{r_i - 1})\\
				0 & \lambda_i^N ( \lambda_i - 1) & \lambda_i^{N-1}  (\lambda_i  \binom{N+1}{1} - \binom{N}{1}) & \ldots & \lambda_i^{N - r_i + 2}  ( \lambda_i  \binom{N+1}{r_i - 2} - \binom{N}{r_i - 2}) \\ \\
				\vdots & &  \ddots & &  \vdots \\ \\
				0 & \ldots & 0 & \lambda_i^N ( \lambda_i - 1) & \lambda_i^{N-1}  (\lambda_i  \binom{N+1}{1} - \binom{N}{1}) \\
				0 & \ldots & 0 & 0 & \lambda_i^N ( \lambda_i - 1)
			\end{pmatrix}
		\end{align*}
		where $\mathbf{I}_{r_i}$ is the $r_i \times r_i$ identity matrix.
		When combining the derived form of $\jordanMatrix_i^{N+1} - \jordanMatrix_i^N$ with \Cref{eq:pN+1-written-our-via-Jordan} and carrying out the multiplication, we obtain 
		\begin{align*}
			p_{N+1} &= \sum_{l = 2}^{q_{\jordanMatrix}} \sum_{j = 1}^{r_{l}} \sum_{k = 1}^{j} \symmetryMatrix_{1, k + h_l} \cdot \symmetryMatrix^{-1}_{j + h_l, n} \cdot \lambda_{l}^{N + k - j} \cdot \left(\lambda_{l} \cdot \binom{N+1}{j-k} - \binom{N}{j-k} \right)
			\\&
			= \sum_{l = a_{\probMatrix} +1}^{q_{\jordanMatrix}} \sum_{j = 1}^{r_{l}} \sum_{k = 1}^{j} \symmetryMatrix_{1, k + h_l} \cdot \symmetryMatrix^{-1}_{j + h_l, n} \cdot \lambda_{l}^{N + k - j} \cdot \left(\lambda_{l} \cdot \binom{N+1}{j-k} - \binom{N}{j-k} \right)
		\end{align*}
		for $h_l = \sum_{i = 1}^{l-1} r_{l}$. The parameter $l$ specifies the Jordan block $\jordanMatrix_{l}$ that we are currently interested in, and the offset $h_l$ ensures that the correct entries of $\symmetryMatrix$ (for the rows) and $\symmetryMatrix^{-1}$ (for the columns) are picked. We do not have to consider $l = 1$ as $\lambda_1 = 1$ and so $\jordanMatrix_1^{N+1} - \jordanMatrix_1^{N}$ equal the (($1 \times 1$)-dimensional) matrix with all entries $0$. 
		
		\medskip
		
		We now consider the case where $\probMatrix$ is allowed to have eigenvalue $0$. We assume that the Jordan blocks $\jordanMatrix_{q_{\jordanMatrix}-z}, \ldots, \jordanMatrix_{q_{\jordanMatrix}-1}, \jordanMatrix_{q_{\jordanMatrix}}$ correspond to the eigenvalue $0$, i.e., that there are a total of $z$ such blocks for some $0 \leq z \leq q_{\jordanMatrix}-1$ (where we have $z < q_{\jordanMatrix}$ as $\lambda_1 = 1$ has a multiplicity of at least $1$). Consider again the identity in \Cref{eq:pN+1-written-our-via-Jordan}, i.e, 
		\begin{align*}
			p_{N+1} = \begin{pmatrix} \symmetryMatrix_{1, 1}  & \symmetryMatrix_{1, 2} & \ldots & \symmetryMatrix_{1, n} \end{pmatrix} \cdot (\jordanMatrix^{N+1} - \jordanMatrix^{N}) \cdot  \begin{pmatrix} \symmetryMatrix^{-1}_{1, n}\\ \symmetryMatrix^{-1}_{2, n}\\ \vdots \\ \symmetryMatrix^{-1}_{n, n} \end{pmatrix}.
		\end{align*}
		When carrying out the multiplications, nothing changes for the first $q_{\jordanMatrix} - z$ Jordan blocks that correspond to a non-zero eigenvalue. However, when considering one of the Jordan blocks $\jordanMatrix_i$ for $i \geq q_{\jordanMatrix}-z$, 
		$\jordanMatrix_i^{N}$ is not of the form described in \Cref{eq:power-of-jordan}, but instead has entries $\jordanMatrix_i^{N}[j,j+N] = 1$ for $1 \leq j \leq r_{i} - N$, and entries $0$ everywhere else. The difference $\jordanMatrix_i^{N+1} - \jordanMatrix_i^{N}$ therefore has entries 
		\begin{align*}
			(\jordanMatrix_i^{N+1} - \jordanMatrix_i^{N})[j,k] &= 
			\begin{cases}
				1, &\text{if } 1 \leq j \leq r_{i} - N - 1, k = j + N + 1 \\
				-1, &\text{if } 1 \leq j \leq r_{i} - N, k = j + N \\
				0, & \text{otherwise}
			\end{cases}.
		\end{align*}
		Thus, for the result of the product of \Cref{eq:power-of-jordan}, these blocks add summands of the form 
		\begin{align*}
			(0 - \symmetryMatrix_{1, h_l + 1} \cdot \symmetryMatrix^{-1}_{h_l + N + 1, n}) + \sum_{j = 1}^{r_{i} - N} (\symmetryMatrix_{1, h_l + h} - \symmetryMatrix_{1, h_l + j + 1}) \cdot \symmetryMatrix^{-1}_{h_l + N + j},
		\end{align*}
		where $h_l$ is, as in the first case, an offset that ensures that the entries of the correct columns resp. rows of $\symmetryMatrix_{1, *}$ resp. $\symmetryMatrix^{-1}_{*, n}$ are chosen. By setting $\symmetryMatrix_{1, h_l + j}^* = \symmetryMatrix_{1, h_l +j}$ if $j > 0$ and  $\symmetryMatrix_{1, h_l + j}^* = 0$ if $j = 0$, summing over all Jordan blocks corresponding to the eigenvalue $0$ yields a term of the form
		\begin{align*}
			R(N, z) \coloneqq \sum_{l = q_{\jordanMatrix}-z+1}^{q_{\jordanMatrix}} \sum_{j=0}^{r_{l} - N - 1} (\symmetryMatrix_{1, h_l + j}^* - \symmetryMatrix_{1, h_l + j + 1}) \cdot \symmetryMatrix^{-1}_{h_l + N + j, n}.
		\end{align*}
 		By combining this with the previous calculations regarding the non-zero eigenvalues we obtain
		\begin{align*}
			p_{N+1} = \sum_{l = a_{\probMatrix} + 1}^{q_{\jordanMatrix}-z} \sum_{j = 1}^{r_{l}} \sum_{k = 1}^{j} \symmetryMatrix_{1, k + h_l} \cdot \symmetryMatrix^{-1}_{j + h_l, n} \cdot \lambda_{l}^{N + k - j} \cdot \left(\lambda_{l} \cdot \binom{N+1}{j-k} - \binom{N}{j-k} \right) + R(N, z). 
		\end{align*}
	
		If $\mathcal{M}$ has a an additional absorbing fail state, then $a_{\probMatrix} = 2$ and a $1 \times 1$ Jordan block for eigenvalue $1$ occurs twice in the top left of \Cref{eq:jordanForm}. However, similar to the above deductions this Jordan block is canceled out in any of the difference $\jordanMatrix^{N+1} - \jordanMatrix^N$, and so the only adjustment needed for this case is an initial value of $l = 3$, which again coincides with $a_{\probMatrix} + 1$. 
	\end{proof}

	\PropJordanBound*
	\begin{proof}
		From \Cref{thm:exact-computation-error-0-delta} we know that
		\begin{align}
			\Diff_t(\mathcal{M}) = \sum_{k = 1}^{R-1} p_k \cdot \Diff_t(\erlang_k) + \sum_{k=R}^{\infty} p_k \cdot \Diff_t(\erlang_k) \label{eq:jordan-bound}
		\end{align}
		for $R$ the maximal size of any Jordan block of $\jordanMatrix$. 
		
		We split the series at position $R$ since for any $k \geq R$ the term $R(k, z)$ from \Cref{thm:exact-value-pn-non-diagonalizable} vanishes, allowing us to ignore it in our subsequent calculations. 
		
		To obtain a bound for $p_{N+1}$  for some $N \geq R-1$ we have to consider its absolute value, as the eigenvalues $\lambda_i$ of $\probMatrix$ can be complex even for stochastic matrices $\probMatrix$. Since $p_{N+1} \in [0,1]$ is a probability it coincides with its absolute value, so
		\begin{align*}
			p_{N+1} &= \left \vert\sum_{l = a_{\probMatrix}+1}^{q_{\jordanMatrix}-z} \sum_{j = 1}^{r_{l}} \sum_{k = 1}^{j} \symmetryMatrix_{1, k + h_l} \cdot \symmetryMatrix^{-1}_{j + h_l, n} \cdot \lambda_{l}^{N + k - j} \cdot \left(\lambda_{l} \cdot \binom{N+1}{j-k} - \binom{N}{j-k} \right) \right \vert
		\end{align*}
		and we want to find a bound for the right-hand side of this equation. Using the triangle inequality we get 
		\begin{align*}
			p_{N+1} &\leq  \sum_{l = a_{\probMatrix}+1}^{q_{\jordanMatrix}-z} \sum_{j = 1}^{r_{l}} \sum_{k = 1}^{j}  \vert \symmetryMatrix_{1, k + h_l} \cdot  \symmetryMatrix^{-1}_{j + h_l, n} \vert \cdot \vert \lambda_{l}^{N + k - j} \vert \cdot \left \vert \left(\lambda_{l} \cdot \binom{N+1}{j-k} - \binom{N}{j-k} \right) \right \vert \\
			&\leq \sum_{l = a_{\probMatrix}+1}^{q_{\jordanMatrix}-z} \sum_{j = 1}^{r_{l}} \sum_{k = 1}^{j}  \vert \symmetryMatrix_{1, k + h_l} \cdot \symmetryMatrix^{-1}_{j + h_l, n}  \vert \cdot \vert \lambda^{N + k - j} \vert \cdot \left \vert \left(\lambda_{l} \cdot \binom{N+1}{j-k} - \binom{N}{j-k} \right) \right \vert \\
			&\leq \vert \lambda \vert^{N - r + 1} \cdot \sum_{l = a_{\probMatrix}+1}^{q_{\jordanMatrix}-z} \sum_{j = 1}^{r_{l}} \sum_{k = 1}^{j}  \vert \symmetryMatrix_{1, k + h_l} \cdot \symmetryMatrix^{-1}_{j + h_l, n} \vert \cdot \left \vert \left(\lambda_{l} \cdot \binom{N+1}{j-k} - \binom{N}{j-k} \right) \right \vert
		\end{align*}
		where the second inequality follows from $\vert \lambda \vert \geq \vert \lambda_i\vert$ for every $a_{\probMatrix} + 1 \leq i \leq q_{\jordanMatrix}-z$, and the last inequality holds since $N + 1 - r$ is the minimal possible exponent of $\lambda$ in any of the summands. 
		
		For $l = a_{\probMatrix}+1, \ldots, q_{\jordanMatrix}-z$ let $\vert \lambda_{l}^* \vert = \max\{\vert \lambda_{l} \vert, \vert 1 - \lambda_{l} \vert\}$. If $j = k = 1$, i.e., if $j - k = 0$, the right-most factor can be bounded via 
		\begin{align*}
			\left \vert \left(\lambda_{l} \cdot \binom{N+1}{j-k} - \binom{N}{j-k} \right) \right \vert &= \vert \lambda_{l} - 1 \vert \leq \lambda_{l}^* = \lambda_{l}^{*} \cdot \frac{(N+1)^{j-k}}{(j-k)!},
		\end{align*}
		while for $j -k > 0$ we obtain 
		\begin{align*}
			&\left \vert \left(\lambda_{l} \cdot \binom{N+1}{j-k} - \binom{N}{j-k} \right) \right \vert \\
			&= \left \vert \left(\lambda_{l} \cdot \binom{N+1}{j-k} - \lambda_{l} \cdot \binom{N}{j-k} + \lambda_{l} \cdot \binom{N}{j-k} - \binom{N}{j-k} \right) \right \vert \\
			&\leq \left \vert \lambda_{l} \cdot \binom{N+1}{j-k} - \lambda_{l} \cdot \binom{N}{j-k} \right \vert + \left \vert \lambda_{l} \cdot \binom{N}{j-k} - \binom{N}{j-k}  \right \vert \\
			&= \vert \lambda_{l} \vert \cdot \left \vert \binom{N+1}{j-k} - \binom{N}{j-k} \right \vert + \binom{N}{j-k} \cdot \vert \lambda_{l} - 1 \vert \\
			&\leq \lambda_{l}^* \cdot \left( \left \vert \binom{N+1}{j-k} - \binom{N}{j-k} \right \vert + \binom{N}{j-k} \right) \\
			&= \lambda_{l}^* \cdot \left( \left \vert \binom{N}{j-k} + \binom{N}{j-k-1} - \binom{N}{j-k} \right \vert + \binom{N}{j-k} \right) \\
			&= \lambda_{l}^* \cdot \left( \binom{N}{j-k-1} + \binom{N}{j-k} \right)  \\
			&=  \lambda_{l}^*  \cdot \binom{N+1}{j-k} \\
			&\leq \vert \lambda_{l}^* \vert \cdot \frac{(N+1)^{j-k}}{(j-k)!}
		\end{align*}
		where the last inequality holds trivially if $N + 1 < j-k$, as in this case we have $\binom{N+1}{j-k} = 0$, and otherwise follows from
		\begin{align*}
			\binom{n}{k} = \frac{n!}{k! \cdot (n-k)!} = \frac{n \cdot {n-1} \cdot \ldots \cdot (n-k+1)}{k!} \leq \frac{n^{k}}{k!}
		\end{align*}
		for every $n \geq k$. As $(j-k)! \geq 1$ for every $j \geq k$, which holds by definition of the corresponding sum, and $j-k \leq r - 1$ we additionally obtain
		\begin{align*}
			\left \vert \left(\lambda_{l} \cdot \binom{N+1}{j-k} - \binom{N}{j-k} \right) \right \vert \leq \vert \lambda_{l}^* \vert \cdot (N+1)^{r-1}. 
		\end{align*}
		All in all, this yields 
		\begin{align*}
			p_{N+1} &\leq \vert \lambda \vert^{N - r + 1} \cdot \sum_{l = a_{\probMatrix}+1}^{q_{\jordanMatrix}-z} \sum_{j = 1}^{r_{l}} \sum_{k = 1}^{j}  \vert \symmetryMatrix_{1, k + h_l}  \cdot \symmetryMatrix^{-1}_{j + h_l, n} \vert \cdot \left \vert \left(\lambda_{l} \cdot \binom{N+1}{j-k} - \binom{N}{j-k} \right) \right \vert \\
			& \leq \vert \lambda \vert^{N - r + 1} \cdot \sum_{l = a_{\probMatrix}+1}^{q_{\jordanMatrix}-z} \sum_{j = 1}^{r_{l}} \sum_{k = 1}^{j}  \vert \symmetryMatrix_{1, k + h_l} \cdot \symmetryMatrix^{-1}_{j + h_l, n}  \vert \cdot \vert \lambda_{l}^* \vert \cdot (N+1)^{r - 1} \\
			&= \vert \lambda \vert^{N - r + 1} \cdot (N+1)^{r - 1} \cdot \underbrace{\sum_{l = a_{\probMatrix}+1}^{q_{\jordanMatrix}-z} \sum_{j = 1}^{r_{l}} \sum_{k = 1}^{j}  \vert \symmetryMatrix_{1, k + h_l} \cdot \symmetryMatrix^{-1}_{j + h_l, n} \vert \cdot \vert \lambda_{l}^* \vert}_{\eqcolon C} \\
			&= \vert \lambda \vert^{N - r + 1} \cdot (N+1)^{r - 1} \cdot C,
		\end{align*}
		where $C$ is a positive constant that is independent of $N$.
		
		In particular, it follows for every $k \geq R$ that
		\begin{align*}
			p_{k} \leq \vert \lambda \vert^{k - r} \cdot k^{r - 1} \cdot C
		\end{align*}
		and plugging this upper bound into \Cref{eq:jordan-bound} finishes the proof.
	\end{proof}

	\section{Additional Proofs of \Cref{sec:reward-bounds}}
	We start with a short recap on the computation of reward-bounded reachability probabilities. We follow \cite{LCPP}[Sec. 3], adjusted to our setting. For $x, y \in \mathbb{R}$, let $x \ominus y = \max\{0, x - y\}$. Further, for a given state $s$ and some reward bound $r \geq 0$, let $K = K_{s, r} = \{x \geq 0 \mid \reward(s) \cdot x \leq r \}$. Lastly, for states $s,s'$, let $R(s,s') = \rate(s) \cdot \prob(s,s')$. 
	
	As we are interested in reward-bounded reachability probabilities, we do not need to consider any constraint on the time up until reaching the goal state $g$ from state $s$. As the reward of $g$ does not affect the reward-bounded probability \emph{until reaching} $g$, we can w.l.o.g. assume $\reward(g) > 0$. In the notation of \cite{LCPP} we get that $\probMeasure_s(\reachability_{\leq r}g)$ for some reward bound $r$ is the least solution of the following set of equations: 
	\begin{align*}
		\probMeasure_s(\reachability_{\leq r}g) = \begin{cases}
			\int_{0}^{\sup K_{s,r}} \sum_{s' \in S} R(s,s') \cdot e^{-\rate(s) \cdot x} \cdot \probMeasure_{s'}(\reachability_{\leq r \ominus \rho(s) \cdot x} g) \ \mathrm{dx}, & \text{if } s \neq g  \\
			1, & \text{if } s = g
		\end{cases}.
	\end{align*}
	
	First, we show that removing self-loops from states with reward $0$ does not change the reward-bounded reachability probabilities of the goal state $g$. Note that in the following lemma we can exclude $\prob(s,s)= 1$, because in this case $s$ is absorbing and the probability to reach $g$ from $s$ equals $0$, independent of the reward-bound. 
	\begin{restatable}{lemma}{LemRemoveSelfLoops}\label{lem:remove-self-loops-reward-bounded-reach}
		Let $\mathcal{M}$ be a CTMC and $s \in \stateSpace$ with $\rho(s) = 0$ and $\prob(s, s) \in (0,1)$. Then $\probMeasure^\mathcal{M}_s(\reachability_{\leq r} g) = \probMeasure^{\mathcal{M}^*}_s(\reachability_{\leq r} g)$, where $\mathcal{M}^*$ is as $\mathcal{M}$, with the only difference being that $\prob^{\mathcal{M}^*}(s,s) = 0$ and $\prob^{\mathcal{M}^*}(s,s') = \frac{\prob^\mathcal{M}(s,s')}{1 - \prob^\mathcal{M}(s,s)}$ for all $s \neq s'$. 
	\end{restatable}
	\begin{proof}
		By our assumption $\rho(g) > 0$ and hence $s \neq g$. Furthermore, as $\rho(s) = 0$ we have $K_{s,r} = [0, \infty)$ for every $r \geq 0$ and $r \ominus \rho(s) \cdot x = r$ for every $x \geq 0$. Thus, the reward-bounded reachability probabilities from $s$ are given as \cite{LCPP}
		\begin{align*}
			\probMeasure_s(\reachability_{\leq r}g) &= \int_{0}^{\infty} \sum_{s' \in S} R(s,s') \cdot e^{- \rate(s) \cdot x} \cdot \probMeasure_{s'}(\reachability_{\leq r} g) \ \mathrm{dx} \\
			&= \int_{0}^{\infty} \sum_{s' \in S \setminus \{s\}} R(s,s') \cdot e^{- \rate(s) \cdot x} \cdot \probMeasure_{s'}(\reachability_{\leq r} g) \ \mathrm{dx} \\
			& \qquad + \int_{0}^{\infty} R(s,s) \cdot e^{- \rate(s) \cdot x} \cdot \probMeasure_{s}(\reachability_{\leq r} g) \ \mathrm{dx}.
		\end{align*}
		Now we can simplify the second integral via 
		\begin{align*}
			\int_{0}^{\infty} R(s,s) \cdot e^{- \rate(s) \cdot x} \cdot \probMeasure_{s}(\reachability_{\leq r} g) \ \mathrm{dx}
			&= R(s,s) \cdot \probMeasure_{s}(\reachability_{\leq r} g)  \cdot \int_{0}^{\infty} e^{- \rate(s) \cdot x} \ \mathrm{dx} \\
			&= R(s,s) \cdot \probMeasure_{s}(\reachability_{\leq r} g) \cdot \frac{1}{\rate(s)} \\
			&= \prob(s,s) \cdot \probMeasure_{s}(\reachability_{\leq r} g).
		\end{align*}
		Therefore, 
		\begin{align*}
			\probMeasure_s(\reachability_{\leq r}g) = \int_{0}^{\infty} \sum_{s' \in S \setminus \{s\}} R(s,s') \cdot e^{- \rate(s) \cdot x} \cdot \probMeasure_{s'}(\reachability_{\leq r} g) \ \mathrm{dx} + \prob(s,s) \cdot \probMeasure_{s}(\reachability_{\leq r} g)
		\end{align*}
		which is equivalent to 
		\begin{align*}
			\probMeasure_{s}(\reachability_{\leq r} g) \cdot (1 - \prob(s,s)) = \int_{0}^{\infty} \sum_{s' \in S \setminus \{s\}} R(s,s') \cdot e^{- \rate(s) \cdot x} \cdot \probMeasure_{s'}(\reachability_{\leq r} g) \ \mathrm{dx}
		\end{align*}
		and hence 
		\begin{align*}
			\probMeasure_{s}(\reachability_{\leq r} g) &= \frac{1}{1-\prob(s,s)} \cdot \int_{0}^{\infty} \sum_{s' \in S \setminus \{s\}} R(s,s') \cdot e^{- \rate(s) \cdot x} \cdot \probMeasure_{s'}(\reachability_{\leq r} g) \ \mathrm{dx} \\
			&= \int_{0}^{\infty} \sum_{s' \in S \setminus \{s\}} \rate(s) \cdot \frac{\prob(s,s')}{1 - \prob(s,s)} \cdot e^{- \rate(s) \cdot x} \cdot \probMeasure_{s'}(\reachability_{\leq r} g) \ \mathrm{dx} \\
			&= \int_{0}^{\infty} \sum_{s' \in S \setminus \{s\}} \rate(s) \cdot \prob^{\mathcal{M}^*}(s,s') \cdot e^{- \rate(s) \cdot x} \cdot \probMeasure_{s'}(\reachability_{\leq r} g) \ \mathrm{dx} \\
			&= \int_{0}^{\infty} \sum_{s' \in S} \rate(s) \cdot \prob^{\mathcal{M}^*}(s,s') \cdot e^{- \rate(s) \cdot x} \cdot \probMeasure_{s'}(\reachability_{\leq r} g) \ \mathrm{dx} \\
			&= \int_{0}^{\infty} \sum_{s' \in S} R^{\mathcal{M}^*}(s,s') \cdot e^{- \rate(s) \cdot x} \cdot \probMeasure_{s'}(\reachability_{\leq r} g) \ \mathrm{dx} \\
			&= \probMeasure^*_s(\reachability_{\leq r} g)
		\end{align*}
		where the second to last equality follows because $\prob^{\mathcal{M}^*}(s,s) = 0$.
	\end{proof}

	Thanks to \Cref{lem:remove-self-loops-reward-bounded-reach} we can assume w.l.o.g. that no state $s$ in $\mathcal{M}$ with reward $\rho(s) = 0$ has a self-loop when computing reward-bounded reachability probabilities.

	Next, we show that it is possible to safely remove states with reward $0$ from a CTMC without changing the reward-bounded reachability probabilities of any state with positive reward. 
	\PropRemoveRewardZeroStates*
	\begin{proof}
		We show the claim by proving that, for any state $t$ in $\mathcal{M}$ and any successor $t'$ of $t$ with $\rho(t) = 0$ the reward-bounded reachability probabilities in $\mathcal{M}$ are the same as in the chain $\mathcal{M}'$ obtained by removing $t'$ from $\stateSpace$ and setting, for every state $t''$ with $\prob(t'', t') > 0$, the probability distribution as $\prob'(t'', t''') =  \prob(t'',t''') + \prob(t'', t') \cdot \prob(t', t''')$ for all $t''' \neq t'$. By iteratively applying this transformation, we can remove all states with reward $0$ from $\mathcal{M}$, obtaining in the end a CTMC $\mathcal{M}_{> 0}$ of the desired form. 
		
		By our assumptions it holds for any state $t' $ with $\rho(t') = 0$ that $t' \neq g$. Furthermore, we can assume w.l.o.g. that $\prob(t', t') = 0$ because of \Cref{lem:remove-self-loops-reward-bounded-reach}. Using the latter, we show that $\prob'(t'', \cdot)$ as defined above is a distribution for every $t'' \in \stateSpace \setminus \{t'\}$. If $\prob(t'', t') = 0$, moving from $\mathcal{M}$ to $\mathcal{M}'$ does not change the transition probabilities of $t''$ at all, so nothing is to show. Otherwise, 
		\begin{align*}
			\prob'(t'', \stateSpace \setminus \{t'\}) &= \sum_{t \in \stateSpace \setminus \{t'\}} \prob'(t'', t) \\
				&= \sum_{t \in \stateSpace \setminus \{t'\}} \prob(t'', t) + \prob(t'', t') \cdot \prob(t', t) \\
					&=  \sum_{t \in \stateSpace \setminus \{t'\}} \prob(t'', t) + \sum_{t \in \stateSpace \setminus \{t'\}} \prob(t'', t') \cdot \prob(t', t) \\
					&= \sum_{t \in \stateSpace \setminus \{t'\}} \prob(t'', t) + \prob(t'', t') \cdot \underbrace{\sum_{t \in \stateSpace \setminus \{t'\}} \prob(t', t)}_{= \sum_{t \in \stateSpace} \prob(t', t) = 1 \text{ as} \prob(t', t') = 0} \\
					&= \sum_{t \in \stateSpace \setminus \{t'\}} \prob(t'', t) + \prob(t'', t') = 1.
		\end{align*}
	
		To show the preservation of reward-bounded reachability probabilities, we first observe that
		\begin{align}
			\probMeasure_{t'}(\reachability_{\leq r} g) &= \int_{0}^{\infty} \sum_{t'' \in S} R(t',t'') \cdot e^{- \rate(t') \cdot x} \cdot \probMeasure_{t''}(\reachability_{\leq r \ominus \rho(t') \cdot x} g) \ \mathrm{dx} \nonumber \\
			&= \int_{0}^{\infty} \sum_{t'' \in S} R(t',t'') \cdot e^{- \rate(t') \cdot x} \cdot \probMeasure_{t''}(\reachability_{\leq r} g) \ \mathrm{dx} \nonumber \\
			&= \sum_{t'' \in \stateSpace} \probMeasure_{t''}(\reachability_{\leq r} g) \cdot R(t', t'') \cdot \int_{0}^{\infty} e^{-\rate(t') \cdot x} \ \mathrm{dx} \nonumber \\
			&= \sum_{t'' \in \stateSpace} \probMeasure_{t''}(\reachability_{\leq r} g) \cdot \prob(t', t'') \cdot \rate(t') \cdot \frac{1}{\rate(t')} \nonumber \\
			&= \sum_{t'' \in \stateSpace} \probMeasure_{t''}(\reachability_{\leq r} g) \cdot \prob(t', t'') \nonumber \\
			&= \sum_{t'' \in \stateSpace \setminus \{t'\}} \probMeasure_{t''}(\reachability_{\leq r} g) \cdot \prob(t', t''). \label{eq:remove-rewards-eq-1}
		\end{align}
		Let now $t \in \stateSpace$ with $\prob(t,t') > 0$. Then 
		\begin{align*}
			&\textcolor{white}{=}\probMeasure^\mathcal{M}_t(\reachability_{\leq r} g) \\&= \int_{0}^{\sup K_{t,r}} \sum_{t''' \in \stateSpace} R(t,t''') \cdot e^{- \rate(t) \cdot x} \cdot \probMeasure_{t'''}(\lozenge_{\leq r\ominus \rho(t) \cdot x} g) \ \mathrm{dx} \\
			&= \int_{0}^{\sup K_{t,r}} \sum_{t''' \in \stateSpace \setminus \{t'\}} R(t,t''') \cdot e^{- \rate(t) \cdot x} \cdot \probMeasure_{t'''}(\lozenge_{\leq r\ominus \rho(t)\cdot x} g) \\
			& \qquad + R(t,t') \cdot e^{- \rate(t) \cdot x} \cdot  \probMeasure_{t'}(\lozenge_{\leq r \ominus \rho(t) \cdot x} g) \ \mathrm{dx} \\
			&= \int_{0}^{\sup K_{t,r}} \sum_{t''' \in \stateSpace \setminus \{t'\}} R(t,t''') \cdot e^{- \rate(t) \cdot x} \cdot \probMeasure_{t'''}(\lozenge_{\leq r\ominus \rho(t)\cdot x} g) \\
			& \qquad + R(t,t') \cdot e^{- \rate(t) \cdot x} \cdot \sum_{t'' \in \stateSpace \setminus \{t'\}} \probMeasure_{t''}(\reachability_{\leq r \ominus \rho(t)\cdot x} g) \cdot \prob(t', t'') \ \mathrm{dx}
			\\ 
			&= \int_{0}^{\sup K_{t,r}} \sum_{t''' \in \stateSpace \setminus \{t'\}} \rate(t) \cdot (\prob(t,t''') + \prob(t,t') \cdot \prob(t', t''')) \cdot e^{- \rate(t) \cdot x} \cdot \probMeasure_{t'''}(\lozenge_{\leq r\ominus \rho(t)\cdot x} g) \ \mathrm{dx} \\
			&= \int_{0}^{\sup K_{t,r}} \sum_{t''' \in \stateSpace \setminus \{t'\}} \rate(t)  \cdot \prob'(t,t''') \cdot e^{- \rate(t) \cdot x} \cdot \probMeasure_{t'''}(\lozenge_{\leq r\ominus \rho(t)\cdot x} g) \ \mathrm{dx} \\
			&= \int_{0}^{\sup K_{t,r}} \sum_{t''' \in \stateSpace \setminus \{t'\}} R^{\mathcal{M}'}(t,t'')  \cdot e^{- \rate(t) \cdot x} \cdot \probMeasure_{t'''}(\lozenge_{\leq r\ominus \rho(t)\cdot x} g) \ \mathrm{dx} \\
			&= \probMeasure^{\mathcal{M}'}_t(\reachability_{\leq r} g),
		\end{align*}
		so removing $t'$ and redistributing the incoming probabilities as described above indeed preserves reward-bounded reachability probabilities. 
	\end{proof}

	\begin{remark}
		Assume that we are given a CTMC $\mathcal{M}$ and a state $s$ with $\rho(s) = 0$, and that we want to compute $\probMeasure_s(\lozenge_{\leq r} g)$. Then, when constructing $\mathcal{M}_{> 0}$ as proposed in (the proof of) \Cref{prop:remove-reward-0-states}, the state $s$ is removed from the chain, making $\probMeasure_s^{\mathcal{M}_{> 0}}$ not well-defined as $s \notin \stateSpace^{\mathcal{M}_{> 0}}$. However, it is obvious that $\probMeasure_s(\lozenge_{\leq r} g) = \sum_{s' \in \Succ(s)} \prob(s,s') \cdot \probMeasure_{s'}(\lozenge_{\leq r} g)$, as $s$ has reward $0$ and since we can safely assume that $\prob(s,s) = 0$ by \Cref{lem:remove-self-loops-reward-bounded-reach}. Hence, we can still construct $\mathcal{M}_{> 0}$, compute the reward bounded reachability probabilities for all $s' \in \Succ(s)$ (potentially applying the above argument multiple times if there is a path starting in $s$ with an initial prefix of states with reward $0$) and taking the weighted sum of these values to obtain the corresponding probability for $s$. 
	\end{remark}
\end{document}